\documentclass[runningheads]{llncs}

\usepackage{amsthm}
\usepackage{amsmath}
\usepackage{amssymb}

\usepackage{graphicx}
\usepackage{wrapfig}

\usepackage{color}
\usepackage{xargs}
\usepackage{xspace}
\usepackage{paralist}
\usepackage{complexity}

\usepackage{subcaption}
\usepackage{hyperref}
\usepackage[capitalise]{cleveref}

\usepackage{fullpage}
\usepackage{todonotes} 
\usepackage{mdframed}  
\usepackage{lineno}

\graphicspath{{fig/}{./fig/}}

\hypersetup{colorlinks=true,citecolor={purple},linkcolor={blue},urlcolor={blue}}

\spnewtheorem{observation}[theorem]{Observation}{\bfseries}{\itshape}

\newtheorem{claimx}{Claim}

\Crefname{observation}{Observation}{Observations}
\Crefname{proposition}{Proposition}{Propositions}
\Crefname{claimx}{Claim}{Claims}
\Crefname{property}{Property}{Properties}
\Crefname{enumi}{Property}{Properties}

\definecolor{realblue}{rgb}{0,0,1}
\definecolor{darkerblue}{rgb}{0.094,0.455,0.804}
\definecolor{darkblue}{rgb}{0.063,0.306,0.545}
\definecolor{red}{rgb}{0,0,0}
\definecolor{green}{rgb}{0,0.588,0.509}
\definecolor{orange}{rgb}{0.903,0.739,0.382}
\definecolor{realred}{rgb}{1,0,0}
\definecolor{darkmagenta}{rgb}{0.545,0,0.545}

\renewcommand{\emph}[1]{{\em \textcolor{realblue}{#1}}\xspace}
\newcommand{\remove}[1]{\xspace}

\newcommand{\FEPR}{\textsc{Fixed Edge-Length Planar Realization}\xspace}
\newcommand{\FEPRshort}{\textsc{FEPR}\xspace}

\begin{document}
	\title
	{
		Testing the Planar Straight-line Realizability of 2-Trees \\with Prescribed Edge Lengths\thanks{This research was partially supported by MIUR Project ``AHeAD'' under PRIN 20174LF3T8, by H2020-MSCA-RISE project 734922 -- ``CONNECT'', and by Roma Tre University Azione 4 Project ``GeoView''. A preliminary version of this paper was presented at the 29\textsuperscript{th} International Symposium on Graph Drawing and Network Visualization (GD2021).}
	}
	
	\author{Carlos Alegr\'{i}a \and Manuel Borrazzo \and Giordano {Da Lozzo} \and \\ Giuseppe {Di Battista} \and Fabrizio Frati \and Maurizio Patrignani}
	
	
	\institute
	{
		Roma Tre University, Rome, Italy\\ \href{mailto:carlos.alegria@uniroma3.it,manuel.borrazzo@uniroma3.it,giordano.dalozzo@uniroma3.it,giuseppe.dibattista@uniroma3.it,fabrizio.frati@uniroma3.it,maurizio.patrignani@uniroma3.it}{name.lastname@uniroma3.it}
	}

	\maketitle
	
	\begin{abstract}
		We study a classic problem introduced thirty years ago by Eades and Wormald.
		Let $G=(V,E,\lambda)$ be a weighted planar graph, where $\lambda: E \rightarrow \mathbb{R}^+$ is a {\em length function}. 
		The {\sc \FEPR} problem (\FEPRshort for short) asks whether there exists a {\em planar straight-line realization} of $G$, i.e., a planar straight-line drawing of $G$ where the Euclidean length of each edge $e \in E$ is $\lambda(e)$.
		
		Cabello, Demaine, and Rote showed that the \FEPRshort problem is \NP-hard, even when $\lambda$ assigns the same value to all the edges and the graph is triconnected.
		Since the existence of large triconnected minors is crucial to the known \NP-hardness proofs, in this paper we investigate the computational complexity of the \FEPRshort problem for weighted $2$-trees, which are $K_4$-minor free.
		We show the \NP-hardness of the problem, even when $\lambda$ assigns to the edges only up to four distinct lengths.
		Conversely, we show that the \FEPRshort problem is linear-time solvable when $\lambda$ assigns to the edges up to two distinct lengths, or when the input has a prescribed embedding.
		Furthermore, we consider the \FEPRshort problem for weighted maximal outerplanar graphs and prove it to be linear-time solvable if their dual tree is a path, and cubic-time solvable if their dual tree is a caterpillar.
		Finally, we prove that the \FEPRshort problem for weighted $2$-trees is slice-wise polynomial in the length of the longest~path.
		
		\keywords
		{
			Graph realizations \and weighted 2-trees \and straight-line drawings \and planarity
		}
	\end{abstract}
	
	
	\section{Introduction}
	\label{se:intro}
	
	The problem of producing drawings of graphs with geometric constraints is a core subject in Graph Drawing (see, e.g., \cite{af-dog,DBLP:conf/gd/AngeliniLBBHPR14,DBLP:journals/talg/AngeliniLBDKRR18,DBLP:conf/gd/BrandesS06,DBLP:conf/gd/ChaplickKLT0ZZ19,dtt-cvrg-92,DBLP:conf/gd/GiacomoDL06,ggm-sle-08,DBLP:conf/gd/LubiwMM18,rt-rpl-86,DBLP:conf/gd/SilveiraSV17,t-cgda-98}).
	In this context, a classic question is the one of testing whether a graph can be drawn planarly so that its edges are straight-line segments of prescribed lengths.
	The study of such a question has connections with several topics in computational geometry~\cite{DBLP:journals/ijcga/CoullardL92,saxe1980embeddability,DBLP:conf/focs/Yemini79}, rigidity theory~\cite{DBLP:conf/dimacs/Connelly90,DBLP:journals/siamcomp/Hendrickson92,DBLP:journals/jct/JacksonJ05}, structural analysis of molecules~\cite{DBLP:journals/jacm/BergerKL99,DBLP:journals/siamjo/Hendrickson95}, and sensor networks~\cite{DBLP:journals/cluster/CapkunHH02,DBLP:conf/mobicom/PriyanthaCB00,940391}.
	Formally, given a weighted planar graph $G=(V,E,\lambda)$, i.e., a planar graph with vertex set $V$ and edge set $E$ equipped with a \emph{length function} $\lambda: E \rightarrow \mathbb{R}^+$, the {\sc \FEPR} problem (\FEPRshort, for short) asks whether there exists a \emph{planar straight-line realization} of $G$, i.e., a planar straight-line drawing of $G$ in the plane where the Euclidean length of each edge $e \in E$ is $\lambda(e)$.
	The \FEPRshort problem was first studied by Eades and Wormald~\cite{DBLP:journals/dam/EadesW90}, who showed its \NP-hardness even for triconnected planar graphs and for biconnected planar graphs with unit lengths. Cabello, Demaine, and Rote strengthened this result by proving \NP-hardness even for triconnected planar graphs with unit lengths~\cite{JGAA-145}. 
	The graphs that admit a planar straight-line realization with unit lengths are also called \emph{matchstick graphs}.
	Abel et al.~\cite{DBLP:conf/compgeom/AbelDDELS16} showed that recognizing matchstick graphs is strongly $\exists\mathbb{R}$-complete, solving an open problem stated by Schaefer~\cite{Schaefer2013}.
	
	
	Since the existence of large triconnected minors is essential in the known hardness proofs of the \FEPRshort problem, we study its computational complexity for weighted $2$-trees, which are the maximal graphs with no K$_4$-minor.  
	A \emph{$2$-tree} is a graph composed of $3$-cycles glued together along edges in a tree-like fashion.
	As an example, \cref{fig:planar} shows a planar and a non-planar straight-line realization of the same weighted $2$-tree.
	%
	The class of $2$-trees has been deeply studied in Graph Drawing (e.g., in~\cite{DBLP:conf/isaac/LozzoDEJ17,DBLP:conf/walcom/GiacomoHL21,DBLP:conf/cccg/GoodrichJ18,DBLP:conf/gd/LenhartLMN13,DBLP:conf/cocoon/RengarajanM95}); such a graph class coincides with the one of the maximal series-parallel graphs. In particular, the edge lengths of planar straight-line drawings of $2$-trees have been studied in~\cite{DBLP:conf/gd/Blazej0L20,DBLP:journals/jocg/BorrazzoF20}.   
	
	First, we show that, in the fixed embedding setting, the \FEPRshort problem can be solved in linear time (\cref{sec:prescribed_embedding}).
	We remark that the \FEPRshort problem is \NP-hard for general weighted planar graphs with fixed embedding~\cite{JGAA-145,DBLP:journals/dam/EadesW90}.
	Second, we show that, in the variable embedding setting, the \FEPRshort problem is \NP-hard when the number of distinct lengths is at least four (\cref{sec:np-hardness}), whereas it is linear-time solvable when the number of distinct lengths is $1$ or $2$ (\cref{sec:few_lengths}).
	Note that, for general weighted planar graphs, the problem is \NP-hard even when all the edges are required to have the same length~\cite{DBLP:journals/dam/EadesW90}. 
	Third, we move our attention to maximal outerplanar graphs (\cref{sec:outerplanar}), which form a notable subclass of $2$-trees.
	We show that the \FEPRshort problem can be solved in linear time for maximal outerpaths, i.e., the maximal outerplanar graphs whose dual tree is a path, and in cubic time for maximal outerpillars, i.e., the maximal outerplanar graphs whose dual tree is a caterpillar. 
	Finally, we present a slice-wise polynomial algorithm for weighted $2$-trees, parameterized by the length of the longest path (\cref{sec:longest_path}).
	\cref{sec:preliminaries} contains preliminaries and \cref{sec:conclusions} provides conclusions and open problems.
	
	Similarly to~\cite{JGAA-145}, in our algorithms, we adopt the real RAM model of computation, which is customary in computational geometry and allows us to perform standard arithmetic operations in constant time. Furthermore, we show NP-hardness in the Turing machine model by exploiting lengths whose encoding has constant size.
	
	\begin{figure}[tb!]
		\centering
		\includegraphics[width=0.48\textwidth,page=1]{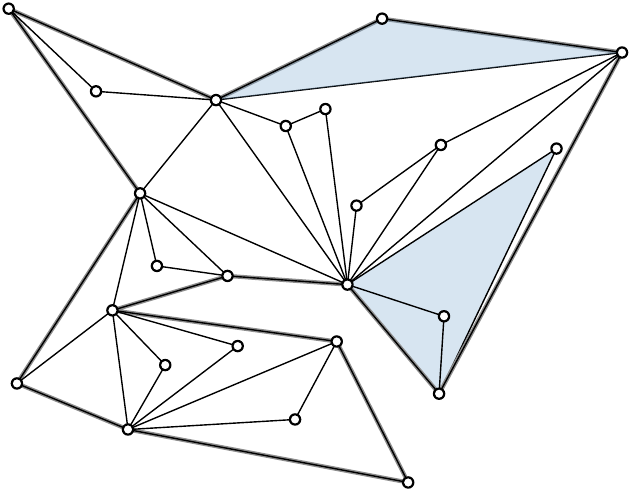}
		\hfil
		\includegraphics[width=0.48\textwidth,page=2]{2-tree_realizations}
		\caption{A planar and a non-planar straight-line realization of the same weighted $2$-tree.}
		\label{fig:planar}
	\end{figure}
	
	
	\section{Preliminaries}
	\label{sec:preliminaries}
	
	Throughout the paper, we assume that graphs are connected and simple (i.e., with no loops or parallel edges).
	Moreover, whenever we refer to the length of a segment or to the distance of two geometric objects, we always assume that these are measured according to the Euclidean metric.
	
	\paragraph{Drawings and embeddings.}
	
	A \emph{drawing} $\Gamma$ of a graph is a mapping of each vertex $v$ to a distinct point $\Gamma(v)$ in the plane and of each edge $(u,v)$ to a Jordan arc $\Gamma(u,v)$ with endpoints $\Gamma(u)$ and~$\Gamma(v)$.
	We often use the same notation for a vertex $v$ and the point $\Gamma(v)$. 
	A drawing is \emph{planar} if it contains no two crossing edges and it is \emph{straight-line} if each curve representing an edge is a straight-line segment.
	A planar drawing partitions the plane into connected regions, called \emph{faces}.
	The bounded faces are the \emph{internal faces}, while the unbounded face is the \emph{outer face}.
	The \emph{boundary} of a face is the circular list of vertices and edges encountered when traversing the geometric border of the face.
	A graph $G$ is \emph{planar} if it admits a planar drawing.
	A planar drawing of $G$ defines a clockwise order of the edges incident to each vertex of $G$; the set of such orders for all the vertices is a \emph{rotation system} for $G$.
	Two planar drawings of $G$ are \emph{equivalent} if (i) they define the same rotation system for $G$ and (ii) their outer faces have the same boundaries. An equivalence class of planar drawings is a \emph{plane embedding} (or simply an \emph{embedding}).
	When referring to a planar drawing $\Gamma$ of a graph that has a prescribed embedding $\mathcal E$, we always imply that $\Gamma$ is in the equivalence class $\mathcal E$; when we want to stress this fact, we say that $\Gamma$ \emph{respects} $\mathcal E$.
	All the planar drawings respecting the same embedding have the same set of face boundaries.
	Hence, we can talk about the face boundaries of a graph with a prescribed embedding.
	With a little overload of terminology, we sometimes write \emph{face} of a graph with a prescribed embedding, while referring to the boundary of the face. 
	A graph is \emph{biconnected} if the removal of any vertex leaves the graph connected.
	In any planar drawing of a biconnected graph every face is bounded by a~simple~cycle.
	Let $\Gamma$ be a drawing of a graph $G$ and let $G'$ be a subgraph of $G$; the \emph{restriction} of $\Gamma$ to $G'$ is the drawing of $G'$ obtained by removing from $\Gamma$ the vertices and edges of $G$ that are not in $G'$. 
	
	\paragraph{2-trees and maximal outerplanar graphs.}
	
	A \emph{$2$-tree} is recursively defined as follows.
	A cycle formed by $k$ edges is a \emph{$k$-cycle}.
	A $3$-cycle is a $2$-tree. Given a $2$-tree $G$ containing the edge $(u,w)$, the graph obtained by adding to $G$ a vertex $v$ and two edges $(v,u)$ and $(v,w)$ is a $2$-tree.
	We observe that any $2$-tree $G$ satisfies the following properties:
	\begin{enumerate}[(P1)]
		\item \label{pr:biconnected} $G$ is biconnected;
		\item\label{pr:degree-2-adjacent} the two neighbors of any degree-$2$ vertex of $G$ are adjacent; and
		\item \label{pr:non-adjacent-degree-2} if $|V(G)|\geq 4$, then $G$ contains two non-adjacent degree-$2$ vertices.
	\end{enumerate}
	
	A $2$-tree $G$ is essentially composed of $3$-cycles glued together along edges in a tree-like fashion.
	We encode this tree-like structure in a tree we call the \emph{decomposition tree} of $G$.
	This tree has a node for each $3$-cycle of $G$, and an edge between two nodes if the corresponding $3$-cycles have a common edge; see \cref{fig:decomposition_tree}.
	Let $T$ be the decomposition tree of $G$, and $n$ be the number of nodes of $G$.
	We observe that $T$ is unique and contains $n-2$ nodes.
	Typically, $T$ is \emph{rooted at} a  $3$-cycle $c$ of $G$.
	We can compute $T$ rooted at $c$ by a recursive construction in which, together with $T$, we produce an auxiliary labeling of the edges of $G$.
	An edge $e$ is labeled with the highest node in $T$ whose corresponding $3$-cycle contains~$e$.
	The construction is as follows.
	If~$G$ coincides with $c$, then $T$ is a unique root node representing~$c$, and each edge of~$c$ is labeled with~$c$ itself.
	Otherwise, by Property~(P3), there is at least one degree-two vertex not belonging to~$c$.
	Let $v$ be any such a vertex, let $u$ and $w$ be the vertices adjacent to $v$, and let $c^{*}$ be the $3$-cycle with vertices $u$, $v$, and $w$.
	Let $T^{*}$ denote the decomposition tree of $G - v$ rooted at $c$.
	We obtain $T$ by adding to $T^{*}$ a leaf representing $c^{*}$, whose parent is the node corresponding to the label of the edge $(u, w)$.
	Further, we set to $c^{*}$ the labels of the edges $(v,u)$ and $(v,w)$.
	It is not hard to see that this procedure constructs $T$ in $O(n)$ time.
	
	\begin{figure}[tb!]
		\centering
		
		\includegraphics[page=1,scale=1]{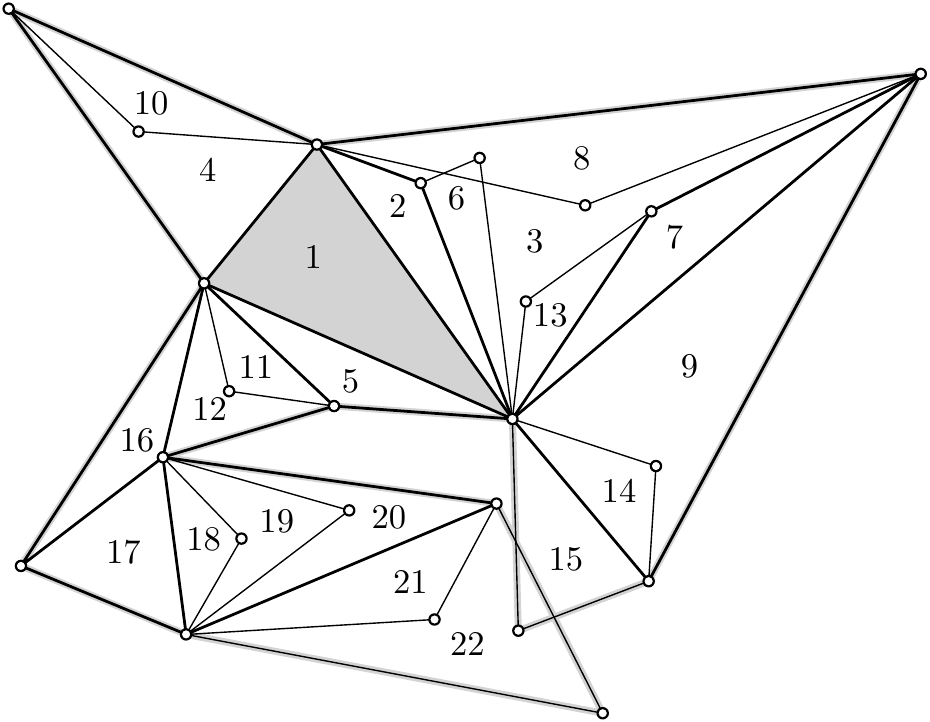}
		\qquad%
		\includegraphics[page=2,scale=1]{2-tree_and_decomposition_tree}
		
		\caption
		{
			On the left, a $2$-tree $G$.
			On the right, the decomposition tree of $G$ rooted at $\triangle_1$.
			Internal nodes are drawn with thick lines and leaf nodes of are drawn with thin lines.
		}
		\label{fig:decomposition_tree}
	\end{figure}
	
	%
	%
	
	An \emph{outerplanar drawing} is a planar drawing in which all the vertices are incident to the outer face.
	A graph is \emph{outerplanar} if it admits an outerplanar drawing.
	An \emph{outerplane embedding} is an equivalence class of outerplanar drawings of a graph; note that a biconnected outerplanar graph has a unique outerplane embedding~\cite{mw-opbe-90,s-cog-79}.
	An outerplanar graph is \emph{maximal} if no edge can be added to it without losing outerplanarity.
	In the unique outerplane embedding of a maximal outerplanar graph, every internal face is delimited by a $3$-cycle.
	The \emph{dual tree} $T$ of a biconnected outerplanar graph $G$ is
	defined as follows.
	Consider the (unique) outerplane embedding $\mathcal O$ of $G$.
	Then $T$ has a node for each internal face of $\mathcal O$ and has an edge between two nodes if the corresponding faces of $\mathcal O$ are incident to the same edge of $G$.
	An \emph{outerpath} is a biconnected outerplanar graph whose dual tree is
	a path.
	A \emph{caterpillar} is a tree that becomes a path if its leaves are removed.
	Such path is called the \emph{spine} of the caterpillar.
	An \emph{outerpillar} is a biconnected outerplanar graph whose dual tree is a caterpillar.
	
	\paragraph{Straight-line realizations.}
	
	A \emph{weighted graph} $G=(V,E,\lambda)$ is a graph equipped with a length function $\lambda$ that assigns a positive real number to each edge.
	A \emph{straight-line realization} of $G$ is a straight-line drawing of $G$ in which each edge $e$ is drawn as a line segment of length $\lambda(e)$.
	If the drawing is planar, we say that the realization is planar.
	The {\sc \FEPR (\FEPRshort)} problem receives as input a weighted planar graph $G$, and asks if there exists a planar straight-line realization of $G$.
	
	In a straight-line realization of $G$, each $3$-cycle is realized as a triangle.
	Hence, if the lengths of the edges of at least one $3$-cycle do not respect the triangle inequality, then $G$ does not admit any (even non-planar) straight-line realization.
	We can trivially test in $O(n)$ time whether the length function of an $n$-vertex weighted $2$-tree is such that every $3$-cycle satisfies the triangle inequality.
	Hereafter we assume that every weighted $2$-tree satisfies this necessary condition.
	
	In our proofs and algorithms, when clear from the context, we refer interchangeably to the $3$-cycles of a $2$-tree $G$, the nodes of its decomposition tree, and the triangles in a straight-line realization of $G$. 
	
	\paragraph{Triangles in a planar straight-line realization.}
	
	Let $G=(V,E,\lambda)$ be a weighted $2$-tree and $\Gamma$ be a planar straight-line realization of $G$.
	We denote with $\triangle_i$ a triangle \emph{drawn} in $\Gamma$, that is, the straight-line realization of a $3$-cycle of $G$.
	Consider two triangles $\triangle_i$ and $\triangle_j$.
	We say that $\triangle_i$ is \emph{drawn inside} $\triangle_j$ if all the points of $\triangle_i$ are points of $\triangle_j$ and at least one vertex of $\triangle_i$ is an interior point of $\triangle_j$.
	In the planar straight-line realization shown in \cref{fig:decomposition_tree}, the triangles $\triangle_6$, $\triangle_7$, and $\triangle_8$ are drawn inside the triangle $\triangle_3$.
	Observe that, if $\triangle_i$ is drawn inside $\triangle_j$, then $\triangle_i$ and $\triangle_j$ have at most two common vertices.
	On the other hand, in general, $\triangle_i$ can be drawn inside $\triangle_j$ regardless of their distance in the decomposition tree of $G$. We have the following observation, which we use implicitly (and sometimes explicitly) in the paper.
	
	
	\begin{observation}\label{obs:containments}
		Let $\ell_i$ and $\ell_j$ be the lengths of the longest sides of two triangles $\triangle_i$ and $\triangle_j$, respectively.
		If $\ell_i > \ell_j$, then $\triangle_i$ is not drawn in $\triangle_j$.
		If $\ell_i=\ell_j$ and $\triangle_i$ is drawn inside $\triangle_j$, then $\triangle_i$ and $\triangle_j$ share a side of length $\ell_i=\ell_j$.
		%
		%
		%
		%
		%
		%
	\end{observation}
	
	Of special interest for us is the case in which $\triangle_i$ and $\triangle_j$ share an edge.
	If such case occurs, then $\triangle_i$ and $\triangle_j$ are either adjacent or siblings in the decomposition tree of $G$.
	We have the following.
	
	\begin{observation}\label{obs:internal_angles}
		Let $\triangle_i$ and $\triangle_j$ be two triangles sharing an edge with end-vertices $a$ and $b$.
		Let $\alpha_i$ and $\beta_i$ ($\alpha_j$ and $\beta_j$) denote the interior angles of $\triangle_i$ (resp., $\triangle_j$) at $a$ and $b$, respectively.
		If $\triangle_i$ is drawn inside $\triangle_j$, then $\alpha_i < \alpha_j$ and $\beta_i < \beta_j$.
	\end{observation}
	
	We derive further useful properties by considering specific types of triangles.
	Consider an isosceles triangle with base of length $b$ and two sides of length $l$.
	We say the triangle is \emph{tall isosceles} if $l > b$, and \emph{flat isosceles} if instead $l < b$.
	
	%
	%
	
	\begin{lemma}\label{obs:containment_sharing_edges}
		Let $\triangle_i$ and $\triangle_j$ be two triangles sharing an edge.
		\begin{enumerate}[a)]
			\item \label{obs:containment_sharing_edges:tall_in_flat}
			
			If $\triangle_i$ is tall isosceles and $\triangle_j$ is flat isosceles, then $\triangle_i$ is not drawn inside $\triangle_j$.
			
			\item \label{obs:containment_sharing_edges:tall_in_equilateral}
			
			If $\triangle_i$ is tall isosceles and $\triangle_j$ is equilateral, then $\triangle_i$ is not drawn inside $\triangle_j$.
			
			\item \label{obs:containment_sharing_edges:equilateral_in_flat}
			
			If $\triangle_i$ is equilateral and $\triangle_j$ is flat isosceles, then $\triangle_i$ is not drawn inside $\triangle_j$.
			
			
		\end{enumerate}  
	\end{lemma}
	
	\begin{proof}
		All the statements follow from \cref{obs:internal_angles} and the facts that (i) a tall isosceles triangle has one interior angle smaller than $60^{\circ}$ and two interior angles greater than $60^{\circ}$, and (ii) a flat isosceles triangle has one interior angle greater than $60^{\circ}$ and two interior angles smaller than $60^{\circ}$.
		Refer to \cref{fig:isosceles_triangles}.
	\end{proof}
	
	\begin{figure}[tb!]
		\centering
		
		\subcaptionbox{\label{fig:isosceles_triangles:1}}
		{\includegraphics[page=1,scale=1]{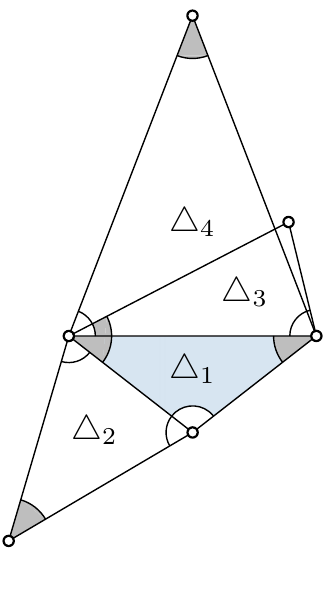}}  
		\qquad%
		\subcaptionbox{\label{fig:isosceles_triangles:2}}
		{\includegraphics[page=2,scale=1]{isosceles_triangles}}  
		\qquad%
		\subcaptionbox{\label{fig:isosceles_triangles:3}}
		{\includegraphics[page=3,scale=1]{isosceles_triangles}}
		
		\caption
		{
			Illustration for \cref{obs:containment_sharing_edges}.
			The angles in gray are smaller than $60^{\circ}$ and the angles in white are greater than $60^{\circ}$.
			(\subref{fig:isosceles_triangles:1}) A tall isosceles triangle ($\triangle_2$, $\triangle_3$, $\triangle_4$, and $\triangle_4$) is not drawn inside a flat isosceles triangle ($\triangle_1$), (\subref{fig:isosceles_triangles:2}) A tall isosceles triangle ($\triangle_2$ and $\triangle_3$) is not drawn inside an equilateral triangle ($\triangle_1$), and (\subref{fig:isosceles_triangles:3}) An equilateral triangle ($\triangle_2$ and $\triangle_3$) is not drawn inside a flat isosceles triangle ($\triangle_1$).
		}
		\label{fig:isosceles_triangles}
	\end{figure}

	Consider now three triangles $\triangle_i$, $\triangle_j$, $\triangle_k$ such that $\triangle_j$ and $\triangle_k$ share a side with $\triangle_i$.
	If the three triangles have a common side, then a necessary condition for $\triangle_i$ and $\triangle_j$ to be drawn inside $\triangle_i$ is given by \cref{obs:internal_angles}.
	If instead there is no common side between $\triangle_j$ and $\triangle_k$, then we have the following lemmas.

	\begin{lemma}\label{lem:flat_in_equilateral}
		Let $\ell_1$ and $\ell_2$ be two lengths such that $2\ell_2 > \ell_1 > \ell_2$, and $\triangle_i$, $\triangle_j$, $\triangle_k$ be three triangles such that:
		\begin{itemize}
			\item
			$\triangle_i$ is an equilateral triangle with sides of length $\ell_1$,
			\item
			$\triangle_j$ and $\triangle_k$ are congruent flat isosceles triangles with base of length $\ell_1$ and two sides of length $\ell_2$,
			\item
			$\triangle_i$ shares one side with $\triangle_j$ and another side with $\triangle_k$, and
			\item $\triangle_j$ and $\triangle_k$ are drawn inside $\triangle_i$ and do not overlap with each other.
		\end{itemize}
		Then $\frac{\ell_1}{\ell_2}>\sqrt 3 \simeq 1.73$.
	\end{lemma}
	
	\begin{proof}  
		Refer to \cref{fig:tall_and_equilateral}. Let $v$ be the common vertex between $\triangle_i$, $\triangle_j$, and $\triangle_k$.
		Since $\triangle_i$ is equilateral the interior angle of $\triangle_i$ at $v$ is $60^{\circ}$.
		Hence, since $\triangle_j$ and $\triangle_k$ are congruent and intersect only at $v$ (they do not overlap), the internal angle $\alpha$ of both $\triangle_j$ and $\triangle_k$ at $v$ is at most $30^{\circ}$.
		We thus have that $\ell_1 = 2\ell_2 \cos(\alpha) > 2\ell_2\cos(30^{\circ})=\sqrt 3{\ell_2}$.
	\end{proof}
	
	\begin{figure}[ht]
		\centering%
		\includegraphics[page=1]{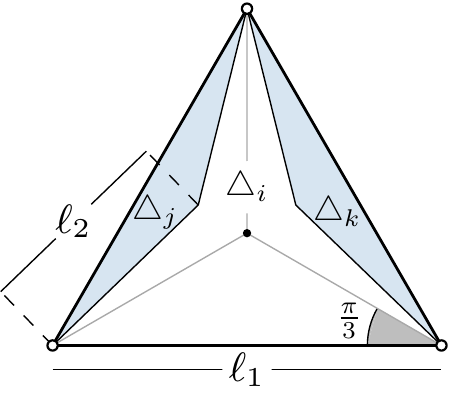}
		\caption
		{Illustration for \cref{lem:flat_in_equilateral}: If two flat isosceles triangles ($\triangle_j$ and $\triangle_k$) are drawn inside an equilateral triangle ($\triangle_i$), then $\frac{\ell_1}{\ell_2}>\sqrt{3}$. The black point is the circumcenter of $\triangle_i$.}
		\label{fig:tall_and_equilateral}
	\end{figure}
	
	
	\remove{
		\begin{figure}[ht]
			\centering%
			\subcaptionbox{\label{fig:tall_and_equilateral:1}}
			{\includegraphics[page=1]{tall_and_equilateral}}
			\qquad%
			\subcaptionbox{\label{fig:tall_and_equilateral:2}}
			{
				\includegraphics[page=2]{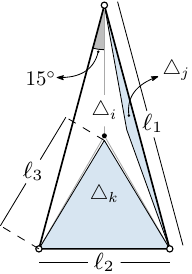}
				\quad
				\includegraphics[page=3]{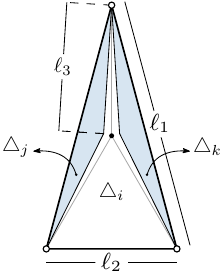}
			}
			\caption
			{
				Illustration for \cref{lem:flat_in_equilateral,lem:flat-isosceles_in_tall-isosceles}.
				The black point is the circumcenter of $\triangle_i$.
				(\subref{fig:tall_and_equilateral:1}) If two flat isosceles triangles ($\triangle_j$ and $\triangle_k$) are drawn inside an equilateral triangle ($\triangle_i$), then $\frac{\ell_1}{\ell_2}>\sqrt{3}$.
				(\subref{fig:tall_and_equilateral:2})
				If two flat isosceles triangles ($\triangle_j$ and $\triangle_k$) are drawn inside a tall isosceles triangle ($\triangle_i$), then $\frac{\ell_1}{\ell_3} > 2 \cos(15^\circ)$.
			}
			\label{fig:tall_and_equilateral}
		\end{figure}
		\todo[inline]{This seems wrong to me. How about if $\ell_1 = \ell_2+\epsilon$ and then the angles on top are close to $30^{\circ}$ not $15^\circ$.}
		\begin{lemma}\label{lem:flat-isosceles_in_tall-isosceles}
			Let $\ell_1$, $\ell_2$, $\ell_3$ be three lengths such that $2\ell_2 > \ell_1 > \ell_2$ and $2\ell_3 > \ell_2 > \ell_3$, and $\triangle_i$, $\triangle_j$, $\triangle_k$ be three triangles such that:
			\begin{itemize}
				\item
				$\triangle_i$ is a tall isosceles triangle with base of length $\ell_2$ and two sides of length $\ell_1$,
				\item
				$\triangle_j$ is a flat isosceles triangle with base of length $\ell_1$ and two sides of length $\ell_3$,
				\item
				$\triangle_k$ is a flat isosceles triangle either congruent to $\triangle_j$, or with base of length $\ell_2$ and two sides of length~$\ell_3$,
				\item
				$\triangle_i$ shares one side with $\triangle_j$ and another side
				with $\triangle_k$, and
				\item
				$\triangle_j$ and $\triangle_k$ are drawn inside $\triangle_i$  and do not overlap with each other.
			\end{itemize}
			Then $\frac{\ell_1}{\ell_3} > 2 \cos(15^\circ) \simeq 1.93$.
		\end{lemma}
		
		\begin{proof}
			Let $v$ be the common vertex between $\triangle_i$, $\triangle_j$, and $\triangle_k$, and $\alpha$ be the interior angle of $\triangle_j$ at $v$.
			Note that $\triangle_j \cap \triangle_k = v$ for $\alpha < 15^\circ$, while $\triangle_j$ and $\triangle_k$ overlap if instead $\alpha \geq 15^\circ$.
			Since $\triangle_j$ and $\triangle_k$ intersect only at $v$, we have that $\ell_1 = 2\ell_3\cos(\alpha) > 2\ell_3\cos(15^\circ)$.
			Refer to \cref{fig:tall_and_equilateral:2}.
		\end{proof}
	}

	\section{Prescribed Embedding}
	\label{sec:prescribed_embedding}
	
	In this section, we describe an $O(n)$-time algorithm to solve the
	\FEPRshort problem for an $n$-vertex $2$-tree whose embedding or
	whose rotation system is prescribed.
	
	We start by showing how to check in linear time whether a straight-line
	realization of a $2$-tree with prescribed embedding is a planar
	straight-line realization respecting the embedding. This basic tool will be exploited throughout the paper.
	
	\begin{theorem}\label{thm:straight-line_realization_planarity}
		Let $G=(V,E,\lambda)$ be an $n$-vertex weighted $2$-tree,
		$\mathcal{E}$ be a plane embedding of $G$, and $\Gamma$ be a
		straight-line realization of $G$. There exists an $O(n)$-time
		algorithm that tests whether $\Gamma$ is a planar straight-line
		realization respecting $\mathcal{E}$.
	\end{theorem}
	
	\begin{proof}
		First, we check whether $\Gamma$ respects $\mathcal{E}$. This can be done in $O(n)$ time by checking: (i) whether the clockwise order of the edges incident to each vertex in $\Gamma$ is the same as in $\mathcal{E}$; and (ii) whether the cycle bounding the outer face of $\Gamma$ (if $\Gamma$ is planar) is the same as in $\mathcal{E}$. More in detail, this check is as follows.
		
		\begin{itemize}
			\item Check (i) only requires to scan, for each vertex $v$ of $G$, the list representing the clockwise order of the edges incident to $v$ in $\mathcal{E}$, and to confirm that the slopes of such edges in $\Gamma$ form a decreasing sequence (where the slope of one of such edges is the one of the half-ray from $v$ through the other end-vertex of the edge and has a value between the slope $s$ of the first considered edge and $s-360^{\circ}$). This can be done in time proportional to the degree of $v$ and hence in $O(n)$ time over all the vertices of $G$.
			\item In order to perform check (ii), we find in $O(n)$ time the vertices $v_{\leftarrow}$ and $v_{\rightarrow}$ with smallest and largest $x$-coordinates, respectively, among the vertices of the cycle $\mathcal C$ bounding the outer face of $\mathcal{E}$. Then we check whether the edge that leaves $v_{\leftarrow}$ in clockwise direction along $\mathcal C$ lies above the edge that enters $v_{\leftarrow}$ in clockwise direction along $\mathcal C$ in $\Gamma$. Together with check (i), this guarantees that all the edges incident to vertices of $\mathcal C$ leave the vertices in $\mathcal C$ towards the region inside $\mathcal C$ in $\Gamma$, if $\Gamma$ is planar. It might still be the case, however, that $\mathcal C$ self-intersects or that other edges of $G$ cross $\mathcal C$, hence at this point we can guarantee that $\mathcal C$ bounds the outer face of $\Gamma$ only by assuming its planarity.
		\end{itemize} 
		
		%
		In order to check whether $\Gamma$ is planar, we augment $\Gamma$ to a straight-line drawing $\Gamma'$ of a maximal planar graph so that, if $\Gamma$ is planar, then $\Gamma'$ is planar, as well. This is done as follows. If $\Gamma$ is a planar straight-line realization respecting $\mathcal{E}$, every internal face $f$ of $\mathcal{E}$ is bounded in $\Gamma$ by a polygon $\mathcal P_f$ whose interior is empty. We  triangulate the interior of $\mathcal P_f$ by the addition of dummy edges. This can be done in linear time over all the internal faces of $\mathcal{E}$ by means of the algorithm by Chazelle~\cite{c-tsp-91}. We also triangulate the outer face of $\Gamma$. This is done similarly as described by Cabello et al.~\cite{JGAA-145}. Namely, we enclose $\Gamma$ into a large equilateral triangle $T$ with a horizontal side $\ell$. This partitions the outer face of $\Gamma$ into two regions, namely a bounded region $\mathcal R$ and an unbounded region $\mathcal R'$. We then add edges from $v_{\leftarrow}$ and $v_{\rightarrow}$ to the left and right end-vertices of $\ell$, respectively. This subdivides $\mathcal R$ into two regions delimited by simple polygons. We triangulate the interior of such polygons, again using Chazelle's algorithm~\cite{c-tsp-91}. Denote by $\Gamma'$ the resulting straight-line drawing.
		
		If $\Gamma$ is a planar straight-line realization respecting $\mathcal{E}$, then (as described in~\cite{JGAA-145}) every execution of Chazelle's algorithm terminates correctly. Hence, if at least one of the executions of Chazelle's algorithm terminates in error, we conclude that $\Gamma$ is not a planar straight-line realization respecting $\mathcal{E}$. It might also be the case that every execution of Chazelle's algorithm terminates correctly and yet $\Gamma$ contained crossings. Hence, we also need to check whether $\Gamma'$ is planar. This can be done in $O(n)$ time, as if $\Gamma'$ is planar, it is a triangulation (i.e., a planar straight-line drawing such that every face is delimited by a triangle) and the planarity of convex subdivisions can be tested in $O(n)$ time~\cite{dlpt-ccpps-98}. \end{proof}
	
	\begin{theorem}\label{thm:straight-line_realization_planarity-rotation}
		Let $G=(V,E,\lambda)$ be an $n$-vertex weighted $2$-tree, let $\mathcal R$ be a rotation system for $G$, and let $\Gamma$ be a
		straight-line realization of $G$. There exists an $O(n)$-time
		algorithm that tests whether $\Gamma$ is a planar straight-line
		realization whose rotation system is $\mathcal{R}$.
	\end{theorem}
	
	\begin{proof}
		We are going to recover the cycle $\mathcal C$ that bounds the outer face of $\Gamma$ (if $\Gamma$ is planar) by exploiting the knowledge of $\mathcal{R}$ and of $\Gamma$ itself.
		Then it suffices to invoke~\cref{thm:straight-line_realization_planarity} in order to test whether $\Gamma$ is a planar straight-line realization respecting $\mathcal E$, where $\mathcal E$ is the plane embedding that has $\mathcal{R}$ as rotation system and $\mathcal C$ as cycle bounding the outer face.
		
		In order to recover the cycle $\mathcal C$ that bounds the outer face of $\Gamma$, we first find in $O(n)$ time the leftmost vertex $v_0$ in $\Gamma$ and we then find in $O(n)$ time the edge $(v_0,v_1)$ incident to $v_0$ with largest slope; note that, if $\Gamma$ is planar, then $v_0$ and $v_1$ are encountered consecutively when traversing $\mathcal C$ in clockwise direction. Now suppose that, for some $i\geq 1$, a path $(v_0,v_1,\dots,v_i)$ has been found such that, if $\Gamma$ is planar and has $\mathcal R$ as rotation system, then the vertices $v_0,v_1,\dots,v_i$ are encountered consecutively when traversing $\mathcal C$ in clockwise direction. Then, if $\Gamma$ is planar and has $\mathcal R$ as rotation system, the vertex $v_{i+1}$ that is encountered after $v_i$ when traversing $\mathcal C$ in clockwise direction is the vertex that follows $(v_i,v_{i-1})$ in the clockwise order of the edges incident to $v_i$ in $\mathcal R$; such a vertex can be found in $O(1)$ time. When this process encounters a vertex $v_k$ for the second time, if such a vertex is not $v_0$, then we conclude that $\Gamma$ is not a planar straight-line realization whose rotation system is $\mathcal{R}$, otherwise we have found the desired cycle $\mathcal C$. 
	\end{proof}
	
	We note that theorems as the previous two cannot be stated for general planar graphs. Indeed, an easy reduction from {\sc Element Uniqueness}~\cite{b-lbact-83} shows that even testing whether a straight-line drawing of an $n$-vertex graph with no edges is planar requires $\Omega(n \log n)$ time.
	
	In order to prove the main results of this section, we need the following lemma, which might be of independent interest and which is applicable to planar graphs that are not necessarily $2$-trees.
	
	\begin{lemma} \label{le:ds-order}
		Let $G$ be an $n$-vertex planar graph and let $\mathcal R$ be a rotation system for $G$. A data structure can be set up in $O(n)$ time that allows one to answer in $O(1)$ time the following type of queries: Given three edges $e_i$, $e_j$, and $e_k$ incident to a vertex $v$ of $G$, determine whether they appear in the order $e_i$, $e_j$, and $e_k$ or in the order $e_i$, $e_k$, and $e_j$ in the clockwise order $\mathcal R(v)$ of the edges incident to $v$ in $\mathcal R$.
	\end{lemma}
	
	\begin{proof}
		The proof's main idea is the same as the one of Aichholzer et al.~\cite{akm-oarp-17} for the $O(n^2)$-time construction of a data structure that allows to answer in $O(1)$ time queries about the orientation of triples of points in the plane.
		
		Consider any vertex $v$ of $G$. Let $\mathcal R(v)=[e_1,\dots,e_\ell]$; note that $\mathcal R(v)$ is a circular list, which is here linearized by picking any edge as $e_1$. Label each edge $e_i$ with the value $(v,i)$. This labeling labels each edge of $G$ twice, hence it is performed in $O(n)$ time. The data structure claimed in the statement simply consists of $G$ equipped with these edge labels.
		
		Now suppose that a query as in the statement has to be answered for three edges $e_i$, $e_j$, and $e_k$ incident to a vertex $v$ of $G$. We recover in $O(1)$ time the labels $(v,i)$, $(v,j)$, and $(v,k)$ respectively associated to $e_i$, $e_j$, and $e_k$. Assume that $i$ is smaller than $j$ and $k$; then $e_i$, $e_j$, and $e_k$ appear in the order $e_i$, $e_j$, and $e_k$ in $\mathcal R(v)$ if and only if $j<k$, which can be checked in $O(1)$ time. Analogously, if $j<i$ and $j<k$, we have that $e_i$, $e_j$, and $e_k$ appear in the order $e_i$, $e_j$, and $e_k$ in $\mathcal R(v)$ if and only if $k<i$, while if $k<i$ and $k<j$, we have that $e_i$, $e_j$, and $e_k$ appear in the order $e_i$, $e_j$, and $e_k$ in $\mathcal R(v)$ if and only if $i<j$.
	\end{proof}
	
	We now present the main results of this section.
	
	\begin{theorem}\label{th:fixed-rotation}
		Let $G$ be an $n$-vertex weighted $2$-tree and $\mathcal{R}$ be a rotation system for $G$. There exists an $O(n)$-time algorithm that tests whether $G$ admits a planar straight-line realization whose rotation system is $\mathcal{R}$ and, in the positive case, constructs such a realization.
	\end{theorem}
	
	\begin{proof}
		The algorithm is as follows. First, among all the $3$-cycles of $G$, we pick a $3$-cycle $c$ with largest sum of the edge lengths, which can clearly be done in $O(n)$ time. Second, we compute in $O(n)$ time the decomposition tree $T$ of $G$ rooted at $c$. Third, by means of~\cref{le:ds-order}, we set up in $O(n)$ time a data structure that allows us to decide in $O(1)$ time whether any three edges $e_i$, $e_j$, and $e_k$ incident to a vertex $v$ of $G$ appear in the order $e_i$, $e_j$, and $e_k$ or in the order $e_i$, $e_k$, and $e_j$ in the clockwise order of the edges incident to $v$ in $\mathcal R$.
		
		Let $\gamma$ be any straight-line realization of $c$ and let $\gamma'$ be the reflection of $\gamma$. Note that $\gamma$ is unique, up to rigid transformations. Hence, if there exists a planar straight-line realization $\Delta$ of $G$ whose rotation system is $\mathcal{R}$ and whose restriction to $c$ is some triangle $\delta$, then there also exists a planar straight-line realization $\Gamma$ of $G$ whose rotation system is $\mathcal{R}$ and whose restriction to $c$ is either $\gamma$ or $\gamma'$. Indeed, $\Gamma$ can be obtained from $\Delta$ by applying the rotation and translation that turn $\delta$ into either $\gamma$ or $\gamma'$; then the rotation system for $\Gamma$ is the same as the one for $\Delta$, hence it is $\mathcal R$.
		
		We show how to test in $O(n)$ time whether there exists a planar straight-line realization of $G$ whose rotation system is $\mathcal{R}$ and whose restriction to $c$ is $\gamma$ (in the positive case, such a realization is constructed within the same time bound). An analogous test can be performed with the constraint that $c$ is represented by $\gamma'$ rather than $\gamma$. Then the algorithm concludes that a planar straight-line realization of $G$ whose rotation system is $\mathcal{R}$ exists if and only if (at least) one of the two tests is successful.
		
		We are going to use the following key property.
		
		\begin{property} \label{pr:first-edge}
			Let $(u,v,w)$ be a cycle of $G$ such that $u$ is not a vertex of $c$. Let $P$ be a shortest path from $u$ to (any vertex of) $c$ and suppose that the edge of $P$ incident to $u$, say $(u,z)$, is such that $z$ is different from $v$ and $w$. Then, in any planar straight-line realization of $G$, the edge $(u,z)$ lies outside $(u,v,w)$.
		\end{property}
		
		\begin{proof}
			Consider any planar straight-line realization $\Gamma$ of $G$. Since $z$ is different from $v$ and $w$ and since $P$ is a shortest path, we have that $P$ contains neither $v$ nor $w$. Hence, if $z$ lies inside $(u,v,w)$, then so does $c$, by the planarity of $\Gamma$. However, this is not possible, since $c$ is a $3$-cycle whose sum of the edge lengths is maximum, by definition. 
		\end{proof}
		
		
		
		Our algorithm visits $T$ in pre-order starting at $c$; let $c_1,\dots,c_{n-2}$ be the order in which the nodes of $T$ are visited. For $i=1,\dots,n-2$, let $G_i$ be the subgraph of $G$ composed of the $3$-cycles $c_1,\dots,c_i$; note that $G_1=c=c_1$ and $G_{n-2}=G$. Also, let $\mathcal R_i$ be the restriction of $\mathcal R$ to $G_i$.
		
		Clearly, if $\Gamma$ is a planar straight-line realization of $G$ whose rotation system is $\mathcal{R}$ and whose restriction to $c$ is $\gamma$, then, for any $i=1,\dots,n-2$, the restriction of $\Gamma$ to $G_i$ is a planar straight-line realization $\Gamma_i$ of $G_i$ whose rotation system is $\mathcal{R}_i$ and whose restriction to $c$ is $\gamma$. We prove that, if such a realization $\Gamma_i$ of $G_i$ exists, then it is unique and can be computed in $O(1)$ time from $\Gamma_{i-1}$ (from $\gamma$, if $i=1$). During the visit of $T$ we maintain, for each vertex $u$ of $G_i$ that does not belong to $c$, the first edge $(u,z)$ of a shortest path $P_u$ from $u$ to (any vertex of) $c$ in $G_i$. Further, we maintain, for each vertex $u$ of $G_i$, a value $\ell(u)$; this is equal to the number of edges of $P_u$, if $u$ does not belong to $c$, and to $0$, otherwise. 
		
		When visiting $c_1$, we set $\Gamma_1=\gamma$ in $O(1)$ time. Obviously, $\Gamma_1$ is the unique planar straight-line realization of $G_1$ in which the straight-line realization of $c_1$ is  $\gamma$ (the rotation system $\mathcal R_1$ does not impose any constraint). We initialize $\ell(u)=0$ for each vertex $u$~of~$c_1$.
		
		Assume that, for some $i\in \{1,\dots,n-3\}$, a straight-line realization $\Gamma_i$ of $G_i$ has been constructed such that, if a planar straight-line realization whose rotation system is $\mathcal R_i$ and whose restriction to $c$ is $\gamma$ exists, then $\Gamma_i$ is the unique such a realization. Let $c_{i+1}=(u,v,w)$, where $e_{uv}:=(u,v)$ is the edge that $c_{i+1}$ shares with the $3$-cycle represented by the parent of $c_{i+1}$ in $T$, while $w$ is the unique vertex of $G_{i+1}$ not in $G_i$. Let $e_{uw}$ and $e_{vw}$ denote the edges $(u,w)$ and $(v,w)$, respectively. Assume, w.l.o.g., that $\ell(u)\leq \ell(v)$. Then we compute in $O(1)$ time the value $\ell(w)$ as $\ell(w)=1+\ell(u)$ (the $\ell(\cdot)$ value of every other vertex of $G_{i+1}$ coincides with the one in $G_i$). Also, again in $O(1)$ time, we set $e_{uw}$ as the first edge of a shortest path from $w$ to $c$ (for every other vertex of $G_{i+1}$, the first edge of a shortest path to $c$ coincides with the one in $G_i$).
		
		Note that there are two placements for $w$ that result in the edges $e_{uw}$ and $e_{vw}$ having the prescribed edge lengths, thus turning $\Gamma_i$ into a straight-line realization of $G_{i+1}$. We show that \cref{pr:first-edge} can be used in order to define a unique placement for $w$ in $\Gamma_i$. 
		
		If $\ell(u)>0$, let $e_{uz}:=(u,z)$ be the first edge of the shortest path from $u$ to $c$ in $G_i$ (recall that we maintain this information). If $\ell(u)=0$, then $u$ is a vertex of $c$; then we define $e_{uz}:=(u,z)$ as any edge of $c$ different from $e_{uv}$ (note that $e_{uv}$ might or might not be an edge of $c$). In both cases, we use~\cref{le:ds-order} to determine in $O(1)$ time whether the clockwise order of the edges $e_{uv}$, $e_{uw}$, and $e_{uz}$ in $\mathcal R$ is $e_{uv}$, $e_{uw}$, and $e_{uz}$ or $e_{uv}$, $e_{uz}$, and $e_{uw}$. Suppose, without loss of generality, it is the former. We now proceed according to the following case distinction.

		%
		
		

		\begin{figure}[htb]
			\centering
			\subfloat[]{
				\includegraphics[page=1,height=.13\textwidth]{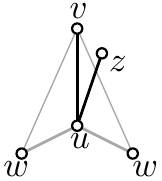}
				\label{fig:fixed-rotation-case1}
			}\hfill
			\subfloat[]{
				\includegraphics[page=2,height=.13\textwidth]{fixed-rotation}
				\label{fig:fixed-rotation-case2}
			}\hfill
			\subfloat[]{
				\includegraphics[page=3,height=.13\textwidth]{fixed-rotation}
				\label{fig:fixed-rotation-case3}
			}
			\caption{Illustrations of Cases 1, 2, and 3 for the insertion of $w$ and its incident edges in $\Gamma_i$. Both possible placements of $w$ are shown (a chosen placement is shown in black).}
			\label{fig:fixed-rotation}
		\end{figure}
		
		\begin{itemize}
			\item[Case 1:] If neither of the two possible placements of $w$ guarantees that the clockwise order of the edges $e_{uv}$, $e_{uw}$, and $e_{uz}$ in $\Gamma_{i+1}$ is $e_{uv}$, $e_{uw}$, and $e_{uz}$, as in~\cref{fig:fixed-rotation-case1}, then we conclude that there exists no planar straight-line realization of $G_{i+1}$ whose rotation system is $\mathcal R_{i+1}$ and whose restriction to $c$ is $\gamma$. 
			\item[Case 2:] If exactly one of the two possible placements of $w$ guarantees that the clockwise order of the edges $e_{uv}$, $e_{uw}$, and $e_{uz}$ in $\Gamma_{i+1}$ is $e_{uv}$, $e_{uw}$, and $e_{uz}$, as in~\cref{fig:fixed-rotation-case2}, then we choose that placement for $w$. Indeed, choosing the other placement for $w$ would imply that, even if the resulting straight-line realization of $G_{i+1}$ was planar, it would not have $\mathcal R_{i+1}$ as its rotation system. 
			\item[Case 3:] If both the possible placements of $w$ guarantee that the clockwise order of the edges $e_{uv}$, $e_{uw}$, and $e_{uz}$ in $\Gamma_{i+1}$ is $e_{uv}$, $e_{uw}$, and $e_{uz}$, as in~\cref{fig:fixed-rotation-case3}, then we observe that one of the two placements of $w$ is such that (at least part of) the edge $e_{uz}$ is internal to the triangle $(u,v,w)$, while the other placement is such that the edge $e_{uz}$ is external to the triangle $(u,v,w)$. We choose the latter placement for $w$.

			If $\ell(u)>0$, this choice is motivated by~\cref{pr:first-edge}. Indeed, $u$ is not a vertex of $c$, given that $\ell(u)>0$; further, $z$ is different from $w$ (as $z$ belongs to $G_i$, while $w$ does not) and from $v$ (by the assumption $\ell(u)\leq \ell(v)$ and $\ell(z)=\ell(u)-1$). Hence,~\cref{pr:first-edge} applies and in any planar straight-line realization of $G_{i+1}$, the edge $e_{uz}$ is external to $(u,v,w)$.
			
			If $\ell(u)=0$, then $e_{uz}$ is an edge of $c$, hence if $e_{uz}$ were internal to $(u,v,w)$, then $c$ would be internal to $(u,v,w)$, as well (the two cycles share the vertex $u$ and, possibly, the vertex $v$ and the edge $(u,v)$). However, this is not possible in any planar straight-line realization of $G_{i+1}$, as the sum of the edge lengths of $c$ is greater than or equal to the sum of the edge lengths of $(u,v,w)$, by the definition of $c$.
		\end{itemize}
		
		The clockwise order of the edges $e_{uv}$, $e_{uw}$, and $e_{uz}$ according to a placement for $w$ can be computed in $O(1)$ time, as well as whether the edge $e_{uz}$ is external to the triangle $(u,v,w)$ or not. Hence, $\Gamma_{i+1}$ can be constructed in $O(1)$ time from $\Gamma_i$. 
		
		When $i=n-2$, we have that, if there exists a planar straight-line realization of $G=G_{n-2}$ whose rotation system is $\mathcal{R}=\mathcal{R}_{n-2}$ and whose restriction to $c$ is $\gamma$, then the constructed representation $\Gamma:=\Gamma_{n-2}$ is the unique such a realization. Hence, it only remains to test whether $\Gamma$ is a planar straight-line realization of $G$ whose rotation system is $\mathcal{R}$. This can be done in $O(n)$ time by means of~\cref{thm:straight-line_realization_planarity-rotation}. 
	\end{proof}
	
	
	We can similarly deal with the case in which the cycle bounding the outer face of the planar straight-line realization is prescribed, as in the following theorem.
	
	\begin{theorem} \label{thm:fixed-embedding}
		Let $G$ be an $n$-vertex weighted $2$-tree and $\mathcal{E}$ be a plane embedding of $G$. There exists an $O(n)$-time algorithm that tests whether $G$ admits a planar straight-line realization that respects $\mathcal{E}$ and, in the positive case, constructs such a realization. 
	\end{theorem}
	
	\begin{proof}
		Let $\mathcal R$ be the rotation system associated with $\mathcal{E}$.
		By exploiting $\mathcal R$ and ignoring the knowledge of the cycle delimiting the outer face of $\mathcal{E}$, at most two straight-line realizations of $G$ can be constructed in $O(n)$ time, exactly as in the proof of~\cref{th:fixed-rotation}, such that $G$ admits a planar straight-line realization whose rotation system is $\mathcal R$ if and only if one of these two realizations is planar and has  $\mathcal R$ as its rotation system.   Differently from the proof of~\cref{th:fixed-rotation}, we use \cref{thm:straight-line_realization_planarity} rather than~\cref{thm:straight-line_realization_planarity-rotation} to test in $O(n)$ time whether any of the two realizations is planar and respects $\mathcal{E}$.
	\end{proof}

	\section{NP-hardness for 2-trees with four edge lengths}
	\label{sec:np-hardness}
	
	In this section, we present a polynomial-time reduction from the {\sc Planar Monotone 3-SAT} problem to the \FEPRshort problem with {\em four} edge lengths.
	As the former is known to be \NP-complete~\cite{DBLP:journals/ijcga/BergK12}, the reduction shows \mbox{\NP-hardness of the latter.}
	We thus have the following.
	
	\begin{theorem}\label{th:np-hard}
		The \FEPR problem is NP-hard for weighted $2$-trees, even for instances whose number of distinct edge lengths is $4$.
	\end{theorem}
	
	Let $\phi$ be a Boolean formula in conjunctive normal form with at most three literals on each clause.
	We denote by $G_\phi$ the \emph{incidence graph} of $\phi$, that is, the graph that has a vertex for each clause of $\phi$, a vertex for each variable of $\phi$, and an edge $(c_i,v_j)$ for each clause $c_i$ that contains the literal
	$v_j$ or the literal $\overline{v}_j$.
	The formula~$\phi$ is an instance of {\sc Planar Monotone 3-SAT} if $G_\phi$ is planar and each clause of $\phi$ is either \emph{positive} or \emph{negative}.
	A positive clause contains only \emph{positive literals} (i.e., literals $v_j$ for some variable $v_j$), while a negative clause only contains \emph{negated literals} (i.e., literals $\overline{v}_j$ for some variable $v_j$).
	
	
	%
	Let $\phi$ be an instance of {\sc Planar Monotone 3-SAT}.
	Hereafter we assume that each clause of $\phi$ contains {\em exactly} three literals.
	This is not a loss of generality, since a clause with less than three literals can be modified by duplicating one of the literals in the clause, which does not alter the satisfiability of $\phi$.
	Note this modification might turn $G_\phi$ into a planar multi-graph.
	Nevertheless, since this is not relevant for our reduction, it will be ignored hereforth.
	
	\begin{figure}[b!]
		\centering%
		\subcaptionbox{\label{fig:planar-3-sat-original} $\Gamma_\phi$}
		{\includegraphics[page=5,height=.3\textwidth]{variable-clause-graph}}%
		\hfil
		\subcaptionbox{\label{fig:planar-3-sat-modified}$\Gamma^*_\phi$}
		{\includegraphics[page=3,height=.3\textwidth]{variable-clause-graph}}
		
		\caption
		{
			Monotone rectilinear representations of $G_\phi$.
			Boxes representing variables are red, whereas boxes and trapezoids representing clauses are blue.
			(\subref{fig:planar-3-sat-original}) The original monotone rectilinear representation $\Gamma_\phi$.
			(\subref{fig:planar-3-sat-modified}) The modified monotone rectilinear representation $\Gamma^*_\phi$ that satisfies properties D\ref{prop:tiling}--D\ref{prop:clause-distances}.
		}
		\label{fig:planar-3-sat}
	\end{figure}
	
	A \emph{monotone rectilinear representation} of $G_\phi$ is a drawing that satisfies the following properties (refer to \cref{fig:planar-3-sat-original}):
	\begin{enumerate}[\bf {P}1:]
		\item
		Variables and clauses are represented by axis-aligned boxes with the same height.
		\item
		The bottom sides of all boxes representing variables lie on the same horizontal line.
		\item
		The boxes representing positive (resp.\ negative) clauses lie above (resp.\ below) the boxes representing variables.
		\item
		The edges connecting variables and clauses are vertical segments.
		\item
		The drawing is crossing-free.
	\end{enumerate}
	The {\sc Planar Monotone 3-SAT} problem has been shown to be \NP-complete, even when the incidence graph is provided along with a monotone rectilinear representation~\cite{DBLP:journals/ijcga/BergK12}.
	Given an instance $\phi$ of the {\sc Planar Monotone 3-SAT} problem and a rectilinear representation $\Gamma_\phi$ of $\phi$, we construct a weighted $2$-tree $H_\phi$ that admits a planar straight-line realization {\em if and only if} $\phi$ is satisfiable.
	The general strategy of the reduction is as follows.
	In \cref{sec:auxiliary_representation} we describe how to modify $\Gamma_\phi$ into an auxiliary representation $\Gamma^*_\phi$ that satisfies a set of convenient properties (see \cref{fig:planar-3-sat-modified}).
	We then exploit the geometric information of $\Gamma^*_\phi$ to construct $H_\phi$, so that the edges of $H_\phi$ are assigned four distinct lengths. Namely, in \cref{sec:np-hardness_gadgets}, we describe gadgets for the variables, for the clauses, and for the edges of $G_\phi$.
	Finally, in \cref{sec:np-hardness_combine_gadgets}, we show how to combine these gadgets to form $H_\phi$.
	
	\subsection{The auxiliary monotone rectilinear representation}
	\label{sec:auxiliary_representation}
	
	We start by describing how to transform $\Gamma_\phi$ into a new representation $\Gamma^*_\phi$ that satisfies the following properties (refer to \cref{fig:planar-3-sat-modified}):
	
	\begin{enumerate}[\bf {D}1:]
		\item \label{prop:tiling}
		
		The corners of the polygons and the end-points of the segments forming $\Gamma^*_\phi$ lie on the points of the grid formed by equidistant lines with slope $0^\circ$, $60^\circ$, and $120^\circ$, in which each grid cell is an equilateral triangle with sides of unit length.
		
		\item \label{prop:bounding-box}
		
		The height and width of the bounding box of $\Gamma^*_\phi$ are polynomially bounded in the size of $\phi$.
		
		\item \label{prop:variables-alignment}
		
		The variables are represented by axis-aligned boxes.
		Let $\delta_\phi$ denote the maximum degree of $G_\phi$.
		The boxes have width $2\delta_\phi-1$, height $2\sqrt 3 (\delta_\phi-1)$, and have their bottom sides on a common horizontal grid~line.
		
		\item \label{prop:clause-representation}
		
		The clauses are represented by trapezoids.
		Each trapezoid has height $2\sqrt 3$, lateral sides with slope $60^\circ$, and horizontal sides whose length is larger than $8$.
		
		\item \label{prop:edge-slopes}
		
		The edges connecting variables and clauses are line segments with slope $60^{\circ}$.
		
		\item \label{prop:variables-attachements}
		
		Consider the box $\mathcal{B}$ representing a variable.  
		The edges incident to $\mathcal{B}$ have horizontal distance from the left side of $\mathcal{B}$ that is a multiple of $2$.
		
		\item \label{prop:clause-distances}
		
		Consider the trapezoid $\mathcal{T}$ representing a positive (resp.\ negative) clause $c$.
		Let $p_1$, $p_2$, and $p_3$ be the intersection points between the segments representing the edges of $G_\phi$ incident to $c$ and the bottom (resp.\ top)
		horizontal side of $\mathcal{T}$.
		The point $p_1$ lies on the bottom-left (resp.\ top-left) corner of $\mathcal{T}$, the horizontal distance between $p_1$ and $p_2$ is at least three, the horizontal distance between $p_2$ and $p_3$ is at least four, and the horizontal distance between $p_3$ and the bottom-right (resp.\ top-right) corner~of~$\mathcal{T}$ is equal to one.
		
	\end{enumerate}
	
	The transformation is described in the following lemma.
	
	\begin{lemma}\label{lem:auxiliary-representation}
		The drawing $\Gamma^*_\phi$ can be constructed in polynomial time starting from $\Gamma_\phi$.
	\end{lemma}
	
	\begin{proof}
		We first transform $\Gamma_\phi$ into an intermediate rectilinear representation $\Gamma^\prime_\phi$ that satisfies properties D\ref{prop:tiling}--D\ref{prop:variables-alignment}, property D\ref{prop:variables-attachements}, and two properties that we denote by {\bf D4'} and {\bf D7'}, obtained from properties D\ref{prop:clause-representation} and D\ref{prop:clause-distances}, respectively, by substituting trapezoids with boxes.
		Since every segment of a geometric object forming $\Gamma_\phi$ is either horizontal or vertical, we can transform $\Gamma_\phi$ into $\Gamma'_\phi$ by suitable scaling the boxes representing variables and clauses, and by a suitable deformation of the plane.
		
		We transform $\Gamma'_\phi$ into $\Gamma^*_\phi$ as follows.
		Observe that, since $\Gamma'_\phi$ satisfies property D\ref{prop:variables-alignment}, the top (resp. bottom) sides of the boxes representing variables lie on a horizontal line $\ell_t$ (resp. $\ell_b$).
		The transformation consists of two horizontal shear mappings that transform vertical segments into segments with slope $60^\circ$ and are applied, respectively, to the points of the plane above $\ell_t$ and below $\ell_b$. More precisely, these shear mappings are as follows.
		Translate the Cartesian reference system so that the $x$-axis coincides with $\ell_t$ (and the origin is anywhere).
		The first shear mapping then maps every point $(x,y)$ with $y\geq 0$ to the point $(x+\frac{\sqrt 3}{2}y,y)$.
		Again translate the Cartesian reference system so that the $x$-axis coincides with $\ell_b$ (and the origin is anywhere).
		The second shear mapping then maps every point $(x,y)$ with $y\leq 0$ to the point $(x-\frac{\sqrt 3}{2}y,y)$.
		
		
		

		We show next that the resulting representation $\Gamma^*_\phi$ satisfies properties D\ref{prop:tiling}--D\ref{prop:clause-distances}:
		\begin{itemize}
			\item[D\ref{prop:tiling}:]
			
			This property is satisfied since $\Gamma'_\phi$ satisfies property D\ref{prop:tiling}. Note that the distance between two consecutive grid lines (with slope either $0^{\circ}$, or $60^{\circ}$, or $120^{\circ}$) is $\frac{\sqrt{3}}{2}$.
			
			\item[D\ref{prop:bounding-box}:]
			
			Let $\mathcal{B}^\prime$ and $\mathcal{B}^*$ denote, respectively, the bounding boxes of $\Gamma^\prime_\phi$ and $\Gamma^*_\phi$.
			Property D\ref{prop:bounding-box} is satisfied since $\Gamma^\prime_\phi$ satisfies property D\ref{prop:bounding-box}, the height of $\mathcal{B}^*$ is equal to the height of $\mathcal{B}^\prime$, and the width of $\mathcal{B}^*$ is at most the width of $\mathcal{B}^\prime$ plus $\sqrt 3$ times the height of $\mathcal{B}^\prime$.
			
			\item[D\ref{prop:variables-alignment} and D\ref{prop:variables-attachements}:]
			
			These properties are satisfied since $\Gamma'_\phi$ satisfies properties D\ref{prop:variables-alignment} and D\ref{prop:variables-attachements}, and the shear mappings do not move the points contained in the horizontal strip bounded by $\ell_t$ and $\ell_b$.
			
			\item[D\ref{prop:clause-representation}, D\ref{prop:edge-slopes}, and D\ref{prop:clause-distances}:]
			
			These properties are satisfied since $\Gamma'_\phi$ satisfies properties D4' and D7', since the shear mappings transform vertical segments outside the horizontal strip delimited by $\ell_t$ and $\ell_b$ into segments with slope $60^\circ$, and since the shear mappings do not alter the distance between any two points on the same horizontal line.
			
		\end{itemize}
		
		We complete the proof by observing that all the transformations described above can be clearly computed in polynomial time.
	\end{proof}
	
	
	\subsection{Overview of the reduction}
	\label{sec:np-hardness_overview}
	
	For the sake of clarity, before describing in detail the reduction from $\Gamma^*_\phi$ to $H_\phi$, we provide next a high-level description of the gadgets we employ to obtain $H_\phi$.
	In \cref{sec:np-hardness_gadgets} we present complete constructions, as well as lemmas to guarantee that the gadgets behave as required for the correctness of the reduction.
	
	
	
	\begin{figure}[tb!]
		\centering
		\includegraphics[page=1,width=.8\textwidth]{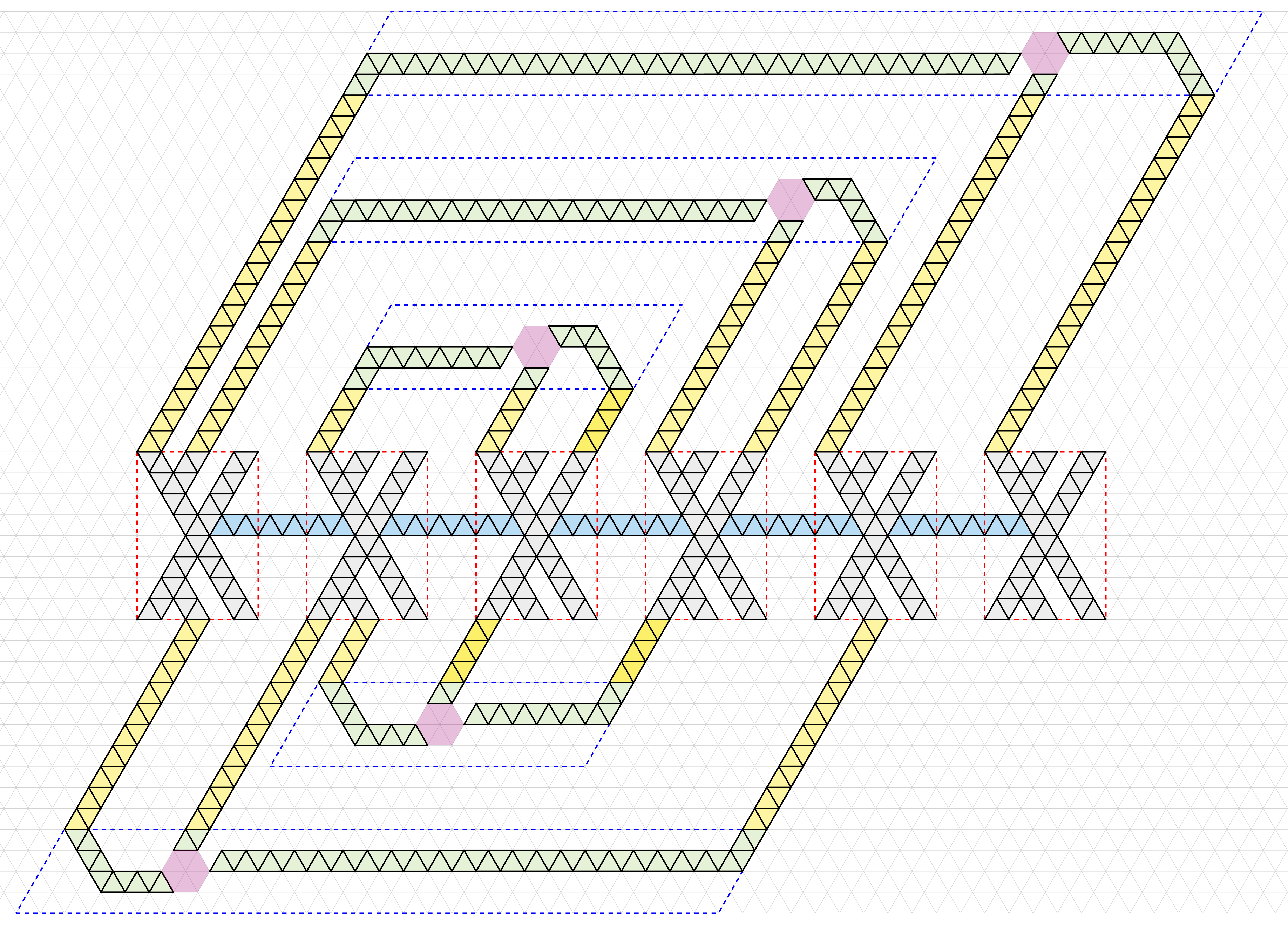}
		\caption{Overview of the reduction showing only the frame triangles of the involved gadgets.}
		\label{fig:overview-reduction}
	\end{figure}
	
	We define three main types of gadgets.
	A variable $v_i \in \phi$ is modeled by means of a gadget we call \emph{variable gadget}, a clause $c_j \in \phi$ by means of a gadget we call a \emph{clause gadget}, and an edge $(v_i,c_j) \in G_\phi$ by means of a gadget we call a \emph{$k$-transmission gadget}.
	We also define two auxiliary gadgets we call \emph{ladder gadgets} and  \emph{flag gadgets}.
	Ladder gadgets are used to provide structural rigidity to the reduction, and flag gadgets are used as an auxiliary tool in the construction of clause gadgets.
	To construct the gadgets we use the lengths $w_1:=1$, $w_2:=0.9$, $w_3:=0.2$, and $w_4:=1.61$.
	In the figures, we represent edges of length $w_1, w_2, w_3$, and $w_4$ by means of black, red, blue, and green segments, respectively.
	%
	The lengths $w_1$ and $w_2$ play a special role in the reduction.
	We use them to define two main types of triangles: Equilateral triangles with sides of length $w_1$, which we call \emph{frame triangles}, and flat isosceles triangles with base of length $w_1$ and two shorter sides of length $w_2$, which we call \emph{transmission triangles}.
	We refer to the subgraph of a gadget formed by its frame triangles as the \emph{frame} of the gadget.
	With the exception of clause gadgets, the frame of every gadget is a maximal outerplanar graph with a unique planar straight-line realization (up to rigid transformations); whereas the frame of a clause gadget consists of three maximal outerplanar graphs.
	%
	The variable, clause, and $k$-transmission gadgets are constructed in such a way that every frame triangle shares two of its sides each with a different transmission triangle.
	Since $\frac{w_1}{w_2}=\frac{1}{0.9}<\sqrt 3$, by \cref{lem:flat_in_equilateral} we have that, in any planar straight-line realization of $H_{\phi}$, two transmission triangles that share a side with a frame triangle cannot both be drawn inside such a frame triangle. 
	This property is exploited in order to propagate along the frame of the gadgets the truth values of each literal; this value is initially determined by the embedding of a transmission triangle inside one or the other of its two incident frame triangles of a variable gadget.
	A final common property of all our gadgets is a set of special edges we call \emph{attachment edges}, which are edges of frame triangles that are also edges of transmission triangles.
	%
	In the figures, attachment edges are represented as thick black segments.
	We use attachment edges to combine gadgets together so that the encoded truth values can be propagated along the resulting $2$-tree.
	We exploit properties D\ref{prop:tiling}-D\ref{prop:clause-distances} of $\Gamma^*_\phi$ to guarantee that, after combining all the gadgets, $H_\phi$ has a planar straight-line realization if and only if $\phi$ is satisfiable.
	
	An example of a planar straight-line realization of $H_\phi$ is shown in \cref{fig:overview-reduction}.
	This example illustrates a high-level sketch of the reduction from $\Gamma^*_{\phi}$ to $H_\phi$.
	The grid formed by lines with slope $0^\circ$, $60^\circ$, and $120^\circ$ is shown in light gray.
	For the sake of clarity, we only show the frames of the gadgets forming $H_\phi$.
	To illustrate the correspondence with the boxes of \cref{fig:planar-3-sat-modified} that represent the variables and clauses of $\phi$, we also show the bounding boxes of the variable gadgets (in dashed red) and the trapezoids enclosing the clause gadgets (in dashed blue).
	We next provide a high-level description of each gadget.
	%
	
	%
	
	\paragraph{The variable gadget.}
	
	Remember that $\delta_\phi$ denotes the maximum degree of $G_\phi$.
	The variable gadget provides $2\delta_\phi$ attachment edges. Further, its frame is a maximal outerplanar graph, hence it has a unique planar straight-line realization, up to rigid transformations. Let $\Gamma$ be a planar straight-line realization of the variable gadget, $\Gamma_f \subset \Gamma$ be the realization of the frame of the gadget in $\Gamma$, and $B_f$ be the bounding box of $\Gamma_f$.
	In \cref{fig:overview-reduction}, the realization $\Gamma_f$ is shaded gray and $B_f$ is shown with red dashed lines.
	Up to a rotation of~$\Gamma$, we have that $\delta_\phi$ attachment edges lie along the top side of $B_f$, say $s_t$, while the remaining $\delta_\phi$ attachment edges lie along the bottom side of $B_f$, say $s_b$.
	%
	%
	The crucial property of the variable gadget is that, in $\Gamma$, either {\em all} the transmission triangles incident to attachment edges along $s_t$ are drawn outside $\Gamma_f$, or {\em all} the transmission triangles incident to attachment edges along $s_b$ are drawn outside $\Gamma_f$.
	We exploit this property to encode the truth assignment of the variable modeled by the gadget.
	In particular, the purpose of the attachment edges on $s_t$ is to propagate the truth value to incident positive clauses, while the purpose of the attachment edges on $s_b$ is to propagate the negated truth value to incident negative clauses. 
	%
	
	
	\paragraph{The ladder gadget.}
	
	By Property D\ref{prop:variables-alignment}, the bottom sides of the boxes representing variables lie on the same horizontal grid line in $\Gamma^*_\phi$.
	Consider the order of appearance of the boxes while traversing such line from left to right.
	The ladder gadget consists of a sequence of frame triangles that form a maximal outerpath and that connect the frames of two consecutive variable gadgets; see the triangles shaded blue in \cref{fig:overview-reduction}.
	%
	The purpose of the ladder gadgets is to make the entire graph $H_{\phi}$ biconnected and to make the union of the frames of all the gadgets a maximal outerplanar graph, which has a unique planar straight-line realization, up to rigid transformations. This realization can then be navigated by the truth values of the variables. 
	
	%
	
	\paragraph{The $k$-transmission gadget.}
	
	Let $v_i$ and $c_j$ be a variable and a clause of $\phi$, respectively, such that $(v_i,c_j)$ is an edge of $G_\phi$.
	Let $\mathcal V_i$ denote the variable gadget modeling $v_i$ and $\mathcal C_j$ denote the clause gadget modeling $c_j$.
	The purpose of the transmission gadget is to ``transmit'' to $\mathcal C_j$ the truth value of $v_i$ corresponding to the realization of $\mathcal V_i$.
	The frame of the $k$-transmission gadget is a sequence of $k$ frame triangles that form a maximal outerpath; see the triangles shaded yellow in \cref{fig:overview-reduction}.
	The parameter $k$ coincides with twice the length of the segment representing the edge $(v_i,c_j)$ in $\Gamma^*_\phi$.
	The $k$-transmission gadget provides two attachment edges, each incident to one of the faces of the frame that correspond to the end-vertices of its dual path.
	To connect $\mathcal V_i$ to $\mathcal C_j$, one of these attachment edges is identified with an attachment edge of $\mathcal V_i$ and the second one is identified with an attachment edge of $\mathcal C_j$.
	%
	%
	
	\paragraph{The clause gadget.}
	
	Unlike the aforementioned gadgets, the clause gadget consists of three distinct connected components.
	Two of these components are formed by frame and transmission triangles.
	The third component is not only formed by frame and transmission triangles, but also by a special subgraph we call a \emph{flag gadget}.
	The frame of each component forms a maximal outerpath; see the triangles shaded light green in \cref{fig:overview-reduction}.
	The clause gadget contains six attachment edges.
	Three of them, that we call \emph{input attachment edges}, are used to combine the clause gadget with the three $k$-transmission gadgets modeling the edges of $G_\phi$ incident to the clause.
	The remaining three, that we call \emph{output attachment edges}, are used to model the logic of the clause; in particular, one of them connects the frame of a component to the flag gadget. 
	
	
	Consider a positive clause $c_j$ that contains the literal $v_i$.
	The clause $c_j$ is modeled in $H_\phi$ by a clause gadget $\mathcal C_j$ and is represented in $\Gamma^*_\phi$ by a trapezoid $t_j$.
	The edge $(v_i, c_j)$ of $G_\phi$ is modeled in $H_\phi$ by a $k$-transmission gadget $\mathcal T_{i,j}$ and is represented in $\Gamma^*_\phi$ by a line segment $s_{i,j}$.
	The gadgets $\mathcal C_j$ and $\mathcal T_{i,j}$ are combined together by means of an attachment edge $e$ of $\mathcal T_{i,j}$ that is identified with an input attachment edge of $\mathcal C_j$.
	The combination of $\mathcal C_j$ with $\mathcal T_{i,j}$ exploits the following geometric property.
	Consider the unique planar straight-line realization of the subgraph of $H_{\phi}$ that is the union of all the ladder gadgets and of the frames of all the variable and transmission gadgets.
	By Property D\ref{prop:edge-slopes} of $\Gamma^*_\phi$ and the fact that $k$ coincides with twice the length of $s_{i,j}$, we have that, in such a realization, the attachment edge $e$ is represented by a line segment lying along the bottom side of $t_j$.
	If $c_j$ is instead negative, then $e$ lies along the top side of $t_j$.
	
	By the geometric property described above, the input attachment edges of the components of $\mathcal C_j$ are represented by line segments that lie either all on the bottom or all on the top side of $t_j$.
	The frames of the components of $\mathcal C_j$ can thus be defined so that, in any planar straight-line realization of $\mathcal C_j$, we have that (i) all the frames are entirely contained in $t_j$ and (ii) the output attachment edges are represented by line segments that are ``close'' to each other.
	In \cref{fig:overview-reduction}, we emphasize this closeness by depicting a pink-shaded hexagon whose vertices coincide with some of the endpoints of these line segments.
	
	Consider the transmission triangles incident to the output attachment edges of the components of $\mathcal C_j$.
	The clause gadget admits a planar straight-line realization {\em if and only if} at least one of these transmission triangles is drawn inside its incident frame triangle.
	Let $C$ be a component of $\mathcal C_j$.
	Assume again that the input attachment edge of $C$ is shared with a $k$-transmission gadget $\mathcal T_{i,j}$ that connects $\mathcal C_j$ to the variable gadget representing a variable $v_i$.
	Our gadgets guarantee that, in any planar straight-line realization of $H_\phi$, the transmission triangle incident to the output attachment edge of $C$ can be drawn inside its adjacent frame triangle only if the gadget modeling $v_i$ represents the value $\texttt{False}$ (resp.\ $\texttt{True}$) and $c_j$ is positive (resp.\ negative). Thus, if all three variables appearing in $c_j$ represent the value $\texttt{False}$ (resp.\ $\texttt{True}$) and $c_j$ is positive (resp.\ negative), then $H_\phi$ admits no planar straight-line realization. Conversely, if at least one of the variables appearing in $c_j$ represents the value $\texttt{True}$ (resp.\ $\texttt{False}$) and $c_i$ is positive (resp.\ negative), then the transmission triangle of the corresponding component of $\mathcal C_j$ can be drawn inside its incident frame triangle so that $\mathcal C_j$ admits a planar straight-line realization.

	\remove{
		We will exploit the following properties of the gadgets for variables and the gadgets for the edges. We refer to the latter as {\em transmission gadgets}.
		
		\begin{property}\label{prop:frame-outerplanar}
			The transmission and the variable gadgets satisfy the following properties:
			\begin{enumerate}
				\item their frame induce a maximal outerplanar graph,
				\item the transmission triangle incident to one attachment edge of the gadget is drawn inside its incident frame triangle {\em if and only if} the the transmission triangle incident to the other attachment edge of the gadget is draw outside its incident frame triangle.
			\end{enumerate}
		\end{property}
		
		By \cref{prop:frame-outerplanar} and since no frame triangle can drawn inside any other frame triangle, we have the following.
		
		\begin{remark}\label{remark:frame-outerplanar-drawing}
			In any planar straight-line realization of a transmission, split, and variable gadget $G$, the drawing of the subgraph of $G$ induced by the frame triangles is outerplanar.
		\end{remark}
	}
	
	\subsection{Description of the gadgets}
	\label{sec:np-hardness_gadgets}
	
	We now present complete constructions for the following gadgets: The \emph{transmission gadget}, the \emph{split gadget}, the \emph{variable gadget}, the \emph{flag gadget}, and the \emph{clause gadget}.
	We will later show how to precisely combine these gadgets to obtain $H_\phi$.
	
	In the following, whenever we combine two gadgets $G_1$ and $G_2$, we proceed as follows.
	First, we identify one attachment edge of $G_1$ with one attachment edge of $G_2$ to obtain a single edge $e$ that belongs to both $G_1$ and $G_2$.
	Second, we remove the degree-$2$ vertex of any of the two transmission triangles sharing the edge $e$, together with the two edges incident to such vertex.
	%
	This operation always results in a $2$-tree, as a consequence of the following lemma.
	
	\begin{lemma}\label{lem:two-2-trees-merge}
		Let $G_1$ and $G_2$ be two $2$-trees, and let $e_1$ be an edge of $G_1$ and $e_2$ be an edge of $G_2$.
		The graph~$G$ obtained from $G_1$ and $G_2$ by identifying $e_1$ and $e_2$ is a $2$-tree.
	\end{lemma}
	
	\begin{proof}
		We prove the statement by induction on $|V(G_2)|$.
		
		In the base case, $|V(G_2)|=3$, that is, $G_2$ is a $3$-cycle.
		Then $G$ is obtained by adding to $G_1$ a vertex $v$ and two edges $(v,u)$ and $(v,w)$, where $u$ and $w$ are the end-vertices of $e_1$.
		Hence $G$ is a $2$-tree.
		
		In the inductive case, let $v$ be a degree-$2$ vertex of $G_2$ that is different from both the end-vertices of $e_2$.
		This vertex exists by Property~(P\ref{pr:non-adjacent-degree-2}) of a $2$-tree.
		Let $u$ and $w$ be the neighbors of $v$ in $G_2$.
		By Property~(P\ref{pr:degree-2-adjacent}) of a $2$-tree, we have that $(u,w)$ is an edge of $G_2$.
		Let $G'_2$ be the graph obtained from $G_2$ by removing $v$ and its incident edges $(v,u)$ and $(v,w)$; note that $e_2$ is an edge of $G'_2$.
		By induction, the graph $G'$ obtained from $G_1$ and $G'_2$ by identifying $e_1$ and $e_2$ is a $2$-tree.
		Then $G$ is obtained by adding to $G'$ the vertex $v$ and the edges $(v,u)$ and $(v,w)$; since $G'$ contains the edge $(u,w)$, it follows that $G$ is a $2$-tree.
	\end{proof}
	
	\paragraph{The $k$-transmission gadget.}
	
	\begin{figure}[bt!]
		\centering
		
		\subfloat[\label{fig:transmission_gadget-a}]{
			\includegraphics[page=1,height=.35\textwidth]{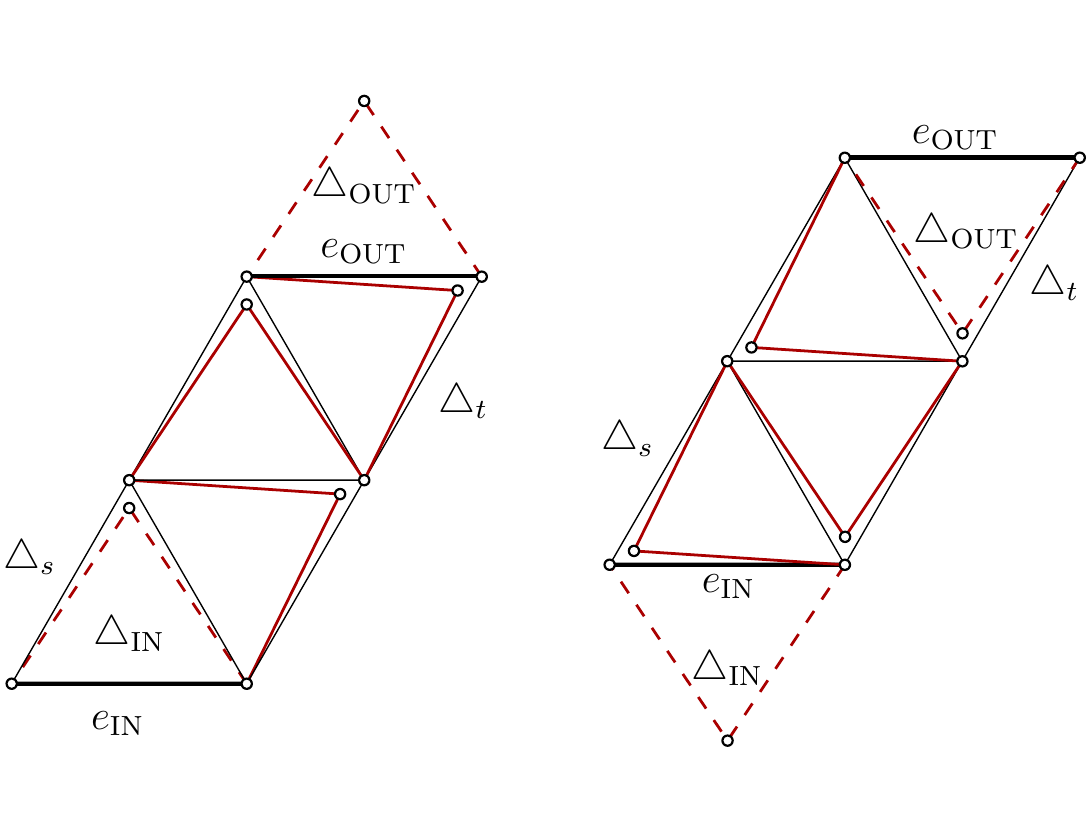}
		}
		\hfil
		\subfloat[\label{fig:transmission_gadget-b}]{
			\includegraphics[page=2,height=.35\textwidth]{np-hardness__transmission_gadget}
		}
		
		\caption
		{
			Illustration of two $k$-transmission gadgets (in (a) and in (b)), differing from one another for the choice of the attachment edge belonging to $\triangle_t$.
			For each of the two $k$-transmission gadgets, both a planar straight-line realization in which $\triangle_{\textrm{IN}}$ is drawn inside $\triangle_s$ and a planar straight-line realization in which $\triangle_{\textrm{OUT}}$ is draw inside $\triangle_t$ are shown.
			In these illustrations, we have $k=4$.
		}
		\label{fig:transmission_gadget}
	\end{figure}
	For any positive even integer $k$, a \emph{$k$-transmission gadget} consists of a sequence $\triangle_s=\triangle_1,\dots,\triangle_k=\triangle_t$ of $k$ frame triangles, in which each triangle shares exactly one edge with its successor; refer to \cref{fig:transmission_gadget}.
	For each pair of consecutive frame triangles, we insert a transmission triangle (solid red triangles in \cref{fig:transmission_gadget}) whose base edge we identify with the common edge between the frame triangles.
	When the value $k$ is not relevant, we refer to a $k$-transmission gadget simply as a ``transmission gadget''.
	A $k$-transmission gadget provides two attachment edges.
	The first attachment edge, which we denote $e_\textrm{IN}$, can be selected as any of the two edges of $\triangle_{s}$ that is not shared with its successor frame triangle.
	Similarly, the second attachment edge, which we denote $e_\textrm{OUT}$, can be selected as any of the two edges of $\triangle_{t}$ that is not shared with its predecessor frame triangle.
	This yields four possible $k$-transmission gadgets.
	We complete the construction of a transmission gadget by inserting a transmission triangle $\triangle_{\textrm{IN}}$ whose base we identify with $e_\textrm{IN}$, and a transmission triangle $\triangle_\textrm{OUT}$ whose base we identify with $e_\textrm{OUT}$; see the dashed triangles in~\cref{fig:transmission_gadget}.
	
	Clearly, a transmission gadget is a $2$-tree and its frame is a maximal outerpath.
	In any planar straight-line realization of a transmission gadget, its attachment edges are either parallel or not.
	In the first case, we say that the transmission gadget is \emph{straight}, whereas in the second case we say that it is \emph{turning}; the transmission gadgets in \cref{fig:transmission_gadget-a} and  \cref{fig:transmission_gadget-b} are straight and turning, respectively.
	We have the following lemma.
	
	
	\begin{lemma}\label{lem:transmission-gadget-behaviour}
		In any planar straight-line realization of a transmission gadget, if $\triangle_\textrm{IN}$ is drawn inside $\triangle_s$, then $\triangle_\textrm{OUT}$ is drawn outside $\triangle_t$.
		Further, if $\triangle_\textrm{OUT}$ is drawn inside~$\triangle_t$, then
		$\triangle_\textrm{IN}$ is drawn outside~$\triangle_s$.
	\end{lemma}
	
	\begin{proof}
		Consider the sequence $\triangle_1,\dots,\triangle_k$ of frame triangles of a transmission gadget, where $\triangle_s = \triangle_1$ and $\triangle_t = \triangle_k$.
		For $i=1,\dots,k-1$, the triangles $\triangle_i$ and $\triangle_{i+1}$ share a single edge, which is the base of a transmission triangle we denote with $\triangle^T_i$.
		
		We only prove the first part of the statement, since the second part is symmetric.
		Consider any planar straight-line realization of a transmission gadget in which $\triangle_\textrm{IN}$ is drawn inside $\triangle_s$.
		Since $\frac{w_1}{w_2}=\frac{1}{0.9}<\sqrt 3$, by \cref{lem:flat_in_equilateral} we have that $\triangle_\textrm{IN}$ and $\triangle^T_1$ are not both drawn inside $\triangle_1 = \triangle_s$; therefore $\triangle^T_1$ is drawn inside $\triangle_2$.
		Similarly, since $\triangle^T_i$ is drawn inside $\triangle_{i+1}$, then $\triangle^T_{i+1}$ is drawn inside $\triangle_{i+2}$ for $i=1,\dots,k-2$.
		Finally, since $\triangle^T_{k-1}$ is drawn inside $\triangle_k = \triangle_t$, we have that $\triangle_\textrm{OUT}$ is drawn outside $\triangle_t$.
	\end{proof}
	
	
	\paragraph{The split gadget.}
	
	\begin{figure}[h!]
		\centering
		\subfloat[\label{fig:split_gadget-a}]{
			\includegraphics[page=3,height=.3\textwidth]{np-hardness__split_gadget}
		}
		\hfil
		\subfloat[\label{fig:split_gadget-b}]{
			\includegraphics[page=4,height=.3\textwidth]{np-hardness__split_gadget}
		}
		\\
		\subfloat[\label{fig:split_gadget-c}]{
			\includegraphics[page=1,height=.3\textwidth]{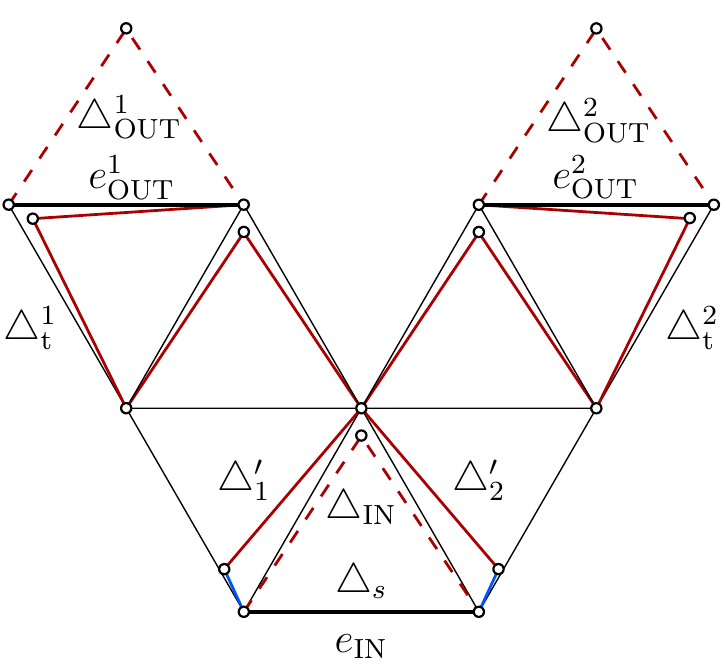}
		}
		\hfil
		\subfloat[\label{fig:split_gadget-d}]{
			\includegraphics[page=2,height=.3\textwidth]{np-hardness__split_gadget}
		}
		\caption
		{
			Split gadget construction.
			(a) The $4$-transmission gadgets $T^1$ and $T^2$.
			(b) The gadgets $T^1$ and $T^2$ after replacing $\triangle'_1$ and $\triangle'_2$ with scalene triangles.
			(c) The split gadget with $\triangle_{\textrm{IN}}$ drawn inside $\triangle_s$.
			(d) The split gadget with $\triangle^i_{\textrm{OUT}}$ drawn inside $\triangle^i_t$, for $i=1,2$.
		}
		\label{fig:split_gadget}
	\end{figure}
	
	Let $T^1$ and $T^2$ be two straight $4$-transmission gadgets.
	The split gadget can be obtained from $T^1$ and $T^2$ as follows.
	For $i = 1,2$, let us denote the triangles $\triangle_\textrm{IN}$, $\triangle_\textrm{OUT}$, $\triangle_s$, and $\triangle_t$ of $T^i$ as $\triangle^i_\textrm{IN}$, $\triangle^i_\textrm{OUT}$, $\triangle^i_s$, and $\triangle^i_t$, respectively, and the attachment edges $e_\textrm{IN}$ and $e_\textrm{OUT}$ of $T^i$ as $e^i_\textrm{IN}$ and  $e^i_\textrm{OUT}$, respectively; see \cref{fig:split_gadget-a}.
	For $i=1,2$, let $\triangle'_i$ be the transmission triangle of $T^i$ that is incident to the frame triangle $\triangle^i_{s}$ and different from $\triangle^i_\textrm{IN}$.
	First, we replace $\triangle'_i$ with a scalene triangle, which we still refer to as $\triangle'_i$, whose longest side is shared with $\triangle^i_{s}$ and whose two shorter sides have length $w_2$ and $w_3$; refer to \cref{fig:split_gadget-b}. 
	After this replacement, we still denote the ``modified'' transmission gadgets as $T^1$ and $T^2$. 
	%
	Finally, we identify the frame triangles $\triangle^1_{s}$ and $\triangle^2_{s}$ to obtain a single frame triangle we denote with $\triangle_s$, and we identify the transmission triangles $\triangle^1_\textrm{IN}$ and $\triangle^2_\textrm{IN}$ to obtain a single transmission triangle we denote with $\triangle_\textrm{IN}$; refer to \cref{fig:split_gadget-c,fig:split_gadget-d}.
	By \cref{lem:two-2-trees-merge} and the fact that the transmission gadgets are $2$-trees, we have that the split gadget is a $2$-tree and its frame is a maximal outerpath.
	
	The split gadget provides three attachment edges. Namely, 
	the attachment edge $e_\textrm{IN}$ incident to $\triangle_\textrm{IN}$ (obtained by identifying $e^1_\textrm{IN}$ and $e^2_\textrm{IN}$), the attachment edge $e^1_\textrm{OUT}$ incident to $\triangle^1_\textrm{OUT}$, and the attachment edge $e^2_\textrm{OUT}$ incident to $\triangle^2_\textrm{OUT}$. 
	%
	We have the following property.
	
	\begin{property}\label{prop:scalene-conflict}
		The scalene triangles $\triangle'_1$ and $\triangle'_2$ satisfy the following properties.
		
		\begin{enumerate}[(a)]
			\item
			The triangles $\triangle'_1$ and $\triangle'_2$ can be drawn together inside $\triangle_{s}$ without intersecting each other.
			\item
			Neither $\triangle'_1$ nor $\triangle'_2$ can be drawn together with $\triangle_\textrm{IN}$ inside  $\triangle_{s}$ without intersecting each other.
		\end{enumerate}
	\end{property}
	
	\begin{proof}
		Statement (a) follows from the fact that the angles of $\triangle'_1$ and $\triangle'_2$ incident to the vertex they share sum up to less than $60^\circ$.
		Indeed, by the law of the cosines, such angles are equal to $\arccos\left({\frac{w_1^2+w_2^2-w_3^2}{2w_1w_2}}\right)<\arccos(0.983)<11^\circ$.
		
		Statement (b) follows from the fact that, for $i=1,2$, the angles of $\triangle'_i$ and $\triangle_{\textrm{IN}}$ incident to the vertex they share sum up to more than $60^\circ$.
		Indeed, by the law of the cosines, such angles are equal to $\arccos\left({\frac{w_1^2+w_3^2-w_2^2}{2w_1w_3}}\right)>\arccos(0.57)>54^\circ$ and $\arccos\left({\frac{w_1}{2w_2}}\right)>\arccos(0.56)>55^\circ$, respectively.
	\end{proof}
	
	Using \cref{prop:scalene-conflict}, the proof of \cref{lem:transmission-gadget-behaviour} can be easily adapted to obtain the following.
	
	\begin{lemma}\label{lem:split-gadget-behaviour}
		In any planar straight-line realization of the split gadget, if
		$\triangle_\textrm{IN}$ is drawn inside $\triangle_s$, then
		$\triangle^1_\textrm{OUT}$ is drawn outside $\triangle^1_t$ and $\triangle^2_\textrm{OUT}$ is drawn outside $\triangle^2_t$. Further, if 
		$\triangle^1_\textrm{OUT}$ is drawn inside~$\triangle^1_t$ or $\triangle^2_\textrm{OUT}$ is drawn inside~$\triangle^2_t$, then
		$\triangle_\textrm{IN}$ is drawn outside~$\triangle_s$.
	\end{lemma}
	
	\paragraph{The variable gadget.} 
	%
	
	%
	%
	Remember that $\delta_\phi$ denotes the maximum degree of $G_\phi$.
	Let $S_1,S_2,\dots,S_{2\delta_\phi-2}$ be $2\delta_\phi-2$ split gadgets, and $T_2,T_3,\dots,T_{2\delta_\phi-3}$ be $2\delta_\phi-4$ straight transmission gadgets such that, for $i=2,\dots,\delta_{\phi}-1$, the gadgets $T_i$ and $T_{2\delta_{\phi}-i-1}$ are $4(i-1)$-transmission gadgets.
	The variable gadget can be obtained by combining these gadgets as follows; refer to \cref{fig:variable_gadget-schema}.
	
	\begin{figure}
		\centering
		{\includegraphics[page=1,width=0.37\textwidth]{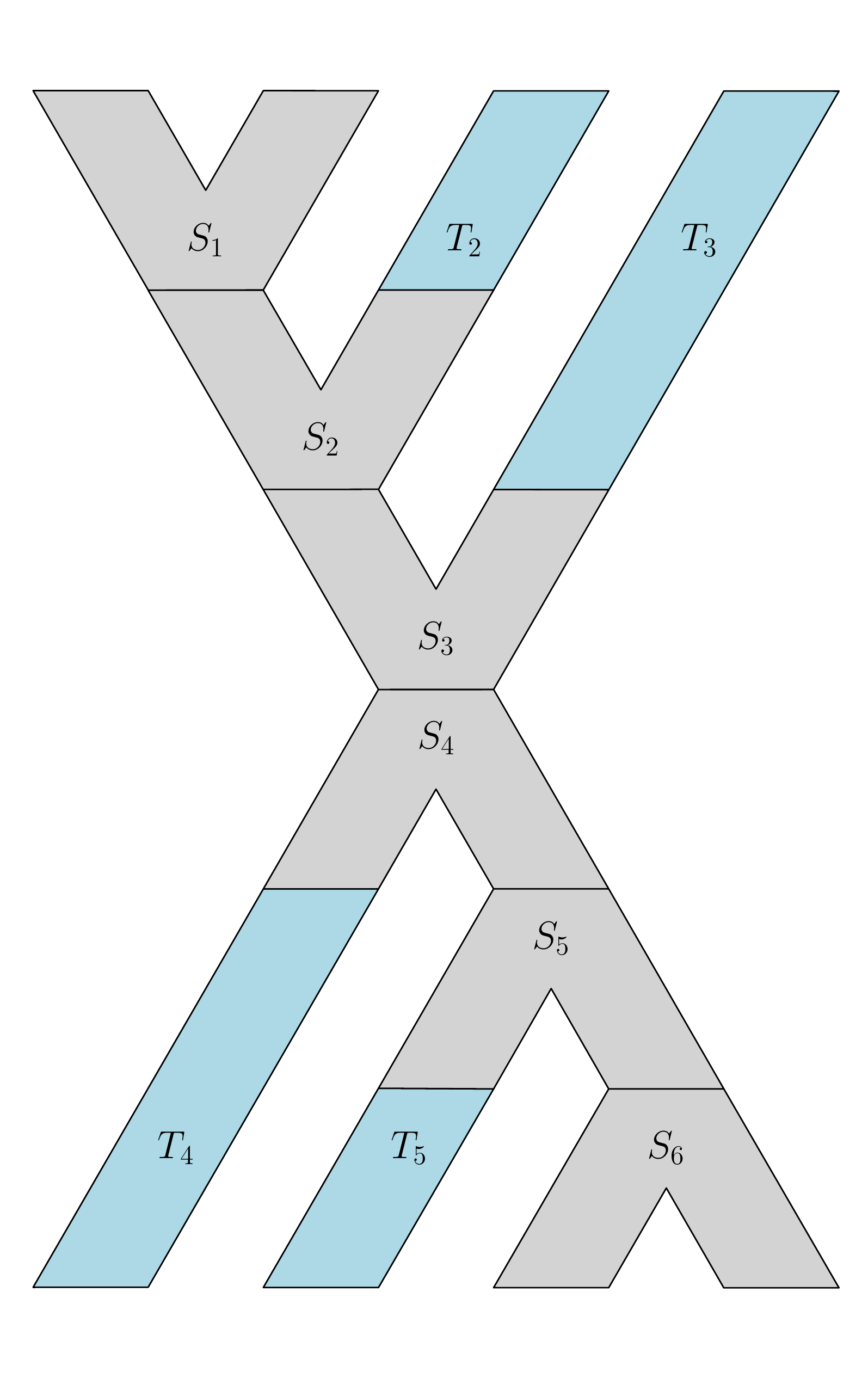}}
		\caption{Schema of the variable gadget when the maximum degree $\delta_\phi$ of~$G_\phi$~is~$4$.}
		\label{fig:variable_gadget-schema}
	\end{figure}
	
	\begin{itemize}
		\item
		
		For $i=1,\dots,\delta_\phi-2$, we combine $S_i$ and $S_{i+1}$ by identifying the attachment edge $e_\textrm{IN}$ of $S_i$ with the attachment edge $e^1_\textrm{OUT}$ of $S_{i+1}$;
		
		\item
		
		we combine $S_{\delta_\phi -1}$ and $S_{\delta_\phi}$ by identifying the attachment edge $e_\textrm{IN}$ of $S_{\delta_\phi -1}$ with the attachment edge $e_\textrm{IN}$ of $S_{\delta_\phi}$;
		
		\item
		
		for $i=\delta_\phi,\dots,2\delta_\phi-3$, we combine $S_i$ and $S_{i+1}$ by identifying the attachment edge $e^1_\textrm{OUT}$ of $S_i$ with the attachment edge $e_\textrm{IN}$ of $S_{i+1}$; and
		
		\item
		
		for $i=2,\dots,2\delta_\phi-3$, we combine $S_i$ and $T_{i}$ by identifying the attachment edge $e^2_\textrm{OUT}$ of $S_i$ with the attachment edge $e_\textrm{IN}$ of $T_{i}$.
		
	\end{itemize}
	
	By \cref{lem:two-2-trees-merge} and the fact that transmission and split gadgets are $2$-trees, we have that the variable gadget is a $2$-tree and its frame is a maximal outerplanar graph.
	
	
	The variable gadget provides $2\delta_\phi$ attachment edges; refer to \cref{fig:variable_gadget}.
	Namely:
	\begin{itemize}
		\item
		
		the attachment edges $e^1_{\textrm{OUT}}$ and $e^2_{\textrm{OUT}}$ of $S_1$,
		
		\item
		
		for $i=2,\dots,2\delta_\phi-3$, the attachment edge $e_\textrm{OUT}$ of $T_i$, which we denote as $e^{i+1}_\textrm{OUT}$, and
		
		\item
		
		the attachment edges $e^1_\textrm{OUT}$ and $e^2_\textrm{OUT}$ of $S_{2\delta_\phi-2}$, which we denote as $e^{2\delta_\phi-1}_\textrm{OUT}$ and $e^{2\delta_\phi}_\textrm{OUT}$, respectively.
		
	\end{itemize}  
	For $i=1,\dots,2\delta_\phi$, we also denote by $\triangle^i_\textrm{OUT}$ and by $\triangle^i_t$ the transmission triangle and the frame triangle incident to $e^i_\textrm{OUT}$, respectively.
	Finally, let $\triangle_\textrm{IN}$ be the transmission triangle that is incident to the attachment edge shared by the split gadgets $S_{\delta_\phi-1}$ and $S_{\delta_\phi}$.
	We call $\triangle_\textrm{IN}$ the \emph{truth-assignment triangle}.
	Intuitively, the truth value associated with a planar straight-line realization of the variable gadget, will depend on whether $\triangle_\textrm{IN}$ lies inside the frame triangle $\triangle_s$ of $S_{\delta_\phi-1}$ or inside the frame triangle $\triangle_s$ of $S_{\delta_\phi}$ in such a realization.
	We show the following simple and yet useful geometric property of the variable gadget.
	
	\vspace{10mm}
	\begin{property}
		In any planar straight-line realization of the variable gadget, up to a rotation, the frame lies in an axis-aligned box $B$ of width $w_1(2\delta_\phi-1)$ and height $2 \sqrt{3} w_1(\delta_\phi-1)$.
		Furthermore, the attachment edges $e^1_\textrm{OUT},\dots,e^{\delta_\phi}_\textrm{OUT}$ lie along the top side of $B$ at distance $w_1$ from each other, and the attachment edges
		$e^{\delta_\phi+1}_\textrm{OUT},\dots,e^{2\delta_\phi}_\textrm{OUT}$ lie along the bottom side of $B$ at distance $w_1$ from each other.
	\end{property}

	We have the following lemma. Refer to \cref{fig:variable_gadget}.
	
	\begin{figure}[tb!]
		\centering
		\subfloat[\label{fig:variable-b}]
		{\includegraphics[page=8,width=.45\textwidth]{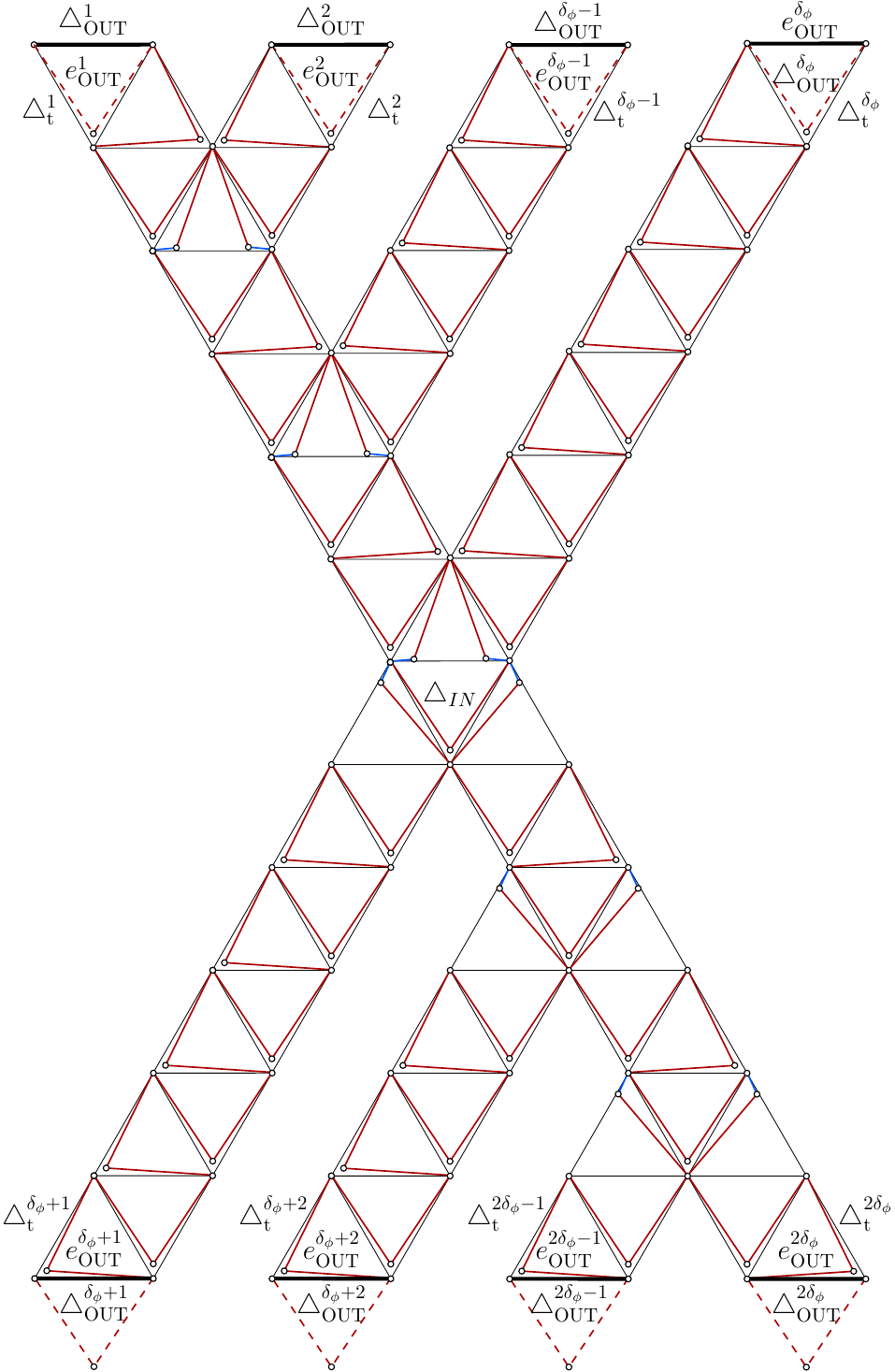}}
		\hfil
		\subfloat[\label{fig:variable-c}]
		{\includegraphics[page=3,width=.45\textwidth]{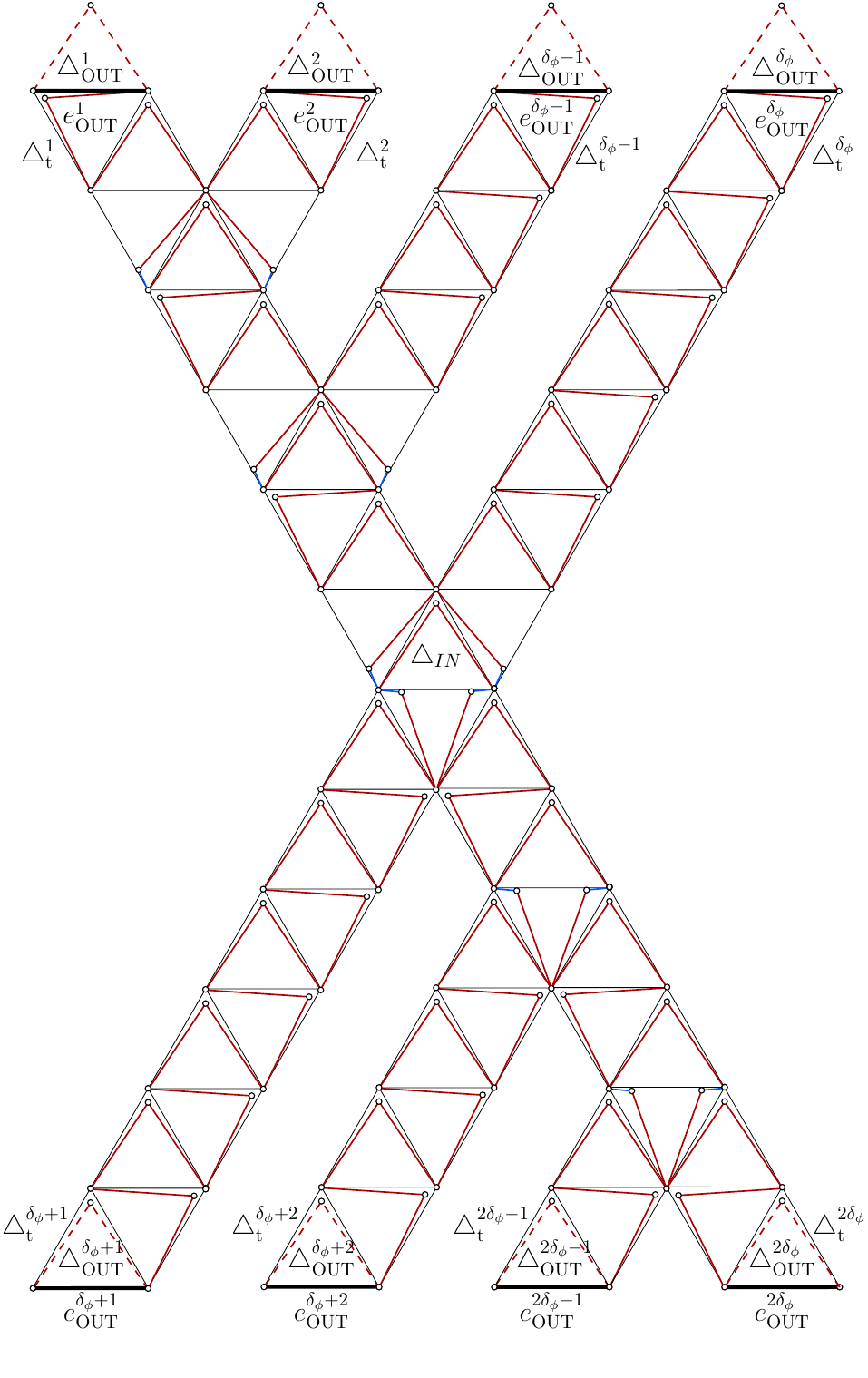}}
		\caption{The variable gadget when the maximum degree $\delta_\phi$
			of $G_\phi$ is $4$. In the logic of the reduction, the realization on the left corresponds to an assignment of \texttt{True} to the variable represented by the variable gadget, while the realization on the right corresponds to an assignment of \texttt{False} to the variable represented by the variable gadget.}
		\label{fig:variable_gadget}
	\end{figure}
	
	\begin{lemma}\label{lem:variable-gadget-behaviour}
		In any planar straight-line realization of the variable gadget, if the truth-assignment triangle $\triangle_\textrm{IN}$ is drawn inside the frame triangle $\triangle_s$ of $S_{\delta_\phi-1}$ (resp. of $S_{\delta_\phi}$), then, for $i=1,\dots,\delta_\phi$ (resp.\ for $i=\delta_{\phi}+1,\dots,2\delta_\phi$), the transmission triangle $\triangle^i_{\textrm{OUT}}$ is drawn outside the frame triangle $\triangle^i_{t}$.
		Furthermore, if there is some $i\in \{1,\dots,\delta_\phi\}$ (resp.\ some $i\in \{\delta_\phi+1,\dots,2\delta_\phi\}$) for which the transmission triangle $\triangle^i_{\textrm{OUT}}$ is drawn inside the frame triangle $\triangle^i_{t}$, then $\triangle_\textrm{IN}$ is drawn inside the frame triangle $\triangle_s$ of $S_{\delta_\phi}$ (resp. of $S_{\delta_\phi-1}$).
	\end{lemma}
	
	\begin{proof}
		We start with the first part of the statement.
		Let $\Gamma$ be a planar straight-line realization of the variable gadget.
		Suppose the truth-assignment triangle $\triangle_\textrm{IN}$ is drawn inside the frame triangle $\triangle_s$ of $S_{\delta_\phi-1}$ in $\Gamma$.
		We show that, for $i=1,\dots,\delta_\phi$, the transmission triangle $\triangle^i_{\textrm{OUT}}$ is drawn outside the frame triangle $\triangle^i_{t}$ in $\Gamma$.
		By a symmetric argument we can proof that, if $\triangle_\textrm{IN}$ is drawn inside the frame triangle $\triangle_s$ of $S_{\delta_\phi}$ in $\Gamma$, then, for $i=\delta_\phi,\dots,2\delta_\phi$, the transmission triangle $\triangle^i_{\textrm{OUT}}$ is drawn outside the frame triangle $\triangle^i_{t}$ in $\Gamma$.
		We have the following claims.
		
		\begin{claimx}\label{claim:splits}
			Consider any two split gadgets $S_i$ and $S_{i-1}$, for $i=\delta_\phi-1,\dots, 2$.
			If the transmission triangle $\triangle^1_{\textrm{OUT}}$ of $S_i$ is drawn outside the frame triangle $\triangle^1_t$ of $S_i$, then the transmission triangle $\triangle^1_{\textrm{OUT}}$ of $S_{i-1}$ is drawn outside the frame triangle $\triangle^1_t$ of $S_{i-1}$.
		\end{claimx}
		
		\begin{proof}
			Observe that, if the transmission triangle $\triangle^1_{\textrm{OUT}}$ of $S_{i}$ is drawn outside the frame triangle $\triangle^1_t$ of $S_{i}$, then it lies inside the frame triangle $\triangle_s$ of $S_{i-1}$.
			Since the transmission triangle $\triangle^1_{\textrm{OUT}}$ is in fact a copy of the transmission triangle $\triangle_{IN}$ of $S_{i-1}$ (which has been removed when combining $S_i$ and $S_{i-1}$), by \cref{lem:split-gadget-behaviour} we have that the transmission triangle $\triangle^1_{\textrm{OUT}}$ of $S_{i-1}$ is drawn outside the frame triangle~$\triangle^1_t$ of~$S_{i-1}$.
		\end{proof}
		
		\begin{claimx}\label{claim:transmission}
			Consider any split gadget $S_i$ and any transmission gadget $T_i$, for $i=\delta_\phi-1,\dots, 2$.
			If the transmission triangle $\triangle^2_{\textrm{OUT}}$ of $S_{i}$ is drawn outside the frame triangle $\triangle^2_t$ of $S_{i}$, then the transmission triangle $\triangle_{\textrm{OUT}}$ of $T_{i}$ is drawn outside the frame triangle $\triangle_t$ of $T_{i}$.
		\end{claimx}
		
		\begin{proof}
			Observe that, if the transmission triangle $\triangle^2_{\textrm{OUT}}$ of $S_{i}$ is drawn outside the frame triangle $\triangle^2_t$ of $S_{i}$, then it lies inside the frame triangle $\triangle_s$ of $T_{i}$.
			Since $\triangle^2_{\textrm{OUT}}$ is in fact a copy of the transmission triangle $\triangle_{IN}$ of $T_{i}$ (which has been removed when combining $S_i$ and $T_{i}$), by \cref{lem:transmission-gadget-behaviour} we have that the transmission triangle $\triangle_{\textrm{OUT}}$ of $T_{i}$ is drawn outside the frame triangle $\triangle_t$ of~$T_{i}$.
		\end{proof}
		
		Since the truth-assignment triangle $\triangle_\textrm{IN}$ lies inside the frame triangle $\triangle_s$ of $S_{\delta_\phi-1}$, then by \cref{lem:split-gadget-behaviour}, the transmission triangles $\triangle^1_{\textrm{OUT}}$ and $\triangle^2_{\textrm{OUT}}$ of $S_{\delta_\phi-1}$ are drawn outside the frame triangles $\triangle^1_t$ and $\triangle^2_t$ of $S_{\delta_\phi-1}$, respectively.
		Hence, in the variable gadget, the transmission triangle $\triangle^i_{\textrm{OUT}}$ is drawn outside the frame triangle $\triangle^i_t$; this follows from \cref{claim:splits} for $i=1,2$, and from \cref{claim:transmission} for $i=3,\dots,\delta_\phi$.
		
		We now prove the second part of the statement.
		Let $\Gamma$ be a planar straight-line realization of the variable gadget.
		Suppose that, for some $j\in \{1,\ldots,\delta_\phi\}$, there is a transmission triangle $\triangle^j_{\textrm{OUT}}$ of the variable gadget drawn inside the frame triangle $\triangle^j_t$ of the variable gadget.
		We show that the truth-assignment triangle $\triangle_\textrm{IN}$ is drawn inside the frame triangle $\triangle_s$ of $S_{\delta_\phi}$ in $\Gamma$.
		By a symmetric argument we can proof that, if there is some $j\in \{\delta_\phi+1,\dots,2\delta_\phi\}$ for which a transmission triangle $\triangle^j_{\textrm{OUT}}$ of the variable gadget is drawn inside the frame triangle $\triangle^j_t$ of the variable gadget, then the truth-assignment triangle $\triangle_\textrm{IN}$ is drawn inside the frame triangle $\triangle_s$ of $S_{\delta_{\phi-1}}$ in $\Gamma$.
		We have the following two claims, whose proofs are symmetric to the proofs of \cref{claim:splits,claim:transmission}.
		
		\begin{claimx}\label{claim:splits_inverse_direction}
			Consider any two split gadgets $S_i$ and $S_{i-1}$, for $i=\delta_\phi-1,\dots, 2$.
			If the transmission triangle $\triangle^1_{\textrm{OUT}}$ of $S_{i-1}$ is drawn inside the frame triangle $\triangle^1_t$ of $S_{i-1}$, then the transmission triangle $\triangle^1_{\textrm{OUT}}$ of $S_i$ is drawn inside the frame triangle $\triangle^1_t$ of $S_i$.
		\end{claimx}
		
		\begin{claimx}\label{claim:transmission_inverse_direction}
			Consider any split gadget $S_i$ and any transmission gadget $T_i$, for $i=\delta_\phi-1,\dots, 2$.
			If the transmission triangle $\triangle_{\textrm{OUT}}$ of $T_i$ is drawn inside the frame triangle $\triangle_t$ of $T_i$, then the transmission triangle $\triangle^2_{\textrm{OUT}}$ of $S_i$ is drawn inside the frame triangle $\triangle^2_t$ of $S_i$.
		\end{claimx}
		
		If $1\leq j \leq 2$, then the transmission triangle $\triangle^j_{\textrm{OUT}}$ of the variable gadget is actually a transmission triangle of the split gadget $S_1$.
		Hence, by \cref{claim:splits_inverse_direction}, we have that the transmission triangle $\triangle^1_{\textrm{OUT}}$ of $S_{\delta_{\phi}-1}$ is drawn inside the frame triangle $\triangle^1_t$ of $S_{\delta_\phi-1}$.
		If instead $3 \leq j \leq \delta_\phi$, then the transmission triangle $\triangle^j_{\textrm{OUT}}$ of the variable gadget is actually a transmission triangle of a transmission gadget.
		Hence, by \cref{claim:transmission_inverse_direction}, we have that the transmission triangle $\triangle^2_{\textrm{OUT}}$ of $S_{j}$ is drawn inside the frame triangle $\triangle^2_t$ of $S_{j}$.
		This fact and again \cref{claim:splits_inverse_direction} imply that either the transmission triangle $\triangle^1_{\textrm{OUT}}$ of $S_{\delta_{\phi}-1}$ is drawn inside the frame triangle $\triangle^1_t$ of $S_{\delta_\phi-1}$ (which happens when $j < \delta_\phi$), or the transmission triangle $\triangle^2_{\textrm{OUT}}$ of $S_{\delta_{\phi}-1}$ is drawn inside the frame triangle $\triangle^2_t$ of $S_{{\delta_\phi}-1}$ (which happens when $j= \delta_\phi$).
		In all cases, by \cref{lem:split-gadget-behaviour}, the truth-assignment triangle $\triangle_{\textrm{IN}}$ is drawn inside the frame triangle $\triangle_s$ of~$S_{\delta_\phi}$~in~$\Gamma$.
	\end{proof}

	\begin{figure}[tb!]
		\centering
		\includegraphics[page=1,width=.33\textwidth]{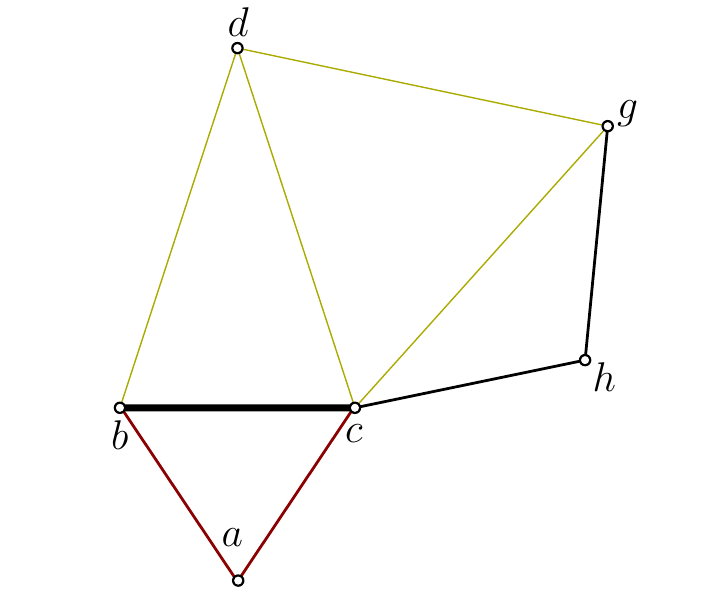}
		\caption{The triangles $(a,b,c)$, $(d,b,c)$, $(g,d,c)$, and $(g,h,c)$ of the flag gadget.}
		\label{fig:flag_gadget-main}
	\end{figure}
	
	\paragraph{The flag gadget.}
	The flag gadget consists of the following triangles (see \cref{fig:flag_gadget,fig:flag_gadget-main}):
	\begin{enumerate}[(i)]
		\item
		a transmission triangle $(a,b,c)$ with base $(b,c)$,
		\item
		a tall isosceles triangle $(d,b,c)$ with base $(b,c)$ of length $w_1$ and two longer sides of length~$w_4$,
		\item
		a flat isosceles triangle $(f,d,c)$ with base $(d,c)$ of length $w_4$ and two shorter sides of length $w_2$, 
		\item
		an equilateral triangle $(g,d,c)$ with sides of length $w_4$,
		\item
		a flat isosceles triangle $(h,g,c)$ with base $(g,c)$ of length $w_4$ and two shorter sides of length $w_1$,
		\item
		four tall isosceles triangles $(h,i,g)$, $(i,l,g)$, $(h,m,c)$, and $(m,n,c)$ with bases $(h,i)$, $(i,l)$, $(h,m)$, and $(m,n)$, respectively, of length $w_3$ and two longer sides of length~$w_1$.
	\end{enumerate}
	By \cref{lem:two-2-trees-merge} and the fact that any two of the listed triangles share at most one edge, we have that the flag gadget is a $2$-tree.
	The flag gadget provides a single attachment edge, which coincides with the base $(b,c)$ of the transmission triangle $(a,b,c)$. 
	\begin{figure}[tb!]
		\centering
		\subfloat[\label{fig:flag_gadget-a}]{\includegraphics[page=2,width=.33\textwidth]{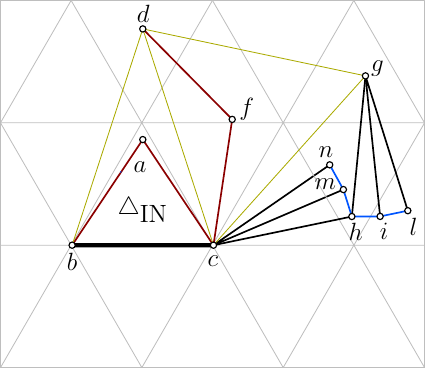}}
		\hfil
		\subfloat[\label{fig:flag_gadget-b}]{\includegraphics[page=3,width=.33\textwidth]{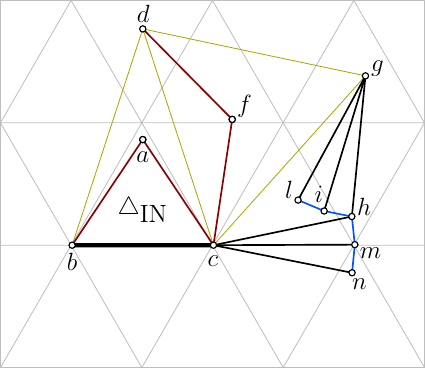}}
		\hfil
		\subfloat[\label{fig:flag_gadget-c}]{\includegraphics[page=4,width=.33\textwidth]{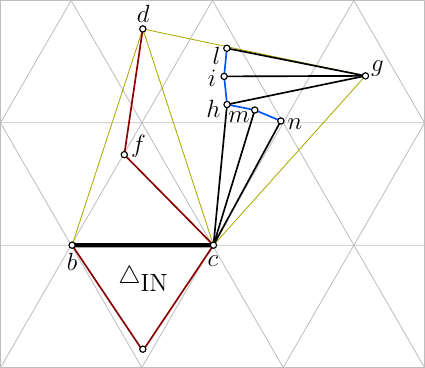}}
		\caption{The flag gadget and three planar straight-line realizations of it. In (a) and (b) the transmission triangle $(a,b,c)$ lies inside the triangle $(d,b,c)$, while in (c) it lies outside $(d,b,c)$. In (a) the isosceles triangles $(h,m,c)$ and $(m,n,c)$ lie inside $(h,g,c)$, in (b) the isosceles triangles $(h,i,g)$ and $(i,l,g)$ lie inside $(h,g,c)$, and in (c) all four the isosceles triangles $(h,m,c)$, $(m,n,c)$, $(h,m,c)$, and $(m,n,c)$ lie inside $(h,g,c)$.}
		\label{fig:flag_gadget}
	\end{figure}
	
	We will exploit the following geometric properties of the flag gadget.
	
	\begin{property}\label{pr:flag-realizations}
		The flag gadget admits, among others, planar straight-line realizations in which:
		\begin{enumerate}[(a)]
			\item
			The triangle $(a,b,c)$ lies inside the triangle $(d,b,c)$, and the triangles $(h,m,c)$ and $(m,n,c)$ lie inside the triangle $(h,g,c)$; see~\cref{fig:flag_gadget-a}.
			\item
			The triangle $(a,b,c)$ lies inside the triangle $(d,b,c)$, and the triangles $(h,i,g)$ and $(i,l,g)$ lie inside the triangle $(h,g,c)$; see~\cref{fig:flag_gadget-b}.
			\item
			The triangle $(a,b,c)$ lies outside the triangle $(d,b,c)$ and the triangles $(h,m,c)$, $(m,n,c)$, $(h,m,c)$, and $(m,n,c)$ lie inside the triangle $(h,g,c)$; see~\cref{fig:flag_gadget-c}.
		\end{enumerate}
	\end{property}
	
	\begin{property}\label{pr:flag-angles}
		The values of some relevant angles of the triangles forming the flag gadget are the following:
		\begin{itemize}
			\item
			$\widehat{bcd}:= \arccos{\frac{w_1}{2w_4}} = \arccos{\frac{1}{3.22}} \simeq 71.91^\circ$
			\item
			$\widehat{bca}:= \arccos{\frac{w_1}{2w_2}} = \arccos{\frac{1}{1.8}} \simeq 56.25^\circ$
			\item
			$\widehat{fcd}:= \arccos{\frac{w_4}{2w_2}} = \arccos{\frac{1.61}{1.8}} \simeq 26.56^\circ$
			\item
			$\widehat{hcg}:= \arccos{\frac{w_4}{2w_1}} = \arccos{\frac{1.61}{2}} \simeq 36.39^\circ$
			\item
			$\widehat{chg}:= 180^\circ - 2 \cdot \widehat{hcg} \simeq 107.22^\circ$
			\item
			$\widehat{chm}:=
			\arccos{\frac{w_3}{2w_1}} = \arccos{\frac{0.2}{2}} \simeq 84.26^\circ$
			\item
			Let $\lambda$ be the smallest internal angle of the tall isosceles triangles with two sides of length $w_1$ and one side of length $w_3$.
			We have $\lambda = 180^\circ -2 \cdot \widehat{chm} \simeq 180^\circ - 2 \cdot 84.26^\circ = 11.48^\circ$.
		\end{itemize}
	\end{property}
	
	We show the following.
	
	\begin{lemma}\label{lem:flag-behaviour}
		In any planar straight-line realization of the flag gadget, the following statements hold true:
		\begin{enumerate}[(a)]
			\item
			If the transmission triangle $(a,b,c)$ is drawn inside the triangle $(d,b,c)$, then the triangle $(h,g,c)$ is drawn outside the triangle $(g,d,c)$.
			\item
			Not both the vertices $l$ and $n$ lie inside the triangle $(h,g,c)$.
		\end{enumerate}
	\end{lemma}
	
	\begin{proof}
		Consider a planar straight-line realization of the flag gadget.
		By Properties~\ref{obs:containment_sharing_edges:tall_in_equilateral}~and~\ref{obs:containment_sharing_edges:equilateral_in_flat} of \cref{obs:containment_sharing_edges}, in such a realization the triangles $(d,b,c)$ and $(g,d,c)$ are not drawn inside one another.
		Hence, the triangle $(f,d,c)$ lies inside $(d,b,c)$ if and only if it lies outside $(g,d,c)$.
		
		
		To prove statement (a), suppose the triangle $(a,b,c)$ is drawn inside the triangle $(d,b,c)$.
		By our discussion above and the fact that $\widehat{fcd} + \widehat{bca} > \widehat{bcd}$, then the triangle $(f,d,c)$ lies inside $(g,d,c)$.
		By this fact and the fact that $\widehat{fcd} + \widehat{hcg} > \widehat{dcg} = 60^\circ$, we have that the triangle $(h,g,c)$ is drawn outside $(g,d,c)$.
		This ends the proof of statement (a).
		
		To prove statement (b) it suffices to show that $m$ and $i$ cannot both lie inside the triangle $(h,g,c)$.
		This follows from the fact that $\widehat{chm} = \widehat{ghi}$ and $2 \cdot \widehat{chm} > \widehat{chg}$.
	\end{proof}
	
	
	\begin{figure}[tb!]
		\centering
		\includegraphics[page=1,width=\textwidth]{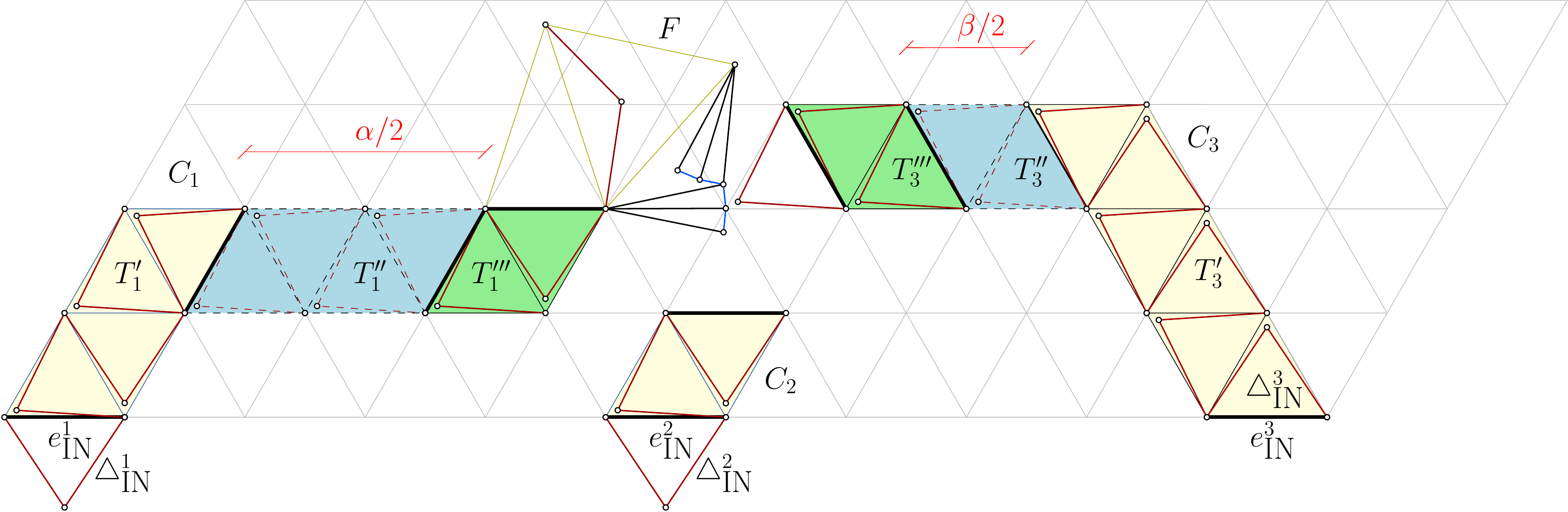}
		\caption{The $(\alpha,\beta)$-clause gadget with $\alpha=4$ and $\beta=2$.}
		\label{fig:clause_gadget}
	\end{figure}
	
	\paragraph{The clause gadget.}
	
	We now describe the construction of the clause gadget. Such a construction is parametric with respect to the distances between the attachment edges of the transmission gadgets coming from the variable gadgets whose literals appear in the clause. By Property~D\ref{prop:clause-distances} and since the attachment edges of the transmission gadgets have length one, the horizontal distance between the two leftmost attachment edges is $\frac{\alpha}{2}+2$, for some non-negative even integer value $\alpha$, and the horizontal distance between the two rightmost attachment edges is $\frac{\beta}{2}+3$, for some non-negative even integer value $\beta$. Based on the values $\alpha$ and $\beta$, we define the clause gadget, which we refer to as the \emph{$(\alpha,\beta)$-clause gadget}. The gadget provides three attachment edges $e^1_{\textrm{IN}}$, $e^2_{\textrm{IN}}$, and $e^3_{\textrm{IN}}$ incident to the transmission triangles $\triangle^1_\textrm{IN}$, $\triangle^2_\textrm{IN}$, and $\triangle^3_\textrm{IN}$, respectively. The attachment edges $e^1_{\textrm{IN}}$, $e^2_{\textrm{IN}}$, and $e^3_{\textrm{IN}}$ are identified with the attachment edges of the transmission gadgets coming from the variable gadgets whose literals appear in the clause. Also, recall that the transmission triangles $\triangle^1_\textrm{IN}$, $\triangle^2_\textrm{IN}$, and $\triangle^3_\textrm{IN}$ are identified with the ones of the transmission gadgets coming from the variable gadgets whose literals appear in the clause. In any planar straight-line realization of $H_{\phi}$, the edges $e^1_{\textrm{IN}}$, $e^2_{\textrm{IN}}$, $e^3_{\textrm{IN}}$ lie in this left-to-right order along the same horizontal line.
	
	We now describe the $(\alpha,\beta)$-clause gadget for a positive clause (the negative clause has a symmetric construction);  refer to \cref{fig:clause_gadget}.
	We exploit the placement of $e^1_{\textrm{IN}}$, $e^2_{\textrm{IN}}$, $e^3_{\textrm{IN}}$ (which is fixed, since such edges belong to the union of the frames of the variable, ladder, and transmission gadgets) as a reference for arranging the gadget components.
	The $(\alpha,\beta)$-clause gadget consists of three subgraphs $C_1$, $C_2$, and $C_3$ defined as follows.
	\begin{itemize}
		\item
		
		The subgraph $C_1$ is obtained by combining into a single connected component a turning $4$-transmission gadget $T'_1$, a straight $\alpha$-transmission gadget $T''_1$, a turning $2$-transmission gadget $T'''_1$, and a flag gadget $F$, as shown in \cref{fig:clause_gadget}.  
		The attachment edge $e_\textrm{IN}$ of $T'_1$ is the attachment edge $e^1_{\textrm{IN}}$ of the clause gadget.
		
		\item
		
		The subgraph $C_2$ is a straight $2$-transmission gadget, whose attachment edge $e_\textrm{IN}$ is the attachment edge $e^2_{\textrm{IN}}$ of the clause gadget.
		
		\item
		
		The subgraph $C_3$ is obtained by combining into a single connected component a turning $6$-transmission gadget $T'_3$, a straight $\beta$-transmission gadget $T''_3$, and a straight $2$-transmission gadget $T'''_3$, as depicted in \cref{fig:clause_gadget}.
		The attachment edge $e_\textrm{IN}$ of $T'_3$ is the attachment edge $e^3_{\textrm{IN}}$ of the clause gadget.
		
	\end{itemize}
	
	%
	
	By \cref{lem:two-2-trees-merge} and the fact that transmission and flag gadgets are $2$-trees, we have that each of $C_1$, $C_2$, and $C_3$ is a $2$-tree.
	We say that a straight-line realization of the $(\alpha,\beta)$-clause gadget is \emph{feasible}, if it satisfies the following properties:
	


	
	
	
	\begin{enumerate}[(A)]
		\item
		
		the attachment edges $e^1_\textrm{IN}$, $e^2_\textrm{IN}$, and $e^3_\textrm{IN}$ lie in this left-to-right order on a horizontal line $\ell$, with horizontal distance $\frac{\alpha}{2}+2$ between $e^1_\textrm{IN}$ and $e^2_\textrm{IN}$, and horizontal distance $\frac{\beta}{2}+3$ between $e^2_\textrm{IN}$ and $e^3_\textrm{IN}$; and
		
		\item
		
		the frame triangles of the clause gadget all lie above $\ell$.
		
	\end{enumerate}
	Consider any feasible planar straight-line realization of the $(\alpha,\beta)$-clause gadget.
	In such realization, consider the trapezoid $\mathcal{T}$ whose base is the shortest horizontal segment containing $e^1_\textrm{IN}$, $e^2_\textrm{IN}$, and $e^3_\textrm{IN}$, whose lateral sides have slope $60^\circ$, and whose height is $2\sqrt{3}$.
	Observe that the base of $\cal T$ has length $8+(\alpha+\beta)/2$.
	A key ingredient for the logic of the reduction is the following lemma.
	
	\begin{lemma}\label{lem:clause-behaviour}
		The $(\alpha,\beta)$-clause gadget admits a feasible planar straight-line realization if and only if at least one of the transmission triangles $\triangle^1_\textrm{IN}$, $\triangle^2_\textrm{IN}$, and $\triangle^3_\textrm{IN}$ lies outside $\cal T$.
	\end{lemma}
	
	\begin{proof}
		Consider any feasible planar straight-line realization of the $(\alpha,\beta)$-clause gadget.
		Note that, for $i=1,2,3$, if the transmission triangle $\triangle^i_\textrm{IN}$ lies inside $\cal T$, then it lies inside the frame triangle $\triangle_s$ of the transmission gadget composing $C_i$ that contains the edge $e^i_\textrm{IN}$.
		By this observation and \cref{lem:transmission-gadget-behaviour}, we obtain the following properties:
		
		\begin{property}\label{prop:triangle_in-1}
			Consider the transmission triangle $(a,b,c)$ and the tall isosceles triangle $(d,b,c)$ of the flag gadget $F$.
			If the transmission triangle $\triangle^1_\textrm{IN}$ lies inside $\cal T$, then the triangle $(a,b,c)$ is drawn inside the triangle $(d,b,c)$.
			Conversely, if the triangle $(a,b,c)$ is drawn outside the triangle $(d,b,c)$, then the transmission triangle $\triangle^1_\textrm{IN}$ lies outside $\cal T$.
		\end{property}
		
		\begin{property}\label{prop:triangle_in-2}
			Consider the transmission triangle $\triangle_{\textrm{OUT}}$ and the frame triangle $\triangle_t$ of $C_2$.
			If the transmission triangle $\triangle^2_\textrm{IN}$ lies inside $\cal T$, then $\triangle_{\textrm{OUT}}$ lies outside $\triangle_t$.
			Conversely, if $\triangle_{\textrm{OUT}}$ lies inside $\triangle_t$, then the transmission triangle $\triangle^2_\textrm{IN}$ lies outside $\mathcal{T}$.
		\end{property}
		
		\begin{property}\label{prop:triangle_in-3}
			Consider the transmission triangle $\triangle_{\textrm{OUT}}$ and the frame triangle $\triangle_t$ of $T'''_3$.
			If the transmission triangle $\triangle^3_\textrm{IN}$ lies inside $\mathcal{T}$, then $\triangle_{\textrm{OUT}}$ is drawn outside $\triangle_t$.
			Conversely, if $\triangle_{\textrm{OUT}}$ is drawn inside $\triangle_t$, then the transmission triangle $\triangle^3_\textrm{IN}$ lies outside $\mathcal{T}$.
		\end{property}
		
		Suppose first that at least one of $\triangle^1_\textrm{IN}$, $\triangle^2_\textrm{IN}$, and $\triangle^3_\textrm{IN}$ lie outside $\cal T$.
		We prove that the $(\alpha,\beta)$-clause gadget admits a feasible planar straight-line realization.
		We have the following cases:
		\begin{itemize}
			\item
			
			The transmission triangle $\triangle^1_\textrm{IN}$ lies outside $\cal T$.
			
			In this case the transmission triangles of $T'_1$, $T''_1$, and $T'''_1$ can be drawn in such a way that, in the flag gadget $F$, the transmission triangle $(a,b,c)$ lies outside the triangle $(b,c,d)$.
			This allows the flag gadget to adopt the realization~(c) of~\cref{pr:flag-realizations} (see also \cref{fig:flag_gadget-c}).    
			Observe that, regardless of whether $\triangle^2_\textrm{IN}$ and $\triangle^3_\textrm{IN}$ lie inside $\cal T$, the transmission triangles of $C_2$, $T'_3$, $T''_3$, and $T'''_3$ can be drawn without crossings.
			
			\item
			
			The transmission triangle $\triangle^2_\textrm{IN}$ lies outside $\cal T$.
			
			In this case the transmission triangles of $C_2$ can be drawn in such a way that the transmission triangle $\triangle_\textrm{OUT}$ of $C_2$ lies inside the frame triangle $\triangle_t$ of $C_2$.
			This allows the flag gadget to adopt the realization~(b) of~\cref{pr:flag-realizations} (see also \cref{fig:flag_gadget-b}).
			Observe that, regardless of whether $\triangle^1_\textrm{IN}$ and $\triangle^3_\textrm{IN}$ lie inside $\cal T$, the transmission triangles of $T'_1$, $T''_1$, $T'''_1$, $T'_3$, $T''_3$, and $T'''_3$ can be drawn without crossings.
			
			\item
			
			The transmission  $\triangle^3_\textrm{IN}$ lies outside $\cal T$.
			
			In this case the transmission triangles of $T'_3$, $T''_3$, and $T'''_3$ in such a way that the transmission triangle $\triangle_\textrm{OUT}$ of $T'''_3$ lies inside the frame $\triangle_t$ of $T'''_3$.
			This allows the flag gadget to adopt the realization~(a) of~\cref{pr:flag-realizations} (see also \cref{fig:flag_gadget-a}).
			Observe that, regardless of whether $\triangle^1_\textrm{IN}$ and $\triangle^2_\textrm{IN}$ lie inside $\cal T$, the transmission triangles of $T'_1$, $T''_1$, $T'''_1$, and $C_2$ can be drawn without crossings.
			
		\end{itemize}
		
		\begin{figure}[tb!]
			\centering
			\subfloat[\label{fig:clause_gadget_hexagon-a}]
			{\includegraphics[page=1,width=0.45\textwidth]{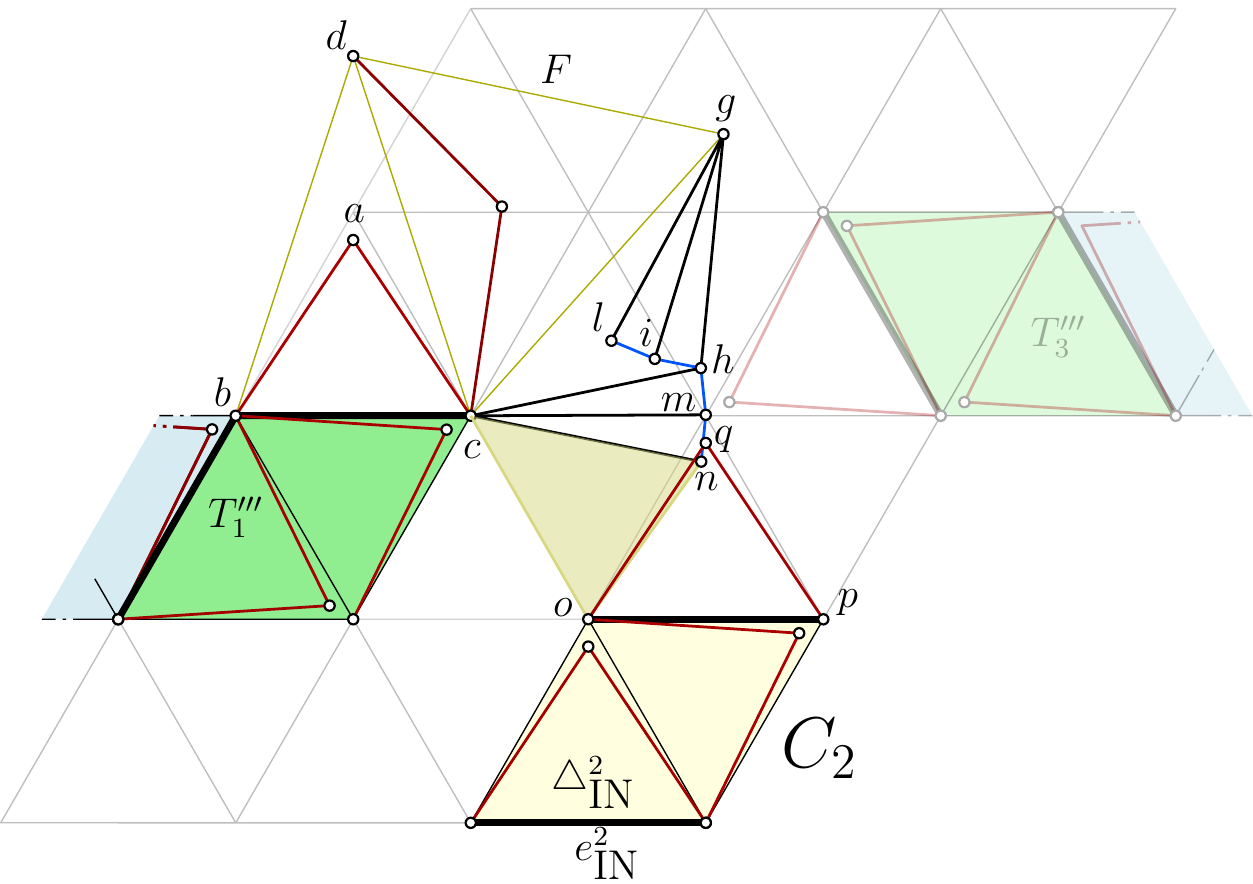}}
			\subfloat[\label{fig:clause_gadget_hexagon-b}]
			{\includegraphics[page=2,width=0.45\textwidth]{np-hardness__clause_gadget_hexagon}}
			
			\caption{Illustration for the proof of \cref{lem:clause-behaviour} when $\triangle^1_\textrm{IN}$, $\triangle^2_\textrm{IN}$, and $\triangle^3_\textrm{IN}$ lie all in the interior of $\cal T$.}
			\label{fig:clause_gadget_hexagon}
		\end{figure}
		
		Suppose now that all $\triangle^1_\textrm{IN}$, $\triangle^2_\textrm{IN}$, and $\triangle^3_\textrm{IN}$ lie inside $\cal T$.
		We prove that the $(\alpha,\beta)$-clause gadget does not admit a feasible planar straight-line realization.
		We analyze the possible planar straight-line realizations of the flag gadget $F$.
		By \cref{prop:triangle_in-1} and statement~(a) of \cref{lem:flag-behaviour}, the triangle $(h, g, c)$ is drawn outside the triangle $(g, d, c)$.
		On the other hand, by statement~(b) of \cref{lem:flag-behaviour}, we have that either (i) the vertex $n$ lies outside the triangle $(h, g, c)$ (refer to \cref{fig:clause_gadget_hexagon-a}), (ii) the vertex $l$ lies outside the triangle $(h, g, c)$ (refer to \cref{fig:clause_gadget_hexagon-b}), or (iii) both vertices $n$ and $l$ lie outside the triangle $(h, g, c)$.
		We show next that in all cases, there is an intersection between the edges two triangles of the $(\alpha,\beta)$-clause gadget.
		
		Consider first the case in which the vertex $n$ lies outside the triangle $(h, g, c)$; refer to \cref{fig:clause_gadget_hexagon-a}.
		The points $o,p,q$ are the vertices of the transmission triangle $\triangle_\textrm{OUT}$ of $C_2$.
		We prove that $\widehat{c o n} + \widehat{q o p} > 120^\circ$ and that the length of the segment $\overline{on}$ is smaller than $w_2$.
		These two statements imply that the segment $\overline{c n}$ intersects the segment $\overline{o q}$.
		To prove both statements we consider an auxiliary isosceles triangle with vertices $n$, $c$, $o$ (note that only the side $\overline{c n}$ of the triangle $(n,c,o)$ corresponds to an edge of $H_{\phi}$).
		We have that:
		\begin{itemize}
			
			\item
			
			The internal angle at $c$ is $\widehat{{n c o}} \simeq 48.74^\circ$.
			Namely, $\widehat{n c o}:= 240^\circ - \widehat{b c d} - \widehat{d c g} - \widehat{g c h} - 2\lambda$.
			By \cref{pr:flag-angles}, we obtain $\widehat{{n c o}} \simeq 240^\circ - 71.91^\circ - 60^\circ - 36.39^\circ - 2 \cdot 11.48^\circ = 48.74^\circ$.
			\vspace{0.4em}
			
			\item
			
			The internal angle at $o$ is $\widehat{c o n} \simeq (180^\circ - \widehat{n c o})/2 = 65.63^\circ$.
			
			\item
			
			The segments $\overline{c o}$ and $\overline{c n}$ have length $w_1$.
			
			\item
			
			The segment $\overline{o n}$ has length $2 w_1 \sin(\frac{\widehat{n c o}}{2}) \simeq 0.83$.
			
		\end{itemize}
		
		Since the triangle $(q,o,p)$ is a transmission triangle, we have that $\widehat{q o p} = \arccos(\frac{w_1}{2w_2}) = \arccos(\frac{1}{1.8}) = 56.25^\circ$, and the segment $\overline{oq}$ has length $w_2$.
		Hence, $\widehat{c o n} + \widehat{q o p} \simeq 65.63^\circ + 56.25^\circ = 121.88^\circ > 120^\circ$, and $\vert \overline{o q}\vert=w_2 = 0.9 > \vert \overline{o n} \vert \simeq 0.83$.
		This implies that $\overline{cn}$ intersects $\overline{oq}$, as claimed.
		
		Consider now the case in which the vertex $l$ lies outside the triangle $(h, g, c)$; refer to \cref{fig:clause_gadget_hexagon-b}.
		The points $r,t,s$ denote the vertices of the transmission triangle $\triangle_{\textrm{OUT}}$ of $T'''_3$.
		Let $u$ be the point at distance $w_1=1$ from $r$ such that the ray from $r$ through $u$ has slope $120^{\circ}$.
		We show that $\widehat{g r u} + \widehat{l r g} + \widehat{t r s} > 180^\circ$ and that the length of the segment $\overline{r l}$ is smaller than $w_2$.
		These two statements imply that the segment $\overline{gl}$ intersects the segment $\overline{rs}$.
		To prove both statements we consider two auxiliary isosceles triangles.
		The first triangle has vertices $c$, $g$, and $r$.
		We have that:
		\begin{itemize}
			\item The segment $\overline{cg}$ has length $w_4$.
			\item The segment $\overline{cr}$ has length $\vert \overline{cr} \vert = 2 \frac{\sqrt{3}}{2} = \sqrt{3}$.
			\item
			The internal angle at $c$ is $\widehat{g c r} \simeq 18.09^\circ$. Namely, $\widehat{b c r} = 150^\circ$.
			Further, $\widehat{b c g} = \widehat{bcd} + \widehat{dcg}$ and, by \cref{pr:flag-angles}, we have $\widehat{bcg} \simeq 71.91^\circ + 60^\circ = 131.91^\circ$, hence $\widehat{g c r} = \widehat{b c r} - \widehat{b c g} \simeq 18.09^\circ$.
			\vspace{0.4em}
			
			\item
			The segment $\overline{gr}$ has length $\vert \overline{gr} \vert \simeq 0.54$. Namely, by the law of cosines, we have $|\overline{gr}|^2 = |\overline{cg}|^2 + |\overline{cr}|^2 -2  |\overline{cg}|\cdot |\overline{cr}| \cdot \cos(\widehat{g c r})$.
			Since $|\overline{cg}| = w_4 = 1.61$ and $|\overline{cr}| = \sqrt{3}$, we obtain $|\overline{gr}|^2 \simeq 2.5921 + 3 - 2 \cdot 1.61 \cdot \sqrt{3} \cdot \cos(18.09^\circ) \simeq 0.29$, hence $|\overline{gr}| \simeq 0.54$.
			\vspace{0.4em}
			
			\item
			The internal angle at $r$ is $\widehat{crg} \simeq 68.25^\circ$. Namely, by the law of cosines, we have $\widehat{crg} = \arccos\left(\frac{|\overline{gr}|^2 +|\overline{cr}|^2 - |\overline{gc}|^2} {2|\overline{gr}|\cdot|\overline{cr}|}\right) \simeq \arccos\left(\frac{0.54^2+3-1.61^2}{2 \cdot 0.54 \cdot \sqrt{3}}\right) \simeq \arccos\left(\frac{0.6924}{1.8684}\right) \simeq \arccos(0.3706) \simeq 68.25^\circ$.
			\vspace{0.4em}
			
			\item
			The internal angle at $g$ is $\widehat{rgc} = 180^\circ - \widehat{crg} - \widehat{gcr} \simeq 180^\circ - 68.25^\circ - 18.09^\circ = 93.66^\circ$.
		\end{itemize}
		
		The second auxiliary triangle we consider has vertices $g$, $r$, and $l$. We have that:
		\begin{itemize}
			\item The segment $\overline{lg}$ has length $w_1=1$.
			\item
			The internal angle at $g$ is $\widehat{r g l} \simeq 34.31^\circ$. Namely, $\widehat{l g c} = \widehat{h g c} + 2 \lambda$. By \cref{pr:flag-angles}, we have $\widehat{l g c} \simeq 36.39^\circ + 2 \cdot 11.48^\circ = 59.35^\circ$.
			We obtain that $\widehat{rgl} = \widehat{rgc} - \widehat{lgc} \simeq 93.66^\circ - 59.35^\circ = 34.31^\circ$.
			\vspace{0.4em}
			\item
			The edge $\overline{l r}$ has length $\vert \overline{l r} \vert \simeq 0.63$. Namely, by the law of cosines, we have that $|\overline{lr}|^2 = |\overline{l g}|^2 + |\overline{rg}|^2 - 2 |\overline{l g}| \cdot |\overline{rg}| \cos(\widehat{rgl})$, and hence $|\overline{lr}| \simeq \sqrt{  1 + 0.54^2 - 2 \cdot 0.54 \cdot\cos(34.31^{\circ})} \simeq \sqrt{0.3995} \simeq 0.63$.
			\vspace{0.4em}
			
			\item
			The internal angle at $r$ is $\widehat{l r g} = 117.2^\circ$. Namely, by the law of cosines, we have $\widehat{lrg} = \arccos\left( \frac{|\overline{gr}|^2 +|\overline{lr}|^2 - |\overline{gl}|^2} {2|\overline{gr}|\cdot|\overline{lr}|}\right) \simeq$  $\arccos\left(\frac{0.54^2+0.63^2-1^2}{2 \cdot 0.54 \cdot 0.63}\right) = \arccos\left(\frac{-0.3115}{0.6804}\right) \simeq \arccos(-0.4578) \simeq 117.2^\circ$.
		\end{itemize}
		
		Let $v$ be the point at distance $w_1=1$ from $r$ such that the ray from $r$ through $v$ has slope $240^{\circ}$. Note that $\widehat{v r u} = \widehat{v r c} + \widehat{c r g} + \widehat{g r u}$, hence $\widehat{g r u} = \widehat{v r u} - \widehat{v r c} - \widehat{c r g} \simeq 120^\circ - 30^\circ - 68.25^\circ = 21.75^\circ$.
		On the other hand, since the triangle $(r,s,t)$ is a transmission triangle, we have $\widehat{t r s} \simeq 56.25^\circ$ and the segment $\overline{rs}$ has length $w_2$.
		Hence $\widehat{g r u} + \widehat{l r g} + \widehat{t r s} \simeq 21.75^\circ +  117.2^\circ + 56.25^\circ = 195.2^\circ > 180^\circ$, and $\vert \overline{r l} \vert \simeq 0.63 < w_2 = 0.9$.
		This implies that $\overline{gl}$ intersects $\overline{rs}$, as claimed, and concludes the proof of the lemma.
	\end{proof}
	
	\subsection{Proof of the reduction}
	\label{sec:np-hardness_combine_gadgets}
	
	We are now ready to prove the main result of this section.
	
	\paragraph{Proof of \cref{th:np-hard}.}
	
	We first prove that, starting from the boolean formula $\phi$, the incidence graph $G_\phi$ of $\phi$, and the monotone rectilinear representation $\Gamma_\phi$ of $G_\phi$, we can construct the $2$-tree $H_\phi$ in polynomial time.
	The first step of the construction is to obtain from $\Gamma_\phi$ the auxiliary drawing $\Gamma^*_\phi$ we described in \cref{sec:auxiliary_representation}.
	By \cref{lem:auxiliary-representation}, we can construct $\Gamma^*_\phi$ in polynomial time.
	The next step is to construct $H_\phi$ using $\Gamma^*_\phi$ as an auxiliary tool.
	We obtain $H_\phi$ by introducing a variable gadget, a clause gadget, and a transmission gadget for each variable, clause, and edge in $G_\phi$, respectively, and then combining these gadgets together by identifying their attachment edges.
	In particular, we exploit $\Gamma^*_\phi$ to (i) define the size $k$ of each $k$-transmission gadget that represents an edge of $G_\phi$, (ii) define the parameters $\alpha$ and $\beta$ of each $(\alpha,\beta)$-clause gadget, and (iii) select the appropriate attachment edges to combine the gadgets together.
	We then merge the frames of the variable gadgets that are consecutive in the left-to-right order of the rectangles representing variables in $\Gamma^*_\phi$.
	This is done by means of \emph{ladder gadgets} (the shaded blue triangles in \cref{fig:overview-reduction}), which are maximal outerpaths, each composed of a sequence of frame triangles. We have that $H_\phi$ is a $2$-tree by repeated applications of \cref{lem:two-2-trees-merge}, and the union of the frame triangles induces a maximal outerplanar graph.
	Moreover, from the detailed description of the gadgets construction of \cref{sec:np-hardness_gadgets}, it is not hard to see that $H_\phi$ can be constructed in polynomial time.
	
	We now prove that $H_\phi$ admits a planar straight-line realization \emph{if and only if} $\phi$ is satisfiable.
	Suppose first that $H_\phi$ admits a planar straight-line realization $\Gamma_H$.
	We show that $\phi$ is satisfiable.
	For each variable $v$ of $\phi$, consider the variable gadget $\mathcal V$ of $H_\phi$ modeling $v$, the truth-assignment triangle $\triangle_{\textrm{IN}}$ of $\mathcal V$, and the frame triangle $\triangle_{s}$ of the split gadget $S_{\delta_\phi}$ of $\mathcal V$; refer to \cref{fig:variable_gadget-schema,fig:variable_gadget}.
	We set $v=\texttt{True}$ if and only if $\triangle_{\textrm{IN}}$ is drawn inside $\triangle_{s}$.
	We next prove this truth assignment satisfies every clause $c$ of $\phi$.
	Assume that $c$ is a positive clause with variables $v_1$, $v_2$, and $v_3$, that is, $c = (v_1 \vee v_2 \vee v_3)$.
	The proof for the case in which $c$ is negative is symmetric.
	Let $\mathcal C$ denote the clause gadget modeling $c$ and, for $i=1,2,3$, let $\mathcal T_i$ denote the transmission gadget modeling the edge $(v_i,c)$ of $G_\phi$.
	From the construction of the clause gadget, the realization of $\mathcal C$ in $\Gamma_H$ is such that the frame of the components of $\mathcal C$ are bounded by a trapezoid $\mathcal T$.
	By \cref{lem:clause-behaviour}, at least one of the triangles $\triangle^1_{\textrm{IN}}$, $\triangle^2_{\textrm{IN}}$, and $\triangle^3_{\textrm{IN}}$ of $\mathcal C$ lies outside $\mathcal T$.
	We assume w.l.o.g. that it is actually $\triangle^1_{\textrm{IN}}$ the triangle that lies outside $\mathcal T$.
	Then $\triangle^1_{\textrm{IN}}$ lies inside the frame triangle $\triangle_t$ of the transmission gadget $\mathcal T_1$.
	Let $\mathcal V_1$ denote the variable gadget modeling the variable $v_1$.
	By \cref{lem:transmission-gadget-behaviour}, there exists an index $j \in \{1,2,\dots, \delta_\phi\}$ such that the transmission triangle $\triangle^j_\textrm{OUT}$ of $\mathcal V_1$ lies inside the frame triangle $\triangle^j_t$ of $\mathcal V_1$ (for such index $j$ the transmission triangle $\triangle^j_\textrm{OUT}$ is shared by $\mathcal V_1$ and $\mathcal T_1$).
	Thus, by \cref{lem:variable-gadget-behaviour}, the truth-assignment triangle $\triangle_{\textrm{IN}}$ of $\mathcal V_1$ is drawn inside the frame triangle $\triangle_s$ of the split gadget $S_{\delta_\phi}$ of $\mathcal V_1$.
	This implies in turn that $v_1$ is assigned the value $\texttt{True}$.
	Since this variable appears as a positive literal in $c$, we have that $c$ is satisfied.
	This shows that the constructed truth assignment satisfies all the clauses of $\phi$.
	
	Suppose now that $\phi$ is satisfiable. We show that $H_\phi$ admits a planar straight-line realization $\Gamma_H$.
	Let $\tau$ be a satisfying truth assignment for $\phi$.
	For each variable $v$ of $\phi$, let $\tau(v)$ denote the truth value of $v$ in $\tau$.
	Observe that, up to a rigid transformation, the union of all the frame triangles of $H_\phi$ admits a unique planar straight-line realization.
	We initialize $\Gamma_H$ to such a realization.
	Further, for each variable $v$, we adopt the configuration of \cref{fig:variable-b} if $\tau(v) = \texttt{True}$ and the configuration of \cref{fig:variable-c} if $\tau(v) = \texttt{False}$.
	For each transmission gadget modeling an edge of $G_{\phi}$ incident to a variable $v$ and to a positive clause $c$, we adopt the configuration of \cref{fig:transmission_gadget-a}\textcolor{blue}{(right)} if $\tau(v) = \texttt{True}$ and the configuration of \cref{fig:transmission_gadget-a}\textcolor{blue}{(left)} if $\tau(v) = \texttt{False}$; a symmetric choice is made if $c$ is negative.
	Since, for each clause $c$ of $\phi$, there exists at least one literal that is $\texttt{True}$, the triangle $\triangle^i_{\textrm{IN}}$, $i \in \{1,2,3\}$, associated with this literal is drawn outside the trapezoid $\cal T$ that bounds the frame of the components of $\mathcal C$ in $\Gamma_H$.
	Then, by \cref{lem:clause-behaviour}, the clause gadget modeling $c$ admits a feasible planar straight-line realization; such realizations are used to complete the planar straight-line realization $\Gamma_H$ of $H_{\phi}$. This concludes the proof of \cref{th:np-hard}.

	\section{A Linear-time Algorithm for 2-trees with Two Edge Lengths}
	\label{sec:few_lengths}
	
	In this section, we study the \FEPRshort problem for weighted $2$-trees in which each edge can only have one of at most two distinct lengths.
	We prove that, in this case, the \FEPRshort problem is linear-time solvable. We first solve the case in which all the edges are prescribed to have the same length; we remark that the \FEPRshort problem is NP-hard for general weighted planar graphs in which all the edges have the same length~\cite{DBLP:journals/dam/EadesW90}.
	
	\begin{theorem}\label{th:same-weight}
		Let $G=(V,E,\lambda)$ be an $n$-vertex weighted $2$-tree, where $\lambda: E \rightarrow \{ w \}$ with $w\in\mathbb{R}^+$.
		There exists an $O(n)$-time algorithm that tests whether $G$ admits a planar straight-line realization and, in the positive case, constructs such a realization.
	\end{theorem}
	
	\begin{proof}
		Suppose that there exists a planar straight-line realization $\Gamma$ of $G$. Since all the edges of $G$ have length $w$, every $3$-cycle of $G$ is represented in $\Gamma$ by an equilateral triangle of side $w$, hence no triangle can be contained inside any other triangle in $\Gamma$. This, together with the fact that $G$ is a $2$-tree, implies that $\Gamma$ is an outerplanar drawing of $G$. Thus, if $G$ admits a planar straight-line realization, then it is a maximal outerplanar graph and, as such, it has a unique outerplane embedding~\cite{mw-opbe-90,s-cog-79}.
		
		Therefore, in order to test whether $G$ has a planar straight-line realization, we test whether it is an outerplanar graph; this can be done in $O(n)$ time~\cite{d-iroga-07,m-laarogmog-79,w-rolt-87}.
		If the test fails, we conclude that $G$ admits no planar straight-line realization.
		Otherwise, we construct in $O(n)$ time its unique outerplane embedding~\cite{d-iroga-07,m-laarogmog-79,mw-opbe-90,s-cog-79,w-rolt-87}.
		Finally, by means of \cref{thm:straight-line_realization_planarity}, we test in $O(n)$ time whether $G$ has a planar straight-line realization that respects the computed outerplane embedding.
		If the test fails, we conclude that $G$ admits no planar straight-line realization.
		Otherwise, \cref{thm:straight-line_realization_planarity} provides us with the desired planar straight-line realization of $G$.
	\end{proof}
	
	We now extend our study to weighted graphs in which each is assigned with one of two possible lengths. We have the following main theorem.
	
	\begin{theorem}\label{th:two-weights}
		Let $G=(V,E,\lambda)$ be an $n$-vertex weighted  $2$-tree, where $\lambda: E \rightarrow \{ w_1, w_2 \}$ with $w_1,w_2\in\mathbb{R}^+$.
		There exists an $O(n)$-time algorithm that tests whether $G$ admits a planar straight-line realization and, in the positive case, constructs such a realization.
	\end{theorem}
	
	In the remainder of the section, we prove \cref{th:two-weights}.
	Hereafter, we assume, w.l.o.g., that $w_1 < w_2$.
	Note that the realization of any $3$-cycle of $G$ is one of the following types of triangles (refer to \cref{fig:possible_triangles}): 
	\begin{enumerate}[(i)]
		\item an equilateral triangle of side $w_1$ (a \emph{small equilateral triangle}),
		\item an equilateral triangle of side $w_2$ (a \emph{big equilateral triangle}),
		\item an isosceles triangle with base $w_1$ and two sides of length $w_2$ (a \emph{tall isosceles triangle}), and
		\item an isosceles triangle with base $w_2$ and two sides of length $w_1$ (a \emph{flat isosceles triangle}).
	\end{enumerate}
	
	Any triangle of one of the types above is called an \emph{interesting triangle}.
	
	\begin{figure}[ht]
		\centering
		\includegraphics[scale=.9]{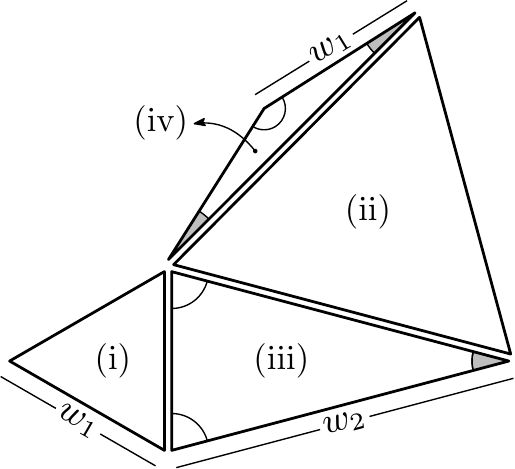}
		
		\caption{
			The four interesting triangles. Since $w_1 < w_2$, the angles in gray are always less than $60^{\circ}$, while the angles in white are always greater than $60^{\circ}$.
		}
		\label{fig:possible_triangles}
	\end{figure}
	
	While small equilateral, big equilateral and tall isosceles triangle exist regardless of the values of $w_1$ and $w_2$, flat isosceles triangles do not exist under the condition described in the following property.
	
	\begin{property} \label{pr:existence-flat}
		If $w_2\geq 2w_1$, then no $3$-cycle of $G$ is realized as a flat isosceles triangle.
	\end{property}
	
	\begin{proof}
		By assumption, every $3$-cycle of $G$ satisfies the triangle inequality, that is, the sum of the lengths prescribed to any two edges is larger than the length of the third edge. Hence, if $w_2 \geq 2w_1$, then no $3$-cycle of $G$ can have two edges with length $w_1$ and the third edge with length $w_2$.
	\end{proof}
	
	
	\subsection{Containment among triangles}
	\label{sec:few_lengths:triangles}
	
	The first central ingredient of our algorithm are the conditions satisfied by interesting triangles when they are drawn inside each other in a planar straight-line realization of $G$.
	In the following, we denote with $r$ the ratio $\frac{w_2}{w_1} > 1$. Regardless of the value of $r$, the following lemma holds true.
	
	\begin{lemma}\label{prop:containments}
		The interesting triangles satisfy the two properties below. 
		\begin{enumerate}[a)]
			\item \label{prop:containments:small-equilateral_flat-isosceles} No interesting triangle can be drawn inside a small equilateral or a flat isosceles triangle.
			\item \label{prop:containments:big-equilateral_tall-isosceles} Neither a big equilateral nor a tall isosceles triangle can be drawn inside any interesting triangle.
		\end{enumerate}
	\end{lemma}
	
	\begin{proof}
		Let $\triangle_1$ and $\triangle_2$ be any two interesting triangles.
		We make a case by case analysis based on the types of interesting triangles of $\triangle_1$ and $\triangle_2$.
		In both statment (a) and (b) we have the following common case: If $\triangle_1$ is the same type of interesting triangle as $\triangle_2$, then $\triangle_1$ and $\triangle_2$ are congruent, hence neither triangle can be drawn inside the other.
		
		We now prove statement~(\ref{prop:containments:small-equilateral_flat-isosceles}).
		We suppose that $\triangle_2$ is either small equilateral or flat isosceles, and show that $\triangle_1$ cannot be drawn inside $\triangle_2$ regardless of the type of $\triangle_1$.
		We have the following cases:
		
		\begin{itemize}
			\item
			
			$\triangle_2$ is small equilateral and $\triangle_1$ is either big equilateral, tall isosceles, or flat isosceles.
			
			In this case $\triangle_1$ cannot be drawn inside $\triangle_2$ by \cref{obs:containments}.
			
			\item
			$\triangle_2$ is flat isosceles.
			\begin{itemize}
				\item
				
				$\triangle_1$ is big equilateral.
				
				If $\triangle_1$ and $\triangle_2$ do not share any edge, then $\triangle_1$ cannot be drawn inside $\triangle_2$ by \cref{obs:containments}.
				If instead $\triangle_1$ and $\triangle_2$ share an edge, then $\triangle_1$ cannot be drawn inside $\triangle_2$ by Property~\ref{obs:containment_sharing_edges:equilateral_in_flat}  of \cref{obs:containment_sharing_edges}.
				
				\item
				
				$\triangle_1$ is tall isosceles.
				
				If $\triangle_1$ and $\triangle_2$ do not share any edge, then $\triangle_1$ cannot be drawn inside $\triangle_2$ by \cref{obs:containments}.
				If instead $\triangle_1$ and $\triangle_2$ share an edge, then $\triangle_1$ cannot be drawn inside $\triangle_2$ by Property~\ref{obs:containment_sharing_edges:tall_in_flat}  of \cref{obs:containment_sharing_edges}.
				
				\item
				
				$\triangle_1$ is small equilateral.
				
				If $\triangle_1$ and $\triangle_2$ share an edge, then $\triangle_1$ cannot be drawn inside $\triangle_2$ by Property~\ref{obs:containment_sharing_edges:equilateral_in_flat} of \cref{obs:containment_sharing_edges}.
				Consider instead that $\triangle_1$ does not share an edge with $\triangle_2$.
				Let $a,b,c$ be the vertices of $\triangle_2$, where $\overline{b c}$ is the side of $\triangle_2$ with length $w_2$; refer to \cref{fig:small_equilateral_in_flat_isosceles}.
				Let $d$ be the orthogonal projection of the vertex $a$ on the segment $\overline{b c}$, let $R_1$ denote the closed region bounded by the triangle with vertices $a,c,d$, and let $R_2$ denote the closed region bounded by the triangle with vertices $a,b,d$.
				Since $\triangle_2$ is flat isosceles, by \cref{pr:existence-flat} we have $w_2<2w_1$, hence the lengths of the segments $\overline{bd}$, $\overline{cd}$, and $\overline{ad}$ are all smaller than $w_1$.
				This implies that the only pair of points inside $R_1$ (resp.\ inside $R_2$) at distance greater than or equal to $w_1$ is $(a,c)$ (resp.\ $(a,b)$).
				Suppose, for a contradiction, that $\triangle_1$ is contained inside $\triangle_2$.
				Then two vertices of $\triangle_1$ are contained in the same triangle $R_i$, say they are contained in $R_1$.
				Since at least one of such two vertices does not coincide with $a$ or $c$ (otherwise $\triangle_1$ and $\triangle_2$ would share an edge), the distance between such two vertices is smaller than $w_1$, while the length of the edge between them is at least $w_1$, a contradiction.
			\end{itemize}
		\end{itemize}
		
		\begin{figure}[ht]
			\centering%
			\includegraphics[scale=1]{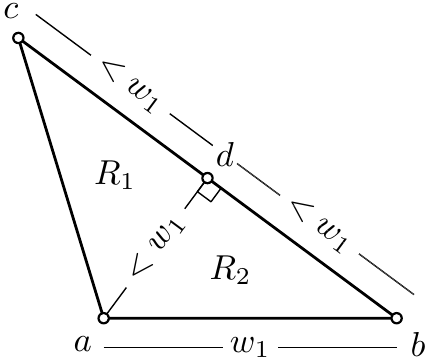}%
			\caption{Illustration for the proof of statement (\ref{prop:containments:small-equilateral_flat-isosceles}) of \cref{prop:containments}.}
			\label{fig:small_equilateral_in_flat_isosceles}
		\end{figure}
		
		We now prove statement~(\ref{prop:containments:big-equilateral_tall-isosceles}).
		We suppose that $\triangle_1$ is either big equilateral or tall isosceles, and prove that $\triangle_1$ cannot be drawn inside $\triangle_2$ regardless of the type of $\triangle_2$.  
		Note that, if $\triangle_2$ is small equilateral or flat isosceles, then $\triangle_1$ cannot be drawn inside $\triangle_2$ by statement (\ref{prop:containments:small-equilateral_flat-isosceles}) of the present lemma.
		We can hence assume that $\triangle_2$ is either a big equilateral triangle or a tall isosceles triangle.
		We have the following cases.
		
		\begin{itemize}    
			\item
			
			$\triangle_1$ and $\triangle_2$ are of different types of interesting triangles and do not share any edge.
			
			In this case $\triangle_1$ cannot be drawn inside $\triangle_2$ by \cref{obs:containments}.
			
			\item
			$\triangle_1$ is tall isosceles, $\triangle_2$ is big equilateral, and $\triangle_1$ and $\triangle_2$ share an edge.
			
			In this case $\triangle_1$ cannot be drawn inside $\triangle_2$ by Property \ref{obs:containment_sharing_edges:tall_in_equilateral} of \cref{obs:containment_sharing_edges}.
			
			\item
			
			$\triangle_1$ is big equilateral, $\triangle_2$ is tall isosceles, and $\triangle_1$ and $\triangle_2$ share an edge.
			
			In this case $\triangle_1$ cannot be drawn inside $\triangle_2$ by \cref{obs:internal_angles}.
			
			
		\end{itemize}
		This concludes the proof of the lemma.
	\end{proof}
	
	\remove{
		As we show next, a flat isosceles can be drawn inside a tall-isosceles, but only for certain values of $r$.
		
		\begin{lemma}\label{lem:flat_in_tall_few_edges}
			If a flat isosceles is drawn inside a tall isosceles, then they have as common side the base of the flat isosceles and $\frac{1+\sqrt{5}}{2} < r < 2$.
		\end{lemma}
		
		\begin{proof}  
			The fact that the flat isosceles is necessarily sharing its base with one of the shorter sides of the tall isosceles, derives from \cref{obs:containments}(\ref{obs:containments:sides}) and the fact that both flat and tall isosceles triangles have a side of length $w_2$. The upper bound $r < 2$ comes from the fact that a flat isosceles has a base of length $w_2$ and two sides of length $w_1$. We now prove the lower bound for $r$. By \cref{lem:flat_in_tall}(\ref{lem:flat_in_tall:flat_base}) we have that $A - 2 AB^2 + B^3 > 0$ with $A > \frac{1}{2}$ and $B>1$, where $A$ is the ratio between the lengths of the shortest side of the flat isosceles and the base of the tall isosceles, and $B$ is the ratio between the longest side and the base of the tall isosceles. Since $A = \frac{w_1}{w_1} = 1$ and $B = \frac{w_2}{w_1} = r$, we obtain the equation
			\begin{align}\label{eq:flat_in_tall}
				r^3-2r^2+1 &> \nonumber\\
				\left( r-1 \right)\left( r^2-r-1 \right) &> 0.
			\end{align}
			Since $r>1$, \cref{eq:flat_in_tall} is satisfied when $r^2-r-1>0$. Hence we obtain that
			\[
			r>\frac{1+\sqrt{5}}{2} \simeq 1.618.
			\]
		\end{proof}
	}
	
	
	
	
	Let $T$ denote the decomposition tree of $G$. Suppose that a planar straight-line realization $\Gamma$ of $G$ exists. The second central ingredient of our algorithm is the next lemma.
	
	\begin{lemma}\label{lem:leaf-triangle}
		Let $\triangle_1$ and $\triangle_2$ be two triangles realizing two different $3$-cycles of $G$ in $\Gamma$. Suppose that $\triangle_1$ is drawn inside $\triangle_2$. Then $\triangle_1$ is a leaf triangle of $T$ that shares a side with $\triangle_2$.
	\end{lemma}
	
	\begin{proof}
		By statement~(\ref{prop:containments:small-equilateral_flat-isosceles}) of \cref{prop:containments}, we have that $\triangle_2$ is either big equilateral or tall isosceles. We discuss the two cases separately.
		
		Suppose first that $\triangle_2$ is big equilateral.
		Consider any interesting triangle $\triangle_3$ that is drawn inside $\triangle_2$ in $\Gamma$ and shares a side with $\triangle_2$.
		Note that $\triangle_3$ is flat isosceles; namely, $\triangle_3$ cannot be tall isosceles or big equilateral by statement~(\ref{prop:containments:small-equilateral_flat-isosceles}) of \cref{prop:containments}, and it cannot be small equilateral since the sides of $\triangle_2$ have length $w_2$.
		Now suppose, for the sake of contradiction, that an interesting triangle $\triangle_4$ shares a side of length $w_1$ with $\triangle_3$.
		Let $v$ be the common vertex between $\triangle_2$, $\triangle_3$, and $\triangle_4$.
		If $\triangle_4$ is small equilateral or tall isosceles, as in \cref{fig:leaf-triangle:1} with $v=a$, then the sum of the angles at $v$ of $\triangle_3$ and $\triangle_4$ exceeds $60^{\circ}$.
		Hence $\triangle_4$ crosses $\triangle_2$, a contradiction.
		Further, if $\triangle_4$ is flat isosceles, as in \cref{fig:leaf-triangle:1} with $v=b$, then $\triangle_4$ crosses $\triangle_2$ by \cref{obs:containments}.
		It follows that the only interesting triangles inside  $\triangle_2$ are flat isosceles triangles that are leaves of $T$ and that share a side with $\triangle_2$.
		
		Suppose next that $\triangle_2$ is tall isosceles. The proof follows similar arguments as for the case in which $\triangle_2$ is big equilateral, although, in this case, we distinguish two subcases.
		Namely, an interesting triangle $\triangle_3$ that is drawn inside $\triangle_2$ and that shares a side with $\triangle_2$ can be either small equilateral (and then it shares a side of length $w_1$ with $\triangle_2$) as in \cref{fig:leaf-triangle:2}, or flat isosceles (and then it shares a side of length $w_2$ with $\triangle_2$) as in \cref{fig:leaf-triangle:3}.
		
		\begin{figure}[ht]
			\centering
			\subcaptionbox{\label{fig:leaf-triangle:1}}
			{\includegraphics[scale=.85]{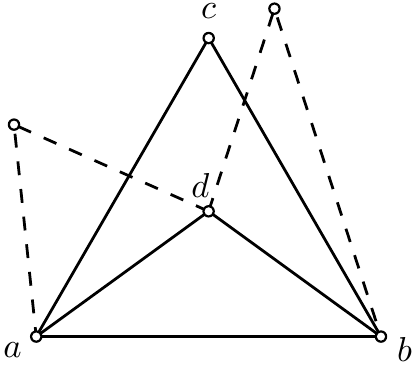}}
			\qquad
			\subcaptionbox{\label{fig:leaf-triangle:2}}
			{\includegraphics[scale=.85]{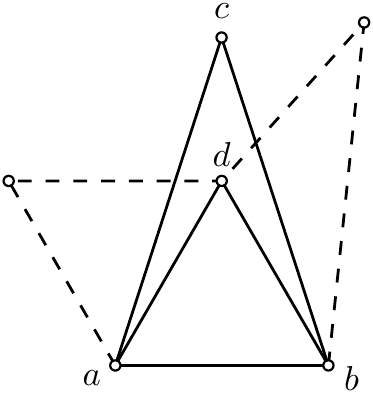}}
			
			\caption
			{
				Two illustrations for the proof of \cref{lem:leaf-triangle}. (\subref{fig:leaf-triangle:1}) The case in which $\triangle_2$ is a big equilateral triangle and $\triangle_3$ is a flat isosceles triangle. (\subref{fig:leaf-triangle:2}) The case in which $\triangle_2$ is a tall isosceles triangle and $\triangle_3$ is a small equilateral triangle.
				In both figures, the possible cases for $\triangle_4$ are shown with dashed lines.
			}
			\label{fig:leaf-triangle}
		\end{figure}
		
		Consider first the subcase in which $\triangle_3$ is small equilateral.
		Suppose, for the sake of contradiction, that an interesting triangle $\triangle_4$ shares a side of length $w_1$ with $\triangle_3$.
		Let $v$ be the common vertex between $\triangle_2$, $\triangle_3$, and $\triangle_4$.
		If $\triangle_4$ is small equilateral or tall isosceles, as in \cref{fig:leaf-triangle:2} with $v=a$, then the sum of the angles at $v$ of $\triangle_3$ and $\triangle_4$ exceeds $90^{\circ}$.
		Hence $\triangle_4$ crosses $\triangle_2$, a contradiction.
		On the other hand, if $\triangle_4$ is flat isosceles, as in \cref{fig:leaf-triangle:2} with $v=b$, then $\triangle_4$ crosses $\triangle_2$ by \cref{obs:containments}.
		
		Consider next the subcase in which $\triangle_3$ is flat isosceles.
		Suppose, for the sake of contradiction, that an interesting triangle $\triangle_4$ shares a side of length $w_1$ with $\triangle_3$.
		Let $v$ be the common vertex between $\triangle_2$, $\triangle_3$, and $\triangle_4$.
		Let $a$ and $b$ be the vertices incident to the base of $\triangle_2$, and let $c$ be the third vertex of $\triangle_2$; refer to \cref{fig:leaf-triangle:3}.
		The following cases arise depending of whether $v$ is equal to $c$ or $a$:
		\begin{itemize}
			\item
			
			Case 1: $v=c$. 
			If $\triangle_4$ is small equilateral or tall isosceles, then the angle at $v$ of $\triangle_4$ is larger than or equal to $60^{\circ}$, which is larger than the angle at $v$ of $\triangle_2$.
			Hence $\triangle_4$ crosses $\triangle_2$.
			If $\triangle_4$ is instead flat isosceles, then $\triangle_4$ crosses $\triangle_2$ by \cref{obs:containments}.
			In both case we have a contradiction.
			
			\item
			
			\begin{figure}[ht]
				\centering
				\subcaptionbox{\label{fig:leaf-triangle:3}}
				{\includegraphics[scale=.85]{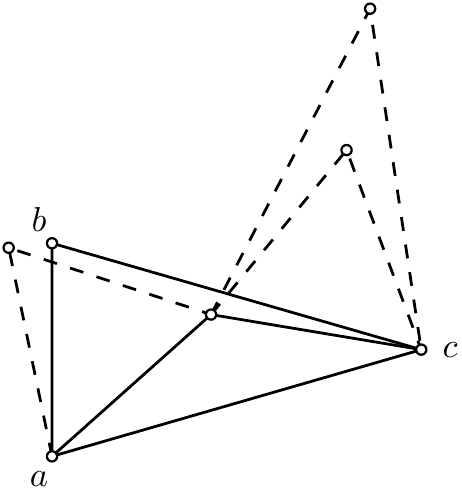}}
				\qquad
				\subcaptionbox{\label{fig:leaf-triangle:4}}
				{\includegraphics[scale=.85]{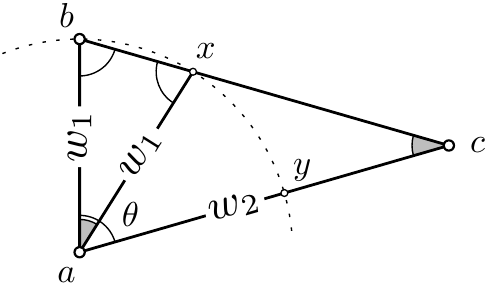}}
				
				\caption
				{
					Two illustrations for the proof of \cref{lem:leaf-triangle}, for the case in which $\triangle_2$ is a tall isosceles triangle and $\triangle_3$ is a flat isosceles triangle. (\subref{fig:leaf-triangle:3}) The case in which $v=c$ and $\triangle_4$ (in dashed lines) is flat isosceles or tall isosceles, and the case in which $v=a$ and $\triangle_4$ (in dashed lines) is small equilateral.
					(\subref{fig:leaf-triangle:4}) A focus on the case in which $v=a$ and $\triangle_4$ is small equilateral.
				}
				\label{fig:leaf-triangle_2}
			\end{figure}
			
			Case 2: $v=a$.
			If $\triangle_4$ is flat isosceles or tall isosceles, then $\triangle_4$ crosses $\triangle_2$ by \cref{obs:containments}, a contradiction.
			Consider that $\triangle_4$ is instead small equilateral, as in \cref{fig:leaf-triangle:3}.
			Let $\angle z$ denote the interior angle of $\triangle_2$ at vertex $z$, for $z \in \{a,b,c\}$.
			Let $C$ be the circle with center on $a$ and radius $w_1$.
			The circle cuts the interior of the segment $\overline{a c}$ in a point $y$ and it cuts the interior of the segment $\overline{b c}$ in a point $x$, since the length of $\overline{a c}$ is $w_2>w_1$ and since $\angle b<90^{\circ}$; see \cref{fig:leaf-triangle:4}.
			If $\triangle_4$ is drawn inside $\triangle_2$, the vertices of $\triangle_4$ that are different from $a$ lie on the circular arc of $C$ that connects $x$ and $y$ through the interior of $\triangle_2$.
			Furthermore, since $\triangle_4$ is equilateral, the internal angle at $a$ of the triangle with vertices $a,c,x$ should be greater than $60^{\circ}$.
			Let $\theta$ denote such angle.
			We show that $\theta$ is actually less than $60^{\circ}$, which is a contradiction.
			Note that $\triangle_2$ and the triangle with vertices $a,b,x$ are similar to each other, hence $\theta = \angle a - \angle c$.
			Since $\angle c + 2\angle a = 180^{\circ}$ we have that
			\begin{equation}
				\label{eq:theta}
				\theta = 3\angle a - 180^{\circ}.
			\end{equation}
			On the other hand, we have that $\cos(\angle a) = \frac{1}{2} \frac{w_1}{w_2}$ and, since $\triangle_1$ exists only if $\frac{w_1}{w_2}>\frac{1}{2}$, then $\cos(\angle a) > \frac{1}{4}$. Hence
			\begin{equation}
				\label{eq:beta}
				\angle a < \arccos(\frac{1}{4}) \simeq 75.52^{\circ},
			\end{equation}
			and from \cref{eq:theta,eq:beta} we obtain that $\theta = 3\angle a-180^{\circ} \simeq 46.56^{\text{o}} < 60^{\text{o}}$.
		\end{itemize}
		
		It follows that the only triangles inside  $\triangle_2$ are flat isosceles triangles and small equilateral triangles that are leaves of $T$ and that share a side with $\triangle_2$.
		
		This concludes the proof of the lemma.   \end{proof}

	\subsection{Conflicts between leaf triangles}
	\label{sec:conflicts}
	
	Hereafter, we assume that the decomposition tree $T$ of $G$ is rooted at any $3$-cycle whose sum of the edge lengths is maximum.
	Hence, the triangle realizing such a $3$-cycle cannot be drawn inside any other triangle in a planar straight-line realization of $G$.
	
	The \emph{framework} of $G$ is the subgraph $G_F \subseteq G$ obtained as follows: For each leaf triangle $\triangle_i$ that can be drawn inside a triangle $\triangle_j$, which is either its parent or one of its siblings, we remove from $G$ the vertex $v$ that $\triangle_i$ does not share with $\triangle_j$, along with the two edges incident to $v$.
	Note that $G_F$ is a $2$-tree.
	Further, by \cref{lem:leaf-triangle}, in any planar straight-line realization $\Gamma$ of $G$, the triangles we removed from $G$ in order to define $G_F$ are the only triangles that can be drawn inside other triangles representing $3$-cycles of $G$.
	It follows that the restriction of $\Gamma$ to $G_F$ is an outerplanar drawing.
	Hence, we start by testing if $G_F$ is outerplanar, which can be done in linear time~\cite{d-iroga-07,m-laarogmog-79,w-rolt-87}; in the negative case, we reject the instance.
	In the positive case, we test whether $G_F$ admits a planar straight-line realization respecting its unique outerplane embedding~$\cal E$.
	Both the computation of~$\cal E$ and the test can be performed in linear time using the algorithms in~\cite{d-iroga-07,m-laarogmog-79,w-rolt-87} and \cref{thm:straight-line_realization_planarity}, respectively.
	If the test is negative we reject the instance, otherwise we have a planar straight-line realization $\Gamma_F$ of $G_F$ with embedding~$\cal E$.
	
	Let $L_{\triangle}$ denote the set of triangles that were removed from $G$ to obtain $G_F$.
	The third central ingredient of our algorithm is a characterization of the conditions under which $\Gamma_F$ can be extended to a planar straight-line realization of $G$, by drawing in $\Gamma_F$ all the triangles of $L_{\triangle}$.
	
	Let $\triangle = (a,b,c)$ be a triangle of $L_{\triangle}$ such that $e=(a,b)$ is the unique edge shared by $\triangle$ and $G_F$.
	Suppose we want to add a drawing of $\triangle$ to $\Gamma_F$.
	If $e$ is incident to two internal faces of $\Gamma_F$, then we call \emph{internal embedding} to each of the two possible drawings of $\triangle$, regardless of whether $a,b,c$ appear in clockwise or counter-clockwise order when traversing $\triangle$.
	Consider that $e$ is instead incident to one internal face and to the outer face of $\Gamma_F$.
	Suppose that $a$ precedes $b$ in a counter-clockwise traversal of the outer face of $\Gamma_F$.
	If $\triangle$ is drawn such that $a,b,c$ appear in clockwise order when traversing $\triangle$, then we call such drawing of $\triangle$ an \emph{outer embedding}.
	Otherwise, we call such drawing of $\triangle$ an \emph{internal embedding}.
	
	
	
	
	
	
	Consider a straight-line realization $\Gamma$ of $G$ that \emph{extends} $\Gamma_F$; that is, whose restriction to $G_F$ is $\Gamma_F$.
	Suppose that $\Gamma$ is not planar.
	We now classify the possible crossings involving triangles in $L_\triangle$.
	
	First, suppose there is a triangle $\triangle \in L_\triangle$ whose drawing in $\Gamma$ crosses an edge of $\Gamma_F$.
	If the drawing of $\triangle$ is an outer embedding, then we say it induces a \emph{framework conflict}.
	If the drawing of $\triangle$ is instead an internal embedding, then we say it induces an \emph{overlapping conflict}.
	Second, consider two triangles $\triangle_i$ and $\triangle_j$ of $L_{\triangle}$.
	If the drawings of $\triangle_i$ and $\triangle_j$ in $\Gamma$ are both outer embeddings, neither of them induces a framework conflict, and they cross each other, then we say that such drawings induce an \emph{external conflict}.
	Similarly, if the drawings of $\triangle_i$ and $\triangle_j$ in $\Gamma$ are both internal embeddings, neither of them induces an overlapping conflict, and they cross each other, then we say that such drawings induce an \emph{internal conflict}.
	In both cases, we say that the drawings of $\triangle_i$ and $\triangle_j$ \emph{conflict with each other}.

	We have the following.
	
	\begin{lemma}\label{lem:internal_conflicts}
		Let $\triangle_i$ and $\triangle_j$ be two triangles in $L_{\triangle}$, and $\Gamma$ be a straight-line realization of $G$ that extends $\Gamma_F$.
		Suppose that there exists a triangle $\triangle \in \Gamma_F$ such that both $\triangle_i$ and $\triangle_j$ are drawn inside $\triangle$ in $\Gamma$.
		The drawings of $\triangle_i$ and $\triangle_j$ induce an internal conflict if at least one of the following statements is true:
		\begin{enumerate}[(a)]
			
			\item \label{lem:internal_conflicts:same_edge} $\triangle_i$ and $\triangle_j$ share an edge.
			
			\item \label{lem:internal_conflicts:tall_isosceles} $\sqrt{3} < r \leq 2 \cos(15^{\circ})\simeq 1.93$ and $\triangle$ is a tall isosceles triangle.
			
			\item \label{lem:internal_conflicts:all_conflicts} $1 < r \leq \sqrt{3}$.
			
		\end{enumerate}
	\end{lemma}
	
	\begin{proof}
		By statement~(\ref{prop:containments:small-equilateral_flat-isosceles}) of \cref{prop:containments} and since each of $\triangle_i$ and $\triangle_j$ is drawn inside $\triangle$, the triangle $\triangle$ is either tall isosceles or big equilateral.
		Furthermore, by statement~(\ref{prop:containments:big-equilateral_tall-isosceles}) of \cref{prop:containments}, each of $\triangle_i$ and $\triangle_j$ is either small equilateral or flat isosceles.
		By \cref{lem:leaf-triangle}, the triangles $\triangle_i$ and $\triangle$ share an edge and so do $\triangle_j$ and $\triangle$.
		
		Consider first the case in which $\triangle_i$ shares an edge with $\triangle_j$, as described in statement (\ref{lem:internal_conflicts:same_edge}).
		Note that, since $G$ is a $2$-tree, the triangles $\triangle_i$, $\triangle_j$, and $\triangle$ actually share the same edge.
		By statement~(\ref{prop:containments:small-equilateral_flat-isosceles}) of \cref{prop:containments}, $\triangle_i$ and $\triangle_j$ cannot be drawn inside each other regardless of their types.
		Since $\triangle_i$ and $\triangle_j$ are drawn inside the same triangle and they cannot be drawn inside each other, their drawings induce an internal conflict inside $\triangle$.
		This proves statement (\ref{lem:internal_conflicts:same_edge}).
		
		Consider now the case in which $\triangle_i$ and $\triangle_j$ do not share an edge.
		The triangles $\triangle_i$ and $\triangle_j$ induce an internal conflict inside $\triangle$ depending on the types of $\triangle_i$, $\triangle_j$, $\triangle$, and the value of $r$.
		Note that $r$ is necessarily less than $2$ since, otherwise, by \cref{pr:existence-flat} both $\triangle_i$ and $\triangle_j$ would be small equilateral triangles and they could not share each a different side of $\triangle$, since $\triangle$ has at most one side of length $w_1$.  
		By a simple geometric argument we can deduce that two triangles are drawn inside $\triangle$ without inducing internal conflicts in the following cases (refer to \cref{fig:internal_conflicts}):
		
		\begin{enumerate}[(i)]
			
			\item If $2 \cos(15^{\circ}) \simeq 1.93 < r < 2$ and $\triangle$ is tall isosceles, then two flat isosceles triangles and one small equilateral triangle can be drawn inside $\triangle$ without inducing internal conflicts.
			
			\item If $\sqrt{3} \simeq 1.73 < r < 2$ and $\triangle$ is big equilateral, then three flat isosceles can be drawn inside $\triangle$ without inducing internal conflicts.
			
		\end{enumerate}
		
		The limit values $r=2 \cos(15^{\circ})$ and $r=\sqrt 3$ split the interval $1 < r <2$ into the following three intervals.
		First, if $2 \cos(15^{\circ}) < r < 2$, then the cases (i) and (ii) are both possible.
		Hence, $\triangle_i$ and $\triangle_j$ can be drawn inside $\triangle$ without inducing internal conflicts, regardless of their types.
		Second, if $\sqrt{3} < r \leq 2 \cos(15^{\circ})$ then the case (i) is no longer possible.
		Hence, as described in statement (\ref{lem:internal_conflicts:tall_isosceles}), $\triangle_i$ and $\triangle_j$ induce an internal conflict inside $\triangle$ if $\triangle$ is a tall isosceles triangle.
		Finally, if $1 < r \leq \sqrt{3}$, then none of the cases is possible.
		Hence, as described in statement (\ref{lem:internal_conflicts:all_conflicts}), $\triangle_i$ and $\triangle_j$ induce an internal conflict inside $\triangle$ regardless of their types.
	\end{proof}
	
	\begin{figure}[ht]
		\centering%
		\includegraphics[page=1]{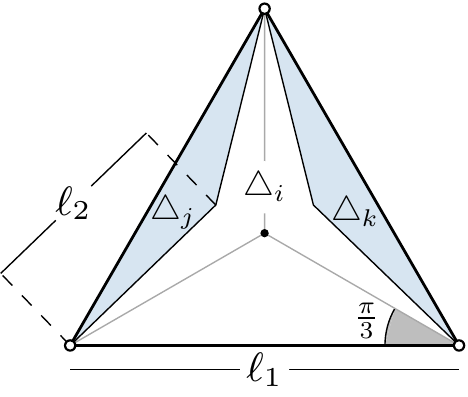}
		\qquad%
		\includegraphics[page=2]{two_distances_type_four_conflict}
		\qquad%
		\includegraphics[page=3]{two_distances_type_four_conflict}
		
		\caption{Illustration for the proof of \cref{lem:internal_conflicts}.
			The vertices that $\triangle_i$ and $\triangle_j$ do not share with $\triangle$ lie on the circumcenter of $\triangle$. On the left, the case in which $\triangle$ is tall isosceles (the segments between the circumcenter of $\triangle$ and its vertices all have the same length $l$). On the right, the case in which $\triangle$ is big equilateral.}
		\label{fig:internal_conflicts}
	\end{figure}
	
	From \cref{lem:internal_conflicts} we obtain the following necessary condition for the existence of a  planar straight-line realization of $G$ extending $\Gamma_F$.
	
	\begin{lemma}\label{lem:consistent_leafs}
		If there exists a planar straight-line realization of $G$, then no subset of three triangles of $L_{\triangle}$ share a side.
	\end{lemma}
	
	\begin{proof}
		Suppose, for a contradiction, that there exists a planar straight-line realization $\Gamma$ of $G$ and that three triangles of $L_{\triangle}$ share an edge $e$. Consider two of such triangles that are on the same side of $e$ in $\Gamma$ and denote them by $\triangle_i$ and $\triangle_j$. If $\triangle_i$ and $\triangle_j$ do not cross each other in $\Gamma$, then they are contained one inside the other. However, this is not possible by statement~(\ref{prop:containments:small-equilateral_flat-isosceles}) of \cref{prop:containments}.
	\end{proof}
	
	If $L_{\triangle}$ satisfies the condition stated in \cref{lem:consistent_leafs}, we say that $L_{\triangle}$ is \emph{consistent}.
	
	\begin{lemma}\label{lem:internal_conflict_bound}
		If $L_{\triangle}$ is consistent, then there are $O(1)$ pairs of triangles of $L_{\triangle}$ whose drawings induce an internal conflict inside a triangle $\triangle\in\Gamma_F$.
	\end{lemma}
	
	\begin{proof}
		By \cref{lem:leaf-triangle}, all the triangles of $L_{\triangle}$ that can be drawn inside $\triangle$ (including those inducing an internal conflict) share an edge with $\triangle$.
		By \cref{lem:consistent_leafs} and the assumed consistency of $L_{\triangle}$, there are at most $6$ of such triangles.
		Hence, the drawings of at most $\binom{6}{2} = 15$ pairs of triangles can induce an internal conflict inside $\triangle$.
	\end{proof}

	\subsection{The algorithm}
	
	We are now ready to describe the algorithm to solve the FEPR problem in $O(n)$ time when the edges of $G$ can only have one of two lengths.
	The algorithm consists of two stages.
	In the first stage, we compute the framework $G_F$ of $G$, and test whether $G_F$ admits a planar straight-line realization $\Gamma_F$.
	Such a realization is constructed in the positive case, and the input instance is rejected otherwise.
	In the second stage, we test whether $\Gamma_F$ can be extended to a planar straight-line realization $\Gamma$ of $G$, by drawing the triangles of $L_{\triangle}$ without inducing conflicts.
	In the positive case, we construct and report $\Gamma$ as the desired realization of $G$.
	As before, the input instance is rejected in the negative case.
	We describe the algorithm in detail next.
	
	\subsubsection{Computing a planar straight-line realization $\Gamma_F$ of the framework of $G$.}
	
	The first stage of the algorithm is formed by the following three steps.
	
	\paragraph*{Step A. Computing the framework $G_F$ of $G$ and the set $L_\triangle$ of leaf triangles.}
	
	To compute the framework of $G$ we first choose, among all the triangles of $G$, a triangle $\triangle_O$ with maximum sum of the length of its sides.
	Note that $\triangle_O$ cannot be drawn inside any other triangle of $G$, hence it does not belong to $L_\triangle$.
	We then compute the decomposition tree $T$ of $G$ rooted at $\triangle_O$.
	Note that, if $G$ admits a planar straight-line realization, then the degree of $T$ is bounded by a constant.
	Indeed, each edge of $G$ can be incident to at most two interesting triangles of the same type (and hence each $3$-cycle of $G$ can be adjacent to at most $24$ distinct $3$-cycles).
	We test whether the degree of $T$ is bounded by $24$ in $O(n)$ time.
	We reject the instance in the negative case.
	In the positive case we traverse $T$ looking for leaf triangles that can be drawn inside their parent or sibling triangles.
	By \cref{lem:leaf-triangle}, these are the only triangles in $G$ that can be drawn inside any other triangle.
	We identify such triangles by means of \cref{obs:internal_angles}.
	Since the degree of $T$ is in $O(1)$, each triangle has a constant number of sibling triangles, hence we can compute both $G_F$ and $L_{\triangle}$ in $O(n)$ time.
	
	\paragraph*{Step B. Testing if $L_\triangle$ is consistent.}
	
	The next step is to test if $L_\triangle$ is consistent.
	By \cref{lem:consistent_leafs}, this is a necessary condition for the existence of a planar straight-line realization of $G$.
	This condition can be trivially tested in $O(n)$ time by looking at the triangles of $L_\triangle$ sharing an edge with a parent or a sibling triangle in $T$.
	The input instance is rejected in the negative case.
	
	\paragraph*{Step C. Constructing a planar straight-line realization $\Gamma_F$ of $G_F$.}
	
	The final step of this stage is to compute a planar straight-line realization of $G_F$, if it exists.
	We first test whether $G_F$ is outerplanar. If this is not the case we reject the instance. Otherwise, we construct its unique outerplane embedding in $O(n)$ time~\cite{d-iroga-07,m-laarogmog-79,w-rolt-87}.
	Then, by means of \cref{thm:fixed-embedding} we test in $O(n)$ time if $G_F$ admits a planar straight-line realization with the computed embedding.
	We construct $\Gamma_F$ in the positive case, and reject the input instance otherwise.
	
	\subsubsection{Extending $\Gamma_F$ to a planar straight-line realization of $G$.}
	
	In the second stage the algorithm distinguishes two cases, namely, the one in which $r\geq 2$ and the one in which $r<2$. In the former case, the algorithm is completed by means of a single step, called D.1, while in the latter case the algorithm is completed by means of four steps, called D.2, E, F, and G.
	
	\smallskip
	\noindent \fbox{\bf Case $\mathbf r \geq 2$.}
	Note that, by \cref{pr:existence-flat}, in this case there are no flat isosceles triangles.
	Hence, every leaf triangle $\triangle$ in $L_\triangle$ is a small equilateral triangle such that the parent or a sibling of $\triangle$ is a tall isosceles triangle.
	We act as described in the following step.
	
	\paragraph*{Step D.1.}
	
	Consider each edge $e$ of $G_F$ incident to a triangle in $L_\triangle$. If $e$ is incident to two triangles in $L_\triangle$, we draw them in $\Gamma_F$ on opposite sides of $e$. If $e$ is incident to one triangle $\triangle\in L_\triangle$, then we draw $\triangle$ inside any tall isosceles triangle incident to $e$ (recall that such a tall isosceles triangle is the parent or a sibling of $\triangle$).
	At the end of this process we obtain a straight-line realization $\Gamma$ of $G$, and a plane embedding $\mathcal E$ of~$G$.
	
	\begin{lemma}\label{le:correctness-large-r}
		$G$ admits a planar straight-line realization if and only if $\Gamma$ is planar.
	\end{lemma}
	
	\begin{proof}
		Obviously, if $\Gamma$ is planar, then it is a planar straight-line realization of $G$.
		
		We prove that, if $G$ admits a planar straight-line realization, then $\Gamma$ is planar. Consider an edge $e \in G_F$.  Since $L_\triangle$ is consistent, the edge $e$ is shared by
		a set $L^e_\triangle \subseteq L_\triangle$ of
		at most two leaf triangles.
		Suppose that $e$ is an internal edge of $\Gamma_F$. Then, in $\Gamma$, each small equilateral triangle $L^e_\triangle$ is drawn as an internal embedding inside a tall isosceles triangle $\triangle^* \in G_F$.
		Note that no other triangle is drawn inside $\triangle^*$, since a tall isosceles triangle has a single side of length $w_1$.
		Hence, no triangle of $L^e_\triangle$ creates a crossing in $\Gamma$.
		
		Suppose now that $e$ is incident to the outer face of $\Gamma_F$.
		If $L^e_\triangle$ consists of a single triangle, then the triangle is drawn as an internal embedding and, as above, it does not create crossings in $\Gamma$.
		If $L^e_\triangle$ consists instead of two triangles, then least one of them, say $\triangle_s$, is drawn as an outer embedding. However, the two triangles in $L^e_\triangle$ are drawn on different sides of $e$ in any planar straight-line realization of $G$, hence if $\triangle_s$ creates a crossing, it does so in every straight-line realization of $G$. Since a planar straight-line realization of $G$ exists by hypothesis, it follows that $\triangle_s$ does not create crossings in $\Gamma$, and hence $\Gamma$ is planar.
	\end{proof}
	
	Observe that $\Gamma$ and $\mathcal E$ are constructed in linear time.
	We test whether $\Gamma$ is a planar straight-line realization of $G$ with embedding $\mathcal E$ in $O(n)$ time by means of \cref{thm:straight-line_realization_planarity}.
	By~\cref{le:correctness-large-r}, in the negative case we conclude that $G$ admits no planar straight-line realization, otherwise $\Gamma$ is the desired planar straight-line realization~of~$G$.
	
	\medskip
	\noindent \fbox{\bf Case $\mathbf r < 2$.}
	The algorithm in this case is more complicated, since there might be flat isosceles triangles in $G$, which might or might not need to be embedded inside their adjacent large equilateral or tall isosceles triangles in $\Gamma_F$.
	Moreover, more than a single leaf triangle might be drawn inside the same interesting triangle of $\Gamma_F$.
	
	An example of a planar straight-line realization of a weighted $2$-tree for which $r < 2$ is shown in \cref{fig:example}.
	The values of $w_1$ and $w_2$ were chosen so that $\sqrt{3} \simeq 1.73 < r \leq 2 \cos(15^{\circ}) \simeq 1.93$.
	By \cref{lem:internal_conflicts:tall_isosceles} of \cref{lem:internal_conflicts}, a pair of flat isosceles triangles can be drawn inside a big equilateral triangle without crossing each other, but any pair of interesting triangles (either two flat isosceles or a flat isosceles and a small equilateral) cross with each other if they are drawn inside a tall isosceles triangle.
	The triangles of $\Gamma_F$ are drawn with thick lines, and the edges of the leaf triangles of $L_{\triangle}$ that are not part of $\Gamma_F$ are drawn with thin lines.
	Observe that any type of interesting triangle can belong to the framework, whereas the triangles of $L_{\triangle}$ are either small equilateral or flat isosceles triangles.
	If a leaf triangle is small equilateral, then it is adjacent to (at least) an interesting triangle of the frame that is tall isosceles.
	If it is instead a flat isosceles, then it is adjacent to (at least) an interesting triangle of the frame that is either tall isosceles or big equilateral.
	
	From \cref{lem:leaf-triangle}, a leaf triangle can be drawn inside its parent triangle, a sibling triangle, or both.
	In our example, $\triangle_2$ can be drawn both inside its parent triangle $\triangle_1$ and inside the sibling triangle $\triangle_3$.
	On the other hand, $\triangle_{18}$ can only be drawn inside its parent triangle $\triangle_{17}$, and $\triangle_{21}$ can only be drawn inside its sibling triangle $\triangle_{22}$.
	Consider the triangle $\triangle_{27}$.
	It can only be drawn inside its parent triangle $\triangle_3$ since it is congruent to its sibling triangle $\triangle_8$. Hence it would induce an overlapping conflict if it were drawn on the same side as $\triangle_8$.
	Consider now the interesting triangles drawn with dashed lines.
	If $\triangle_9$ was drawn inside $\triangle_{17}$, then $\triangle_9$ and $\triangle_{18}$ would cross each other.
	On the other hand, if the drawing of $\triangle_{16}$ is an outer embedding, then it would induce a framework conflict.
	Finally, if the drawings of both $\triangle_{13}$ and $\triangle_6$ are outer embeddings, then they would induce an external conflict.
	\begin{figure}[bt!]
		\centering
		
		\includegraphics[page=1,scale=.9]{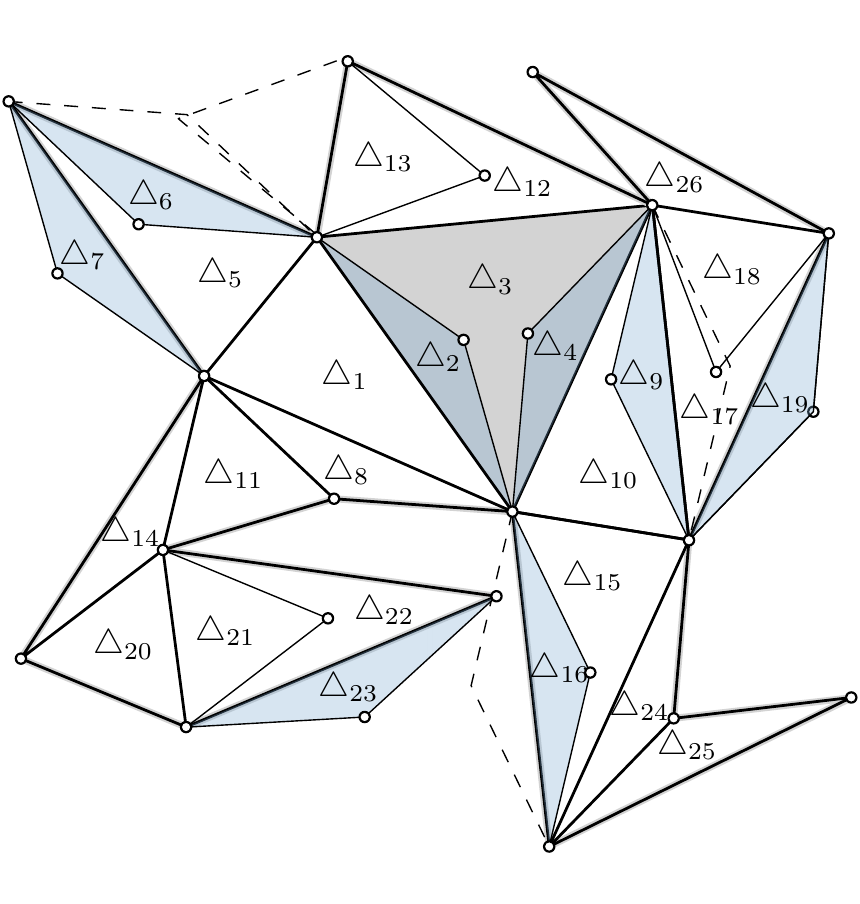}
		\quad%
		\includegraphics[page=2,scale=.9]{two_distances_external_conflicts}
		
		\caption{
			A planar straight-line realization of a weighted $2$-tree for $\sqrt{3} < r \leq 2\cos(15^{\circ})$.
			On the left, the planar straight-line realization.
			Triangles in the framework are drawn with thick lines. On the right, the decomposition tree rooted at $\triangle_1$.}
		\label{fig:example}
	\end{figure}
	
	To deal with this case we proceed as follows.
	First, we compute the pairs of triangles whose drawings may induce a conflict.
	Then we construct a 2SAT formula $\phi$ that contains a Boolean variable for each triangle $\triangle \in L_\triangle$.
	The value of the variable is associated to a choice of the two possible drawings of $\triangle$: recall that an edge of $\triangle$ is part of $G_F$, while the vertex of $\triangle$ not incident to such an edge might be placed on either side of it. For each triangle $\triangle$ that induces a framework or an overlapping conflict in one of its two drawings, there is a clause in $\phi$ that is \texttt{True} if and only if the drawing of the triangle is not the one inducing a conflict. Similarly, for each pair of triangles whose drawings induce an internal or an external conflict, there is a clause in $\phi$ that is \texttt{True} if and only if the drawings of the triangles are not both the ones inducing a conflict.
	If $\phi$ is satisfiable, then there is a planar straight-line realization of $G$ which can be obtained by embedding the triangles of $L_\triangle$ according to the values of the variables of $\phi$.
	We reject the instance otherwise.
	We describe the algorithm in detail next.
	
	\paragraph*{Step D.2. Detecting internal conflicts.}
	
	We start by computing all the pairs of triangles of $L_{\triangle}$ that induce an internal conflict.
	This is done by traversing $T$ while processing each triangle $\triangle \in \Gamma_F$ as follows.
	We first obtain the set $L^{\prime}_{\triangle} \subseteq L_{\triangle}$ of leaf triangles that can be drawn inside $\triangle$. By~\cref{lem:leaf-triangle}, these triangles share an edge with $\triangle$, so they are children or siblings of $\triangle$.
	We then decide which pairs of leaf triangles of $L^{\prime}_{\triangle}$ induce an internal conflict inside $\triangle$, by means of \cref{lem:internal_conflicts}.
	Since $L_{\triangle}$ is consistent, by \cref{lem:internal_conflict_bound} there are $O(1)$ triangles in $L^{\prime}_{\triangle}$.
	Hence, we process each triangle in $O(1)$ time, and the traversal takes $O(n)$ time.
	
	\paragraph*{Step E. Detecting overlapping conflicts.} Next, for each triangle $\triangle \in L_{\triangle}$, we detect whether any of its two drawings induces an overlapping conflict. Let $\triangle^*$ be the parent of $\triangle$ in $T$ and let $e$ be the edge shared by $\triangle$ and $\triangle^*$. Note that, since $\triangle \in L_{\triangle}$, one of the two drawings of $\triangle$ lies inside $\triangle^*$ and hence does not induce an overlapping conflict. We first check whether the other drawing of $\triangle$ is an internal or an outer embedding (in the latter case, we can conclude that $\triangle$ does not induce an overlapping conflict). This is equivalent to checking whether there is a triangle $\triangle_s$ in $\Gamma_F$ on the other side of $e$ with respect to $\triangle^*$. Note that, if such a triangle $\triangle_s$ exists, then it is a sibling of $\triangle$ in $T$, given that it shares the edge $e$ with $\triangle$ and that it is not the parent $\triangle^*$ of $\triangle$. Hence, for each of the $O(1)$ siblings of $\triangle$ in $T$, we check in $O(1)$ time whether it is a triangle in $\Gamma_F$ and whether it is incident to $e$. If no sibling of $\triangle$ in $T$ satisfies both checks, then we conclude that no drawing of $\triangle$ induces an overlapping conflict. Otherwise, there is a sibling of $\triangle$ in $T$ that satisfies both checks; then this sibling is the desired triangle $\triangle_s$ and hence the drawing of $\triangle$ outside $\triangle^*$ is an internal embedding of $\triangle$. Then we also check in $O(1)$ time whether $\triangle_s$ and the drawing of $\triangle$ outside $\triangle^*$ cross each other. In the positive case, we conclude that the drawing of $\triangle$ outside $\triangle^*$ induces an overlapping conflict, while in the negative case we conclude that it does not. Since we process each triangle $\triangle \in L_{\triangle}$ in $O(1)$ time, this step takes $O(n)$ time.      
	
	\paragraph*{Step F. Detecting framework and external conflicts.}
	
	Detecting framework and external conflicts in linear time is more challenging than overlapping and internal conflicts. Indeed, an outer embedding of a triangle $\triangle\in L_{\triangle}$ might cross a triangle of $\Gamma_F$ or the outer embedding of a triangle in $L_{\triangle}$ that is far, in terms of graph-theoretic distance, from $\triangle$. Hence, the trivial solution would detect framework and external conflicts by comparing each triangle $\triangle\in L_{\triangle}$ with every triangle of $\Gamma_F$ and with every triangle in $L_{\triangle}$. However, this would result in a quadratic running time. In order to achieve linear running time, we decompose the plane into cells and then derive several combinatorial properties of such a decomposition, including: (i) only a linear number of cells can host conflicts and such cells can be found efficiently; (ii) conflicts can only occur between triangles in a single cell or in adjacent cells; and (iii) each cell can host only a constant number of outer embeddings of triangles from $L_{\triangle}$. The last property is only true because of the assumption $r<2$. We now describe the details of Step F.
	
	Let $L^\prime_{\triangle}$ be the subset of  $L_\triangle$ composed of those triangles that are incident to external edges of $\Gamma_F$. Only the outer embeddings of the triangles in $L^\prime_{\triangle}$ can induce framework and external conflicts. 
	We translate the Cartesian axes so that the bottom-left corner of the bounding box of $\Gamma_F$ lies on the origin, and give an outer embedding to every triangle in $L^\prime_{\triangle}$.
	This results in a possibly non-planar straight-line realization $\Gamma^\prime_F$ of the graph $G^\prime_F:=G_F \cup L^\prime_{\triangle}$.
	Our strategy to detect external and framework conflicts is based on the construction of a linear-size ``proximity'' graph $H$ defined on the vertices of $G^\prime_F$.
	We show that, if two triangles of $L^\prime_\triangle$ induce an external conflict, then their vertices are associated either with the same node or with adjacent nodes of $H$.
	Similarly, if a triangle $\triangle \in L^\prime_\triangle$ induces a framework conflict by intersecting an edge $e$ of $G^\prime_F$, then the vertices of $\triangle$ and the end-vertices of $e$ are associated either with the same node or with adjacent nodes of $H$.
	We then detect external and framework conflicts by a linear-time traversal of $H$.
	
	In the following, we say that a vertex of $G_F$ is a \emph{framework vertex} and a vertex of $L^\prime_{\triangle}$ that is not contained in $G_F$ is a \emph{leaf vertex}.
	For each leaf vertex $v$, we denote with $\triangle(v)$ the outer embedding of the triangle in $L^\prime_\triangle$ whose leaf vertex is $v$.
	We define the graph $H$ as follows.
	Consider a square grid whose cells have side length $3w_2$ covering the plane; see \cref{fig:external_conflicts_example}.
	We assign a label $l(v) = (i,j)$ to every vertex $v \in V(G^\prime_F)$, where $i$ and $j$ respectively denote the row and the column of one of the (at most four) grid cells that contain $v$ in their interior or on their boundary.
	More precisely, the vertex $v$ gets the label $(\lfloor\frac{x(v)}{3w_2}\rfloor, \lfloor\frac{y(v)}{3w_2}\rfloor)$, where $(x(v),y(v))$ denotes the coordinates of $v$ in $\Gamma'_F$.
	The graph $H$ has a node for each label assigned to at least one vertex of $G'_F$, and two nodes $(i,j)$ and $(i^\prime,j^\prime)$ are connected if and only if $\vert i-i' \vert \leq 1$ and $\vert j-j' \vert \leq 1$.
	Note that $H$ has at most $n$ nodes and maximum degree 8, hence it has $O(n)$ edges.
	We say that a vertex of $G'_F$ is \emph{associated with} a node $\mu$ of $H$ if $\mu$ represents the cell $(i,j)$ and $l(v)=(i,j)$.
	
	\begin{figure}[ht]
		\centering
		
		\includegraphics[page=1,scale=1.19]{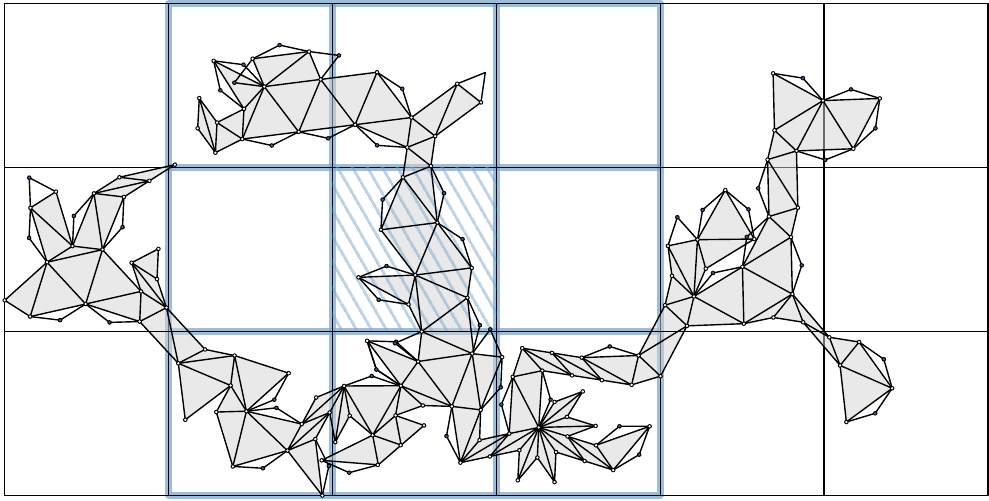}
		\hfil
		\includegraphics[page=2,scale=1.19]{two_distances_external_conflict_graph}
		
		\caption {
			On the left, the straight-line realization $\Gamma^\prime_F$ of the graph $G^\prime_F:=G_F \cup L^\prime_{\triangle}$, where each triangle in $L'_{\triangle}$  has an outer embedding.
			The triangles of $G_F$ are shown in gray, and the outer embeddings of the triangles in $L^\prime_\triangle$ are shown in white.
			Framework vertices are shown in white and leaf vertices are shown in blue.
			The grid cell filled with dashed lines is zoomed-in on the right.
			All the triangles in $L^\prime_{\triangle}$ contained in this cell can induce external or framework conflicts only with the triangles contained in the nine highlighted cells.
		}
		\label{fig:external_conflicts_example}
	\end{figure}
	
	\begin{lemma} \label{lemma:H-properties}
		The graph $H$ satisfies the following properties:
		\begin{enumerate}[(a)]
			\item \label{lemma:H-properties:cells}
			
			For each node $\mu$ of $H$, there are $O(1)$ leaf vertices associated with $\mu$.
			
			\item \label{lemma:H-properties:external}
			
			Let $\triangle_1,\triangle_2$ be two triangles in $L'_\triangle$.
			For $i=1,2$, let $\mu_i$ denote a node of $H$ such that at least one vertex of $\triangle_i$ is associated with $\mu$.
			If the outer embeddings of $\triangle_1$ and $\triangle_2$ induce an external conflict, then either $\mu_1=\mu_2$ or $(\mu_1,\mu_2)\in E(H)$.
			
			\item \label{lemma:H-properties:framework}
			
			Let $\triangle$ be a triangle in $L'_\triangle$, and $e$ be an edge of $G_F$ incident to the outer face of $\Gamma_F$.
			Let $\mu$ and $\nu$ denote two nodes of $H$ such that at least one vertex of $\triangle$ is associated with $\mu$ and at least one end-vertex of $e$ is associated with $\nu$.
			If the outer embedding of $\triangle$ induces a framework conflict by intersecting $e$, then either $\mu_1=\mu_2$ or $(\mu_1,\mu_2)\in E(H)$.
			
		\end{enumerate}
	\end{lemma}
	
	\begin{proof}
		Consider any leaf vertex $v$ associated with a node $\mu$ of $H$.
		Let $(i,j)$ be the cell corresponding to $\mu$ and $\triangle$ be the parent triangle of the triangle in $L'_\triangle$ that has $v$ as a vertex.
		By statement~(\ref{prop:containments:big-equilateral_tall-isosceles}) of \cref{prop:containments}, we have that $\triangle$ is either a big equilateral or a tall isosceles triangle. Further, $\triangle$ is contained in the union of the nine cells $(i',j')$ for $|i-i'| \leq 1$ and $|j-j'| \leq 1$ (see the highlighted cells in \cref{fig:external_conflicts_example}); this is because the Euclidean distance between $v$ and every point of $\triangle$ is smaller than $2w_2$, while the side length of a grid cell is $3w_2$.   Observe that the area of a big equilateral or tall isosceles triangle is larger than the area of a small equilateral triangle, which is equal to $w_1^2\sqrt 3/4$.
		Moreover, any two parent triangles are interior-disjoint in $\Gamma'_F$.
		Since the area of the union of the nine cells is $81w_2^2$, then at most $81w_2^2\frac{4}{w_1^2\sqrt 3} = \frac{324}{\sqrt 3} \left(\frac{w_2}{w_1}\right)^2 \in O(r^2)$ parent triangles are contained in the nine cells, and $O(r^2)$ leaf vertices are associated with $\mu$.
		Since $r < 2$, we have $O(r^2) \in O(1)$.
		This proves Property~(\ref{lemma:H-properties:cells}).
		
		To prove Property~(\ref{lemma:H-properties:external}), suppose, for the sake of contradiction, that the outer embeddings of $\triangle_1$ and $\triangle_2$ induce an external conflict, with $\mu_1 \neq \mu_2$ and $(\mu_1,\mu_2) \notin E(H)$.
		Since $\triangle_1$ and $\triangle_2$ induce an external conflict, then there is at least a point $x$ in $\triangle_1 \cap \triangle_2$.
		Note that $x$ is at distance smaller than or equal to $w_2$ from the vertices of $\triangle_1$ and $\triangle_2$, which implies that every vertex of $\triangle_1$ is at distance smaller than or equal to $2w_2$ from any vertex of $\triangle_2$.
		Since there is a distance larger than or equal to $3w_2$ between two vertices of $G'_F$ associated to non-adjacent nodes of $H$, then either $\mu_1 = \mu_2$ or $(\mu_2,\mu_2) \in E(H)$, which is a contradiction.
		
		Finally, Property~(\ref{lemma:H-properties:framework}) can be proved by similar arguments to those we used to prove Property~(\ref{lemma:H-properties:external}).
	\end{proof}
	
	In order to detect external and framework conflicts, we define each node $\mu$ of $H$ representing the cell $(i,j)$ as a record storing (i) a label $l(\mu)= (i,j)$ and (ii) a list $L(\mu)$ that contains the vertices of $G'_F$ associated with $\mu$.
	We next describe an efficient algorithm to construct $H$.
	
	\begin{lemma} \label{le:H-construction}
		The graph $H$ can be constructed in $O(n)$ time.
	\end{lemma}
	
	\begin{proof}
		We first show how to construct the vertex set $V(H)$ of $H$.
		
		Consider a total order $\pi$ of the vertices of $G'_F$ such that, for any two vertices $u$ and $v$ with $l(u)=(i_u,j_u)$ and $l(v)=(i_v,j_v)$, it holds that $u$ precedes $v$ in the order if (i) $i_u < i_v$ or (ii) $i_u = i_v$ and $j_u < j_v$; if $i_u = i_v$ and $j_u =j_v$, then $u$ and $v$ are in any relative order in $\pi$.   Since $G'_F$ is connected and any edge of $G'_F$ has length at most $w_2$, then $i,j \leq \frac{w_2 n}{3w_2} = \frac{1}{3} n$ for any label $(i,j)$.
		Hence, we can compute~$\pi$ in $O(n)$ time by means of counting sort.
		
		We process the vertices of $G'_F$ in the order in which they appear in $\pi$.
		Let $s$ be the first vertex in the order~$\pi$. We create a node $\nu$ of $H$, set $l(\nu) = l(s)$, and add $s$ to $L(\nu)$.
		Let now~$v$ be the currently considered vertex of $\pi$ and let us denote by $\mu$ the last created node of $H$.
		If $l(v) = l(\mu)$, then we add $v$ to $L(\mu)$; otherwise, we create a node $\mu'$ of $H$, set $l(\mu') = l(v)$, and add $v$ to $L(\mu')$.
		Note that, if two vertices of $G^\prime_F$ have the same label, then they are consecutive in $\pi$.
		Therefore, the above procedure creates a new node of $H$ for each label of a vertex of $G'_F$, and vertices with the same label $(i,j)$ are all placed in the list $L(\mu)$ of the node $\mu$ such that $l(\mu)=(i,j)$. The procedure can clearly be performed in overall $O(n)$ time.
		
		We now show how to construct the edge set $E(H)$ of $H$.
		
		We define a partition of $E(H)$ into four subsets $E_{-}, E_{|}, E_{/}, E_{\backslash}$; refer to \cref{fig:external_conflicts_graph_example}. Consider any edge $(\mu,\nu)\in E(H)$, where we denote $l(\mu)=(i_\mu,j_\mu)$ and $l(\nu)=(i_\nu,j_\nu)$. Then $(\mu,\nu)$ belongs to $E_{-}$, $E_{|}$, $E_{/}$, or $E_{\backslash}$, depending on whether $(\mu,\nu)$ satisfies Property $P_{-}$, $P_{|}$, $P_{/}$, or $P_{\backslash}$, respectively, where:
		
		\begin{description}
			\item[\bf $P_{-}$:] $i_\mu=i_\nu$ and $|j_\mu-j_\nu|=1$,
			\item[\bf $P_{|}$:] $j_\mu=j_\nu$ and $|i_\mu-i_\nu|=1$,
			\item[\bf $P_{/}$:] either $i_\mu-i_\nu = j_\mu-j_\nu = 1$ or $i_\mu-i_\nu= j_\mu-j_\nu= -1$, and
			\item[\bf $P_{\backslash}$:] either $i_\mu-i_\nu= 1$ and $j_\mu-j_\nu= -1$ or $i_\mu-i_\nu= -1$ and $j_\mu-j_\nu= 1$.
		\end{description}

		\begin{figure}[t!]
			\centering
			\includegraphics[page=3,scale=1.45]{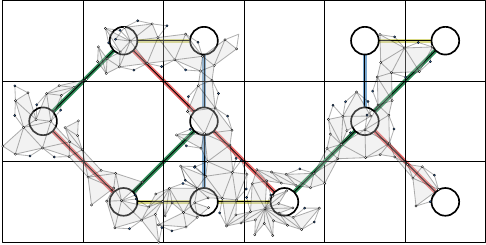}
			\caption{
				Construction of the set of edges $E(H)$.
				The edges on the sets $E_{-}, E_{|}, E_{/}$, and $E_{\backslash}$ are shown in yellow, blue, green, and red, respectively.
			}
			\label{fig:external_conflicts_graph_example}
		\end{figure}
		
		In order to construct the sets $E_{-}, E_{|}, E_{/}$, and $E_{\backslash}$, we define four total orders $\pi_{-}$, $\pi_{|}$, $\pi_{/}$, and $\pi_{\backslash}$ of the nodes of $H$, respectively.
		For any two nodes $\mu$ and $\nu$ of $H$, it holds that $\mu$ precedes $\nu$ in $\pi_{-}$, $\pi_{|}$, $\pi_{/}$, and $\pi_{\backslash}$  if and only if the following conditions hold, respectively:  
		
		\begin{description}
			\item[\bf $C_{-}$:] either $i_\mu<i_\nu$ or $i_\mu=i_\nu$  and $j_\mu < j_\nu$, 
			\item[\bf $C_{|}$:] either $j_\mu<j_\nu$ or $j_\mu=j_\nu$  and $i_\mu < i_\nu$, 
			\item[\bf $C_{/}$:] either $i_\mu-j_\mu < i_\nu-j_\nu$ or $i_\mu-j_\mu = i_\nu-j_\nu$ and $i_\mu<i_\nu$, or
			\item[\bf $C_{\backslash}$:] either $i_\mu+j_\mu < i_\nu+j_\nu$ or $i_\mu+j_\mu = i_\nu+j_\nu$ and $i_\mu<i_\nu$.
		\end{description}
		
		We construct the set $E_i$ with  $i \in \{-,|,/,\backslash\}$ as follows.
		First, we compute the total order $\pi_i$ of the nodes of $H$ for which condition $C_i$ holds for any pair $(\mu,\nu)$ of nodes such that $\mu$ precedes $\nu$ in $\pi_i$.
		Similarly as in the construction of $\pi$, the order $\pi_i$ can be computed in $O(n)$ time by means of counting sort.
		Then we process the nodes of $V(H)$ in the order $\pi_i$ and add a new edge $(\mu,\nu)$ to $E_i$ whenever condition {\bf $P_{i}$} holds for two nodes $\mu$ and $\nu$ such that $\mu$ immediately precedes $\nu$ in $\pi_i$.
		Since, by construction, nodes that are adjacent in $E_i$ are consecutive in~$\pi_i$, it follows that the above procedure correctly computes each set $E_i$. Clearly, the procedure can be performed in overall $O(n)$ time.
		This concludes the construction~of~$H$.
	\end{proof}
	
	We now describe an algorithm, which we call {\sc ConflictFinder}, to detect external and framework conflicts.
	Broadly speaking, the algorithm performs two linear-time traversals of the edges of $H$.
	In the first traversal, we detect external conflicts.
	When processing the edge $(\mu,\nu) \in E(H)$ we detect intersections between the triangles of $L^\prime_\triangle$ whose leaf vertices are contained in $L(\mu) \cup L(\nu)$.
	In the second traversal, we detect framework conflicts.
	In this case, when processing the edge $(\mu,\nu) \in E(H)$ we detect intersections between the triangles of $L^\prime_\triangle$ and the edges of $G^\prime_F$ whose leaf and framework end-vertices, respectively, are contained in $L(\mu) \cup L(\nu)$.
	
	The algorithm consists of the following steps:
	
	
	
	
	
	
	\begin{enumerate}[\bf Step 1:]
		
		\item
		
		Compute the drawing $\Gamma^\prime_F$ of $G'_F$ and the auxiliary graph $H$ as described above.
		
		\item
		
		For every framework vertex $v$, create the list $L_E(v)$ of the two edges of $\Gamma_F$ incident to both $v$ and the outer face of $\Gamma_F$.
		
		\item
		
		For every node $\mu \in V(H)$, create the lists $L_\ell(\mu)$ and $L_F(\mu)$ of leaf and framework vertices contained in $L(\mu)$, respectively.
		
		\item
		
		Detect which triangles of $L^\prime_\triangle$ induce external conflicts as follows.
		Let $\mathcal{E}_\ell = \{ (u,v) \in L_\ell(\mu) \times \{ L_\ell(\mu) \cup L_\ell(\nu) \} \, \vert \, (\mu,\nu) \in E(H) \}$.
		For every pair of vertices $(u,v) \in \mathcal{E}_\ell$, test if $\triangle(u)$ and $\triangle(v)$ intersect each other.
		In the positive case, the triangles $\triangle(u)$ and $\triangle(v)$ induce an external conflict.
		
		\item
		
		Detect which triangles of $L^\prime_\triangle$ induce framework conflicts as follows.
		Let $\mathcal{E}_F = \{ (u,v) \in L_F(\mu) \times \{ L_\ell(\mu) \cup L_\ell(\nu) \} \, \vert \, (\mu,\nu) \in E(H) \}$.
		For every pair of vertices $(u,v) \in \mathcal{E}_F$, test if there is an edge $e \in L_E(u)$ such that $e$ intersects $\triangle(v)$.
		In the positive case, the triangle $\triangle(v)$ induces a framework conflict with the edge $e$.
		
	\end{enumerate}
	
	The correctness and complexity of the algorithm are proved in the following lemma.
	
	\begin{lemma}\label{lem:conflictfinder}
		The algorithm {\sc ConflictFinder} detects all the external and framework conflicts in $O(n)$ time.
	\end{lemma}
	
	\begin{proof}
		
		We first prove the correctness of the algorithm.
		Consider first the detection of external conflicts.
		Let $u$ and $v$ be two leaf vertices associated respectively to two nodes $\mu$ and $\nu$.
		Remember that $\triangle(u)$ and $\triangle(v)$ induce an external conflict if and only if they intersect each other.
		Then, by Property~(\ref{lemma:H-properties:external}) of \cref{lemma:H-properties}, we have that either $\mu = \nu$ or $(\mu, \nu) \in E(H)$.
		Since in Step~4 the algorithm processes each pair of leaf triangles whose leaf vertices are associated with the same node of $H$ or with adjacent nodes of $H$, it follows that all the external conflicts are being detected.
		
		Consider now the detection of framework conflicts.
		Let $u$ be a framework vertex associated to a node $\mu$, let $v$ be a leaf vertex associated to a node $\nu$, and $e$ be an edge in $L_E(u)$.
		Remember that $\triangle(v)$ induces a framework conflict with $e$ if and only if $\triangle(v)$ and $e$ intersect each other.
		Then, by Property~(\ref{lemma:H-properties:framework}) of \cref{lemma:H-properties}, we have that either $\mu = \nu$ or $(\mu, \nu) \in E(H)$.
		Since in Step~5 the algorithm processes each pair composed of a framework and a leaf vertex that are associated either with the same node of $H$ or with adjacent nodes of $H$, it follows that all the framework conflicts are being detected.
		
		We now bound the time complexity of the algorithm. This is done by bounding the time complexity of each of the Steps~1--5.
		\begin{itemize}
			\item Since there are $O(n)$ triangles in $L^\prime_\triangle$, we have that $\Gamma_F$ can be extended to $\Gamma^\prime_F$ in $O(n)$ time.
			Furthermore, the graph $H$ can be computed in $O(n)$ time by \cref{le:H-construction}.
			Hence Step~1 takes $O(n)$ time.
			\item Step~2 takes $O(n)$ time since the boundary of the outer face of $\Gamma_F$ consists of $O(n)$ line segments.
			\item Step~3 takes also $O(n)$ time, since $H$ has at most $n$ nodes, and every vertex $v$ of $G^\prime_F$ is contained in the list $L(\mu)$ of a single node $\mu$ of $H$.
			\item The time complexity of Step~4 is bounded by $O(\vert \mathcal{E}_\ell \vert)$. Note that $H$ contains $O(n)$ edges; moreover, for every node $\mu$ of $H$, the list $L_\ell(\mu)$ contains $O(1)$ leaf vertices, by Property~(\ref{lemma:H-properties:cells}) of \cref{lemma:H-properties}.
			Hence $\mathcal{E}_\ell$ contains $O(n)$ pairs of vertices and Step~4 takes $O(n)$ time.
			\item Finally, the time complexity of Step~5 is bounded by $O(\vert \mathcal{E}_F \vert)$.
			In this case, even though $\Theta(n)$ framework vertices might be associated with the same node of $H$, each framework vertex appears in $O(1)$ pairs of $\mathcal{E}_F$; namely, if a framework vertex $u$ is associated with a node $\mu$ of $H$, then it only appears in the pairs $(u,v)$ of $\mathcal{E}_F$ such that the leaf vertex $v$ is associated either with $\mu$ or with one of the (at most $8$) neighbors of $\mu$ in $H$. By Property~(\ref{lemma:H-properties:cells}) of \cref{lemma:H-properties}, there are $O(1)$ leaf vertices associated with the same node of $H$, hence $\mathcal{E}_F$ contains $O(n)$ pairs of vertices and Step~5 takes $O(n)$ time.
		\end{itemize}

		This concludes the proof of the time complexity and hence the proof of the lemma.
	\end{proof}
	
	\paragraph*{Step G. Extending $\Gamma_F$ to a planar straight-line realization of $G$.}
	
	The final step is to draw the triangles of $L_{\triangle}$ in $\Gamma_F$ without inducing conflicts, so that the resulting drawing remains planar.
	We decide if this is possible by means of a 2SAT formula.
	Each leaf triangle $\triangle_i \in L_{\triangle}$ can have two  embeddings (which differ from each other for the position of the leaf vertex of $\triangle_i$).
	We arbitrarily associate each embedding with a distinct truth value of a suitable Boolean variable, and construct the clauses of the formula as follows:
	\begin{itemize}
		\item If a drawing of $\triangle_i$ induces an overlapping conflict, then we introduce a clause $(\neg\ell_i)$, where $\ell_i$ is the literal that is \texttt{True} when $\triangle_i$ has the drawing that induces an overlapping conflict.
		\item If $\triangle_i$ induces a framework conflict, we introduce a clause $(\neg\ell_i)$, where $\ell_i$ is the literal that is \texttt{True} when $\triangle_i$ has an outer embedding.
		
		\item If two leaf triangles $\triangle_i,\triangle_j \in L_{\triangle}$ induce an internal conflict inside a triangle $\triangle\in\Gamma_F$, we introduce a clause $(\neg\ell_i \vee \neg\ell_j)$, where $\ell_i$ (resp. $\ell_j$) is the literal that is \texttt{True} when $\triangle_i$ (resp. $\triangle_j$) is drawn inside $\triangle$.
		
		\item If two leaf triangles $\triangle_i,\triangle_j \in L_{\triangle}$ induce an external conflict, we introduce a clause $(\neg\ell_i \vee \neg\ell_j)$, where $\ell_i$ (resp. $\ell_j$) is the literal that is \texttt{True} when $\triangle_i$ (resp. $\triangle_j$) has an outer embedding.

	\end{itemize}
	
	Let $\phi$ be the 2SAT formula produced by the clauses introduced as described above.
	Clearly, the clauses of $\phi$ describe all and only the overlapping, framework, internal, and external conflicts defined in \cref{sec:conflicts}.
	Moreover, since there are $O(n)$ triangles in $L_{\triangle}$, the number of framework and overlapping conflicts is in $O(n)$. Further, since $H$ has $O(n)$ vertices and edges, by \cref{lem:internal_conflict_bound,lemma:H-properties} the number of internal and external conflicts is in $O(n)$, respectively. It follows that $\phi$ contains $O(n)$ clauses. The following lemma shows that a solution to $\phi$ corresponds to a planar straight-line realization of $G$.
	
	\begin{lemma}\label{lem:2sat_formula}
		The formula $\phi$ is satisfiable if and only if $\Gamma_F$ can be extended to a planar straight-line realization of $G$.
		Also, if $\phi$ is satisfiable, then such a planar straight-line realization can be constructed in $O(n)$ time.
	\end{lemma}
	\begin{proof}
		We prove the two directions separately.
		
		\vspace{0.5em}%
		\noindent%
		$(\Longrightarrow)$ Assume that $\phi$ is satisfiable.
		We show that $\Gamma_F$ can be extended to a planar straight-line realization of $G$.
		Consider a truth assignment for the Boolean variables that satisfies $\phi$.
		The variables of $\phi$ are in one to one correspondence with the triangles of $L_{\triangle}$.
		We draw in $\Gamma_F$ each triangle $\triangle_i \in L_{\triangle}$ with the embedding associated with the truth value of the variable corresponding to $\triangle_i$ in the satisfying assignment for $\phi$. This can clearly be done in $O(1)$ time per triangle in $L_{\triangle}$, hence in overall $O(n)$ time.
		Since each clause of $\phi$ is satisfied, it follows that $\triangle_i$ is drawn in $\Gamma_F$ without inducing conflicts. Since this is true for every $\triangle_i \in L_{\triangle}$, it follows that the resulting straight-line realization of $G$ is planar.
		Furthermore, such a realization can be constructed in $O(n)$ time since there are $O(n)$ clauses in $\phi$.
		
		\vspace{0.5em}%
		\noindent%
		$(\Longleftarrow)$ Let $\Gamma$ be a planar straight-line realization of $G$.
		We show that $\phi$ is satisfiable.
		The embedding in $\Gamma$ of each triangle $\triangle_i \in L_{\triangle}$ identifies a truth value for the corresponding variable in $\phi$, and thus, a truth assignment for $\phi$.
		Such an assignment satisfies $\phi$ since an unsatisfied clause would imply that we have either an overlapping, or a framework, or an internal, or an external conflict, which in turn implies that the boundary of some triangle $\triangle_i \in L_{\triangle}$ crosses the boundary of a triangle of $\Gamma_F$ or the boundary of another triangle of $L_{\triangle}$, contradicting the fact that $\Gamma$ is planar.
		This concludes the proof.
	\end{proof}
	
	By the discussion in Steps D.2, E, and F, we have that $\phi$ can be constructed in $O(n)$ time. Further, as argued above, $\phi$ contains $O(n)$ clauses. We can therefore test whether $\phi$ is satisfiable in $O(n)$ time~\cite{tarjan_1979}.
	If $\phi$ is not satisfiable, we reject the instance.
	Otherwise, we extend $\Gamma_F$ to a planar straight-line realization of $G$ in $O(n)$ time as described in the proof of \cref{lem:2sat_formula}. This concludes the proof of \cref{th:two-weights}.
	
	\section{Maximal Outerplanar Graphs}
	\label{sec:outerplanar}
	
	In this section we study the \FEPRshort problem for weighted outerplanar
	$2$-trees, that is, for weighted maximal outerplanar graphs. While we could not establish the computational complexity of the \FEPRshort problem for general maximal outerplanar graphs, we can prove the following two theorems. 
	
	\begin{theorem}\label{th:outerpaths}
		Let $G$ be an $n$-vertex weighted maximal outerpath. There exists an
		$O(n)$-time algorithm that tests whether $G$ admits a planar
		straight-line realization and, in the positive case, constructs such
		a realization.
	\end{theorem}
	
	\begin{theorem}\label{th:outerpillars}
		Let $G$ be an $n$-vertex weighted maximal outerpillar. There exists
		an $O(n^3)$-time algorithm that tests whether $G$ admits a planar
		straight-line realization and, in the positive case, constructs such
		a realization.
	\end{theorem}
	
	\newenvironment{proofA}{\noindent{\em Proof of \cref{th:outerpaths}.}}{\hspace*{\fill}${\small\qed}$\vspace{2mm}}
	\newenvironment{proofB}{\noindent{\em Proof of \cref{th:outerpillars}.}}{\hspace*{\fill}${\small\qed}$\vspace{2mm}}
	
	\begin{figure}[htb]
		\centering
		\subcaptionbox{}
		{\includegraphics[scale=1]{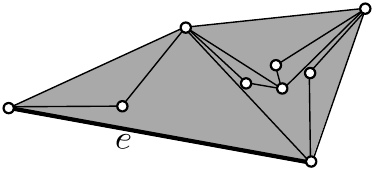}}%
		\hfil
		\subcaptionbox{}
		{\includegraphics[scale=1]{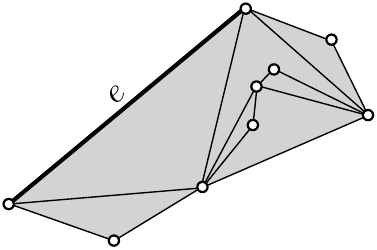}}%
		\hfil
		\subcaptionbox{}
		{\includegraphics[scale=1]{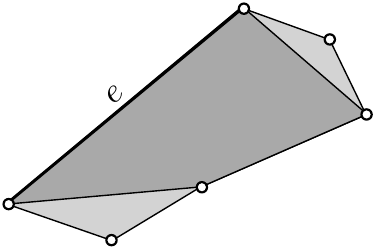}}%
		\caption{(a) and (b) show two $e$-outer realizations $\Gamma$ and $\Gamma'$ of a weighted $2$-tree. The gray regions show $\mathcal I_{\Gamma}$ and $\mathcal I_{\Gamma'}$. We have that $\Gamma$ is $e$-smaller than $\Gamma'$. Indeed, as shown in (c), $\Gamma$ can be rotated,
			translated, and reflected so that the representation of $e$
			coincides with the one in $\Gamma'$ and so that its interior is a
			subset of the interior of $\Gamma'$.}
		\label{fig:outerpath-smaller}
	\end{figure}
	
	Before proving~\cref{th:outerpaths,th:outerpillars}, we establish some
	notation and definitions that are common to both proofs. Let $G$ be a
	weighted $2$-tree. For a planar straight-line realization $\Gamma$ of
	$G$, the \emph{interior} of $\Gamma$, denoted by
	$\mathcal I_{\Gamma}$, is the complement of the outer face of
	$\Gamma$, i.e., the union of the closure of the internal faces of
	$\Gamma$. Let $e$ be an edge of $G$. An \emph{$e$-outer realization}
	of $G$ is a planar straight-line realization of $G$ such that $e$ is
	incident to its outer face. Given two $e$-outer realizations $\Gamma$
	and $\Gamma'$ of $G$, we say that $\Gamma$ is \emph{$e$-smaller than}
	$\Gamma'$ if there exists an $e$-outer realization $\Gamma''$ that can
	be obtained by a rigid transformation of $\Gamma$ such that the
	end-vertices of $e$ lie at the same points in $\Gamma''$ and in
	$\Gamma'$ and such that
	$\mathcal I_{\Gamma''}\subseteq \mathcal I_{\Gamma'}$ (see \cref{fig:outerpath-smaller}); that is
	$\Gamma$ is $e$-smaller than $\Gamma'$ if it can be rotated,
	translated, and possibly reflected so that the representation of $e$
	coincides with the one in $\Gamma'$ and so that its interior is a
	subset of the interior of $\Gamma'$. A set $\mathcal R$ of $e$-outer
	realizations of $G$ is \emph{$e$-optimal} if, for every $e$-outer
	realization $\Gamma'$ of $G$, there exists an $e$-outer realization
	$\Gamma$ of $G$ in $\mathcal R$ such that $\Gamma$ is $e$-smaller than
	$\Gamma'$. An $e$-outer realization $\Gamma$ of $G$ is
	\emph{$e$-optimal} if the set $\{\Gamma\}$ is $e$-optimal. 
	
	The following lemma will be very useful; see~\cref{fig:outer-replacement}. Roughly speaking, it asserts that, in order to construct a planar straight-line realization of a weighted $2$-tree, for certain subgraphs it is not a loss of generality to use realizations that are ``as $e$-small as possible''. 
	
	\begin{lemma} \label{le:optimal-replacement} Let $G$ be a weighted
		$2$-tree that admits a planar straight-line realization
		$\mathcal G$. Let $H$ and $K$ be two subgraphs of $G$ that satisfy
		the following properties:
		\begin{itemize}
			\item $H$ and $K$ are weighted $2$-trees;
			\item $H\cup K = G$ (that is, every vertex and edge of $G$ is in $H$
			or in $K$);
			\item $H\cap K = e$ (that is, $H$ and $K$ share an edge $e$, the
			end-vertices of $e$, and no other edge or vertex);
			\item the restriction $\mathcal H$ of $\mathcal G$ to $H$ is an
			$e$-outer realization of $H$; and
			\item no vertex of $K$ lies inside an internal face of $\mathcal H$ in
			$\mathcal G$.
		\end{itemize}
		Let $\mathcal H'$ be any $e$-outer realization of $H$ that is
		$e$-smaller than $\mathcal H$. Then it is possible to replace
		$\mathcal H$ with an $e$-outer realization of $H$ that can be obtained
		by a rigid transformation of $\mathcal H'$ in such a way that the
		resulting straight-line realization of $G$ is planar.
	\end{lemma}
	
	\begin{figure}[tb]
		\centering
		\subcaptionbox{}
		{\includegraphics[scale=1]{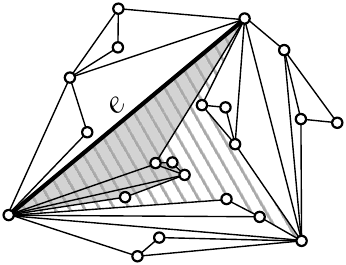}}%
		\hfil
		\subcaptionbox{}
		{\includegraphics[scale=1]{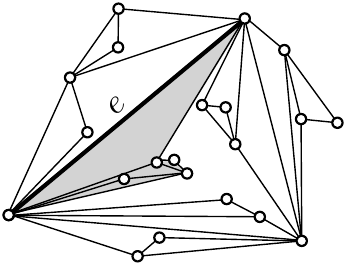}}%
		\caption{(a) The drawing $\mathcal G$; the set $\mathcal I_{\mathcal H}$ is shaded gray. The region with dark gray strips is the face $f$ of $\mathcal K$. (b) The drawing obtained from $\mathcal G$ by replacing $\mathcal H$ with $\mathcal H''$; the set $\mathcal I_{\mathcal H''}$ is gray.}
		\label{fig:outer-replacement}
	\end{figure}
	
	\begin{proof}
		Let $\mathcal K$ be the restriction of $\mathcal G$ to $K$. Assume
		that $H$ has more than two vertices, as otherwise $H$ contains a
		single edge, hence $\mathcal H'$ can be obtained by a rigid
		transformation of $\mathcal H$ and the statement is trivial. By
		assumption, no vertex of $K$ lies in $\mathcal I_{\mathcal H}$;
		further, $H$ is biconnected, by Property (P\ref{pr:biconnected}) of a $2$-tree. It follows
		that there is a face $f$ of $\mathcal K$ that contains $\mathcal H$
		entirely, except for the edge $e$ and its end-vertices, which are
		shared by $H$ and $K$. By assumption, there exists an $e$-outer
		realization $\mathcal H''$ of $H$ that can be obtained by a rigid
		transformation of $\mathcal H'$ such that
		$\mathcal I_{\mathcal H''}\subseteq \mathcal I_{\mathcal H}$, hence
		the part of $\mathcal H''$ different from the edge $e$ does not
		intersect $\mathcal K$; the edge $e$ does not intersect $\mathcal K$
		either, as it is part of it and $\mathcal K$ is planar. Hence, the
		straight-line realization of $G$ resulting from the replacement of
		$\mathcal H$ with $\mathcal H''$ in $\mathcal G$ is planar.
	\end{proof}

	We are now ready to prove~\cref{th:outerpaths,th:outerpillars}.
	
	\medskip
	\begin{proofA}
		Let $G=(V,E,\lambda)$ be an $n$-vertex weighted maximal outerpath
		(as the one in~\cref{fig:outerpath-structure}) and let $\mathcal O$
		be its outerplane embedding, which can be found in $O(n)$
		time~\cite{d-iroga-07,m-laarogmog-79,w-rolt-87}. Since $G$ is an
		outerpath, its dual tree is a path, which we denote by
		$P:=(p_1,\dots,p_k)$, where $k=n-2$; note that $P$ can be easily
		recovered from $\mathcal O$ in $O(n)$ time. Each node $p_i$ of $P$
		corresponds to an internal face of $\mathcal O$; we denote by $c_i$
		the $3$-cycle of $G$ bounding such a face; each cycle $c_i$ has a unique (planar)
		straight-line realization (up to a rigid transformation), which we
		denote by $\mathcal C_i$. Further, for $i=1,\dots,k-1$, let $e_i$ be
		the edge of $G$ that is dual to the edge $(p_i,p_{i+1})$ of $P$;
		note that $e_i$ is an internal edge of $\mathcal O$.
		
		\begin{figure}[ht]
			\centering
			\includegraphics[scale=1]{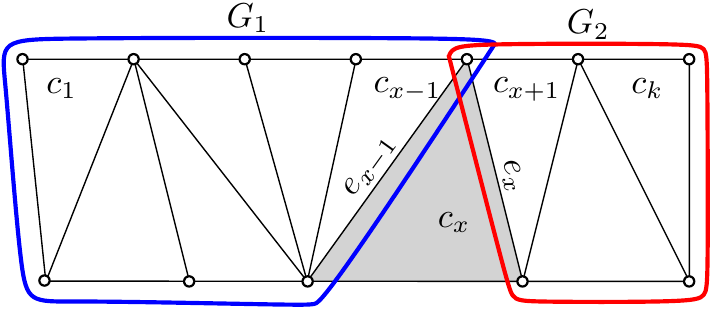}%
			\caption{The outerplane embedding of an outerpath. The gray face is the one bounded by $c_x$. The graphs $G_1$ and $G_2$ are enclosed inside fat curves.}
			\label{fig:outerpath-structure}
		\end{figure}
		
		Let $x\in\{1,\dots,k\}$ be an index such that the sum of the
		lengths of the edges of $c_{x}$ is maximum. Let $G_1$ be the
		subgraph of $G$ composed of the $3$-cycles $c_1,c_2,\dots,c_{x-1}$
		and let $G_2$ be the subgraph of $G$ composed of the $3$-cycles
		$c_{x+1},c_{x+2},\dots,c_k$. One of these graphs might be undefined
		if $x=1$ or $x=k$; however, in the following, we assume that $1<x<k$
		as the arguments for the cases $x=1$ and $x=k$ are analogous and
		actually simpler. Note that $G_1$ and $G_2$ are maximal outerpaths.
		
		Assume that a planar straight-line realization $\Gamma$ of $G$
		exists; for $i=1,2$, let $\Gamma_i$ be the restriction of $\Gamma$
		to $G_i$. Since the sum of the lengths of the edges of $c_x$ is
		maximum among all the $3$-cycles of $G$, it follows that $c_x$ does
		not lie inside any $3$-cycle of $G$ in $\Gamma$; hence, $\Gamma_1$
		is an $e_{x-1}$-outer realization of $G_1$ and $\Gamma_2$ is an
		$e_x$-outer realization of $G_2$.
		
		A key ingredient of our algorithm is the fact that, if there exists an $e_{x-1}$-outer realization of $G_1$ (resp.\ an $e_x$-outer realization of $G_2$), then there exists an {\em $e_{x-1}$-optimal} realization of $G_1$ (resp.\ an {\em $e_x$-optimal} realization of $G_2$).  Indeed, in the following, we show an $O(n)$-time algorithm that either (i)
		concludes that $G_1$ admits no $e_{x-1}$-outer realization or that $G_2$ admits no $e_x$-outer realization, and thus $G$ admits no planar straight-line realization, or (ii) constructs an
		$e_{x-1}$-optimal realization $\Gamma_1$ of $G_1$ and its plane
		embedding $\mathcal E(\Gamma_1)$, and an $e_{x}$-optimal realization
		$\Gamma_2$ of $G_2$ and its plane embedding $\mathcal E(\Gamma_2)$.
		
		\begin{figure}[b!]
			\centering
			{\includegraphics[width=\textwidth]{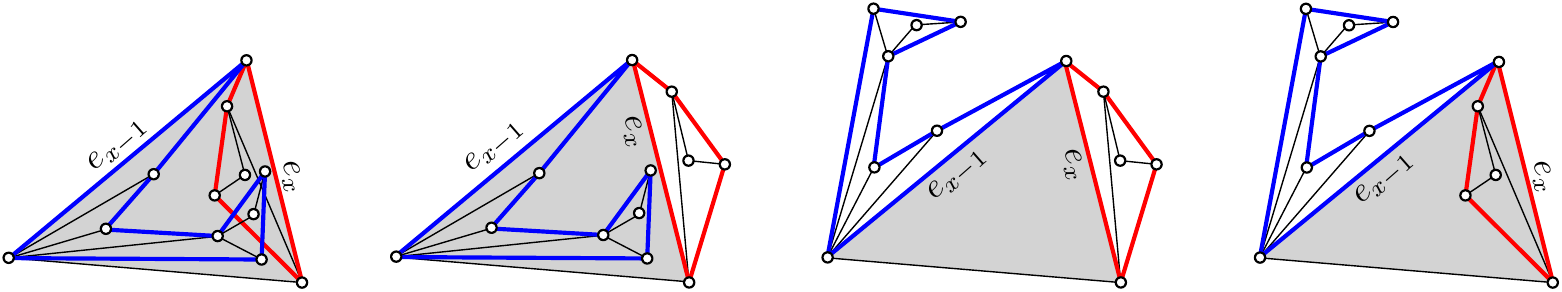}}%
			
			\caption{The four different ways to combine $\Gamma_1$ and
				$\Gamma_2$ with $\mathcal C_x$. In this example, the first
				realization is not planar, while the other three are.}
			\label{fig:outerpath-drawing}
		\end{figure}
		
		In view of the argument above, the existence of such an algorithm
		suffices to conclude the proof of the theorem. Indeed, if the
		algorithm concludes that $G_1$ admits no $e_{x-1}$-outer realization
		or that $G_2$ admits no $e_x$-outer realization, then $G$ admits no
		planar straight-line realization. Otherwise, the algorithm
		constructs an $e_{x-1}$-optimal realization $\Gamma_1$ of $G_1$,
		together with its plane embedding $\mathcal E(\Gamma_1)$, and an
		$e_{x}$-optimal realization $\Gamma_2$ of $G_2$, together with its
		plane embedding $\mathcal E(\Gamma_2)$. Then, in order to test
		whether $G$ admits a planar straight-line realization, we consider
		the unique (up to a rigid transformation) planar straight-line
		realization $\mathcal C_x$ of $c_x$ and we try to combine it with
		$\Gamma_1$ and $\Gamma_2$ (see~\cref{fig:outerpath-drawing}). This
		combination can be done in four ways. Indeed, once $\Gamma_1$ and
		$\mathcal C_x$ coincide on the representation of $e_{x-1}$, it still
		remains to decide whether the vertex of $c_{x-1}$ not incident to
		$e_{x-1}$ and the vertex of $c_x$ not incident to $e_{x-1}$ lie on
		different sides or on the same side of the line through $e_{x-1}$,
		and similarly for $\Gamma_2$ and $\mathcal C_x$. These two
		independent choices result in four straight-line realizations of
		$G$. Each of these realizations, say $\Gamma$, can be constructed in
		$O(n)$ time; further, its plane embedding $\mathcal E(\Gamma)$ can
		also be constructed in $O(n)$ time from $\mathcal E(\Gamma_1)$ and
		$\mathcal E(\Gamma_2)$. Then it can be tested in $O(n)$ time whether $\Gamma$ is a planar straight-line realization respecting $\mathcal E(\Gamma)$, due to
		\cref{thm:straight-line_realization_planarity}. If (at least) one of
		the four realizations is planar, then it is the desired planar
		straight-line realization of $G$. Otherwise, we claim that there
		exists no planar straight-line realization of $G$. For a
		contradiction, assume that there exists such a realization, and
		denote it by $\Gamma$. Since $\Gamma_1$ is $e_{x-1}$-optimal, it is
		$e_{x-1}$-smaller than the restriction of $\Gamma$ to $G_1$; hence,
		we can apply \cref{le:optimal-replacement} with $H=G_1$,
		$K=G_2\cup c_x$, $\mathcal G=\Gamma$, $e=e_{x-1}$, and
		$\mathcal H'= \Gamma_1$, in order to obtain a planar straight-line
		realization $\Gamma'$ of $G$ whose restriction to $G_1$ is
		$\Gamma_1$. Further, since $\Gamma_2$ is $e_{x}$-optimal, it is
		$e_{x}$-smaller than the restriction of $\Gamma'$ to $G_2$; hence,
		we can apply \cref{le:optimal-replacement} with $H=G_2$,
		$K=G_1\cup c_x$, $\mathcal G=\Gamma'$, $e=e_x$, and
		$\mathcal H'= \Gamma_2$, in order to obtain a planar straight-line
		realization $\Gamma''$ of $G$ whose restriction to $G_1$ is
		$\Gamma_1$ and whose restriction to $G_2$ is $\Gamma_2$; however,
		this is one of the four straight-line realizations for which the
		planarity testing was negative, a contradiction.
		
		It remains show an $O(n)$-time algorithm that either concludes that
		$G_1$ admits no $e_{x-1}$-outer realization, or determines an
		$e_{x-1}$-optimal realization of $G_1$. A very similar application
		of the algorithm allows us to either conclude that $G_2$ admits no
		$e_{x}$-outer realization or to determine an $e_{x}$-optimal
		realization of $G_2$.
		
		For $i=1,\dots,x-1$, let $G^i_1$ be the subgraph of $G_1$ composed
		of the $3$-cycles $c_1,c_2,\dots,c_i$; note that $G^{x-1}_1=G_1$. A
		key observation for our algorithm is the following. Assume that an
		$e_{x-1}$-outer realization $\Gamma_1$ of $G_1$ exists; let
		$\Gamma^i_1$ be the restriction of $\Gamma_1$ to $G^i_1$. Then
		$\Gamma^i_1$ is an $e_i$-outer realization of $G^i_1$. This is
		trivial when $i=x-1$; further, when $i<x-1$, if $e_i$ were not
		incident to the outer face of $\Gamma^i_1$, then $e_{x-1}$ could not
		be incident to the outer face of $\Gamma_1$ without violating the
		planarity of $\Gamma_1$.
		
		Our algorithm works by induction on $i$ in order to either conclude
		that $G^i_1$ admits no $e_i$-outer realization (and thus $G_1$
		admits no $e_{x-1}$-outer realization), or to determine an
		$e_i$-optimal realization $\Gamma^i_1$ of $G^i_1$. This is trivial
		when $i=1$, namely $G^1_1$ is the $3$-cycle $c_1$ which has a unique
		planar straight-line realization $\Gamma^1_1:=\mathcal C_1$, up to a
		rigid transformation; then $\Gamma^1_1$ is indeed an $e_1$-optimal
		realization.
		
		Assume that we have an $e_{i-1}$-optimal realization
		$\Gamma^{i-1}_1$ of $G^{i-1}_1$, for some $i\in \{2,\dots,x-1\}$;
		indeed, if we already concluded that $G_1$ admits no $e_{x-1}$-outer
		realization, there is nothing else to do.
		
		Note that the cycles $c_{i-1}$ and $c_i$ share the edge $e_{i-1}$;
		let $v_{i-1}$ and $v_i$ be the vertices of $c_{i-1}$ and $c_i$,
		respectively, not incident to $e_{i-1}$. Any straight-line
		realization of $G^i_1$ is of one of two types: it is a
		\emph{different-side} realization if $v_{i-1}$ and $v_i$ lie on
		different sides of the line through $e_{i-1}$, while it is a
		\emph{same-side} realization if $v_{i-1}$ and $v_i$ lie on the same
		side of the line through $e_{i-1}$.
		
		In order to determine whether an $e_i$-outer realization of $G^i_1$
		exists, we combine $\Gamma^{i-1}_1$ with $\mathcal C_i$ in two ways,
		so to form a different-side realization and a same-side
		realization. Namely, $G^{i-1}_1$ and $c_i$ share the edge $e_{i-1}$;
		hence, once $\Gamma^{i-1}_1$ and $\mathcal C_i$ coincide on the
		representation of $e_{i-1}$, it still remains to decide whether
		$v_{i-1}$ and $v_i$ lie on different sides or on the same side of
		the line through $e_{i-1}$. This choice leads to two straight-line
		realizations of $G^i_1$, which we respectively denote by
		$\Gamma^{i,a}_1$ and $\Gamma^{i,b}_1$.
		
		We check each of $\Gamma^{i,a}_1$ and $\Gamma^{i,b}_1$ for planarity and discard each realization that is not planar. Further, if $\Gamma^{i,b}_1$ was not discarded, then the triangles $\mathcal C_{i-1}$ and $\mathcal C_i$ are one nested into the other one (except for the edge $e_{i-1}$, which is common to both triangles) in such a realization. If $\mathcal C_i$ is nested into $\mathcal C_{i-1}$, then we discard $\Gamma^{i,b}_1$, as it is not an $e_i$-outer
		realization of $G^i_1$.  We have now three cases:
		
		\begin{itemize}
			\item If we did not discard $\Gamma^{i,b}_1$, then we set
			$\Gamma^i_1=\Gamma^{i,b}_1$. Indeed, if we did not discard
			$\Gamma^{i,b}_1$, then $\Gamma^{i,b}_1$ is a planar straight-line
			realization of $G^i_1$ in which $\mathcal C_{i-1}$ is nested into
			$\mathcal C_i$; this implies that all of $\Gamma^{i-1}_1$ lies
			inside $\mathcal C_i$, hence $\mathcal I_{\Gamma^i_1}$ coincides
			with the closure of the interior of $\mathcal C_i$. It follows
			that $\Gamma^i_1$ is an $e_i$-optimal realization of $G^i_1$,
			given that in any planar straight-line realization of $G^i_1$ the
			cycle $c_i$ is represented by the triangle $\mathcal C_i$.
			
			\item If we discarded $\Gamma^{i,b}_1$ but we did not discard
			$\Gamma^{i,a}_1$, then we set $\Gamma^i_1=\Gamma^{i,a}_1$. We
			prove that $\Gamma^i_1$ is an $e_i$-optimal realization of
			$G^i_1$.
			
			First, we prove that there exists no same-side $e_i$-outer
			realization of $G^i_1$. This is obvious if $\Gamma^{i,b}_1$ was
			discarded because $\mathcal C_i$ is nested into $\mathcal C_{i-1}$, as the
			same would be true in any same-side realization of
			$G^i_1$. Suppose hence that $\Gamma^{i,b}_1$ was discarded because
			it is not planar and suppose, for a contradiction, that there
			exists a same-side $e_i$-outer realization $\Psi^i_1$ of
			$G^i_1$. Since $\Gamma^{i-1}_1$ is $e_{i-1}$-optimal, it is
			$e_{i-1}$-smaller than the restriction of $\Psi^i_1$ to
			$G^{i-1}_1$; hence, we can apply \cref{le:optimal-replacement}
			with $H=G^{i-1}_1$, $K=c_i$, $\mathcal G=\Psi^i_1$, $e=e_{i-1}$,
			and $\mathcal H'= \Gamma^{i-1}_1$, in order to obtain a planar
			straight-line realization $\Gamma^{i,*}_1$ of $G$ whose
			restriction to $G^{i-1}_1$ is $\Gamma^{i-1}_1$. Since $\Psi^i_1$
			is a same-side realization of $G^i_1$, we have that
			$\Gamma^{i,*}_1$ is also a same-side realization of
			$G^i_1$. Hence, $\Gamma^{i,*}_1$ is the same realization as
			$\Gamma^{i,b}_1$, a contradiction to the fact that
			$\Gamma^{i,b}_1$ is not planar.
			
			Second, since we did not discard $\Gamma^{i,a}_1$, we have that
			$\Gamma^i_1=\Gamma^{i,a}_1$ is planar. This, together with the
			fact that $\Gamma^i_1$ is a different-side realization, implies
			that $\Gamma^i_1$ is an $e_i$-outer realization. The proof that
			$\Gamma^i_1$ is an $e_i$-optimal realization of $G^i_1$ easily
			follows by induction. Namely, as proved above, any $e_i$-outer
			realization of $G^i_1$ is a different-side realization. Consider
			any such realization $\Psi^i_1$ and let $\Psi^{i-1}_1$ be the
			restriction of $\Psi^i_1$ to the vertices and edges of
			$G^{i-1}_1$. By induction, we have that $\Gamma^{i-1}_1$ is
			$e_{i-1}$-smaller than $\Psi^{i-1}_1$, hence $\Gamma^{i}_1$ is
			$e_{i}$-smaller than $\Psi^{i}_1$, given that the realization
			$\mathcal C_i$ of $c_i$ coincides in $\Gamma^{i}_1$ and
			$\Psi^{i}_1$.

			\item Finally, if we discarded both $\Gamma^{i,a}_1$ and
			$\Gamma^{i,b}_1$, we conclude that there exists no $e_i$-outer
			realization of $G^i_1$. Namely, the proof that there exists no
			same-side $e_i$-outer realization of $G^i_1$ is the same as for
			the case in which we discarded $\Gamma^{i,b}_1$ but we did not
			discard $\Gamma^{i,a}_1$. The proof that there exists no
			different-side $e_i$-outer realization of $G^i_1$ is also very similar to the proof for a same-side realization. Indeed, that
			we discarded $\Gamma^{i,a}_1$ was unequivocally due to the fact
			that it is not planar. Suppose, for a contradiction, that there
			exists a different-side $e_i$-outer realization $\Psi^i_1$ of
			$G^i_1$. Since $\Gamma^{i-1}_1$ is $e_{i-1}$-optimal, it is
			$e_{i-1}$-smaller than the restriction of $\Psi^i_1$ to
			$G^{i-1}_1$; hence, we can apply \cref{le:optimal-replacement}
			with $H=G^{i-1}_1$, $K=c_i$, $\mathcal G=\Psi^i_1$, $e=e_{i-1}$,
			and $\mathcal H'= \Gamma^{i-1}_1$, in order to obtain a planar
			straight-line realization $\Gamma^{i,*}_1$ of $G$ whose
			restriction to $G^{i-1}_1$ is $\Gamma^{i-1}_1$. Since $\Psi^i_1$
			is a different-side realization of $G^i_1$, we have that
			$\Gamma^{i,*}_1$ is also a different-side realization of
			$G^i_1$. Hence, $\Gamma^{i,*}_1$ is the same realization as
			$\Gamma^{i,a}_1$, a contradiction to the fact that
			$\Gamma^{i,a}_1$ is not~planar.
		\end{itemize}
		
		It remains to show how the described algorithm can be implemented to
		run in $O(n)$ time. Note that a naive implementation of our algorithm
		would take $O(n^2)$ time. Indeed, for each of the $O(n)$ inductive
		steps, the realizations $\Gamma^{i,a}_1$ and $\Gamma^{i,b}_1$,
		together with their plane embeddings $\mathcal E(\Gamma^{i,a}_1)$ and
		$\mathcal E(\Gamma^{i,b}_1)$, can be constructed in $O(1)$ time from
		$\Gamma^{i-1}_1$ and its plane embedding $\mathcal
		E(\Gamma^{i-1}_1)$. Further, the test on whether $\mathcal C_i$ is nested into
		$\mathcal C_{i-1}$ in $\Gamma^{i,b}_1$ can also be performed in $O(1)$
		time. Finally, by~\cref{thm:straight-line_realization_planarity}, it can be tested in $O(n)$ time whether $\Gamma^{i,a}_1$ and $\Gamma^{i,b}_1$ are planar straight-line realizations of $G^{i-1}_1$ respecting $\mathcal E(\Gamma^{i,a}_1)$ and
		$\mathcal E(\Gamma^{i,b}_1)$, respectively.
		
		In order to achieve $O(n)$ total running time, we do not test
		planarity at each inductive step. Rather, for $i=1,\dots,x-1$, we
		compute a ``candidate'' straight-line realization $\Gamma^i_1$ of
		$G^i_1$, and we only test for the planarity of the final realization
		$\Gamma^{x-1}_1$. ``Candidate'' here means $e_i$-optimal, however only in the case in which $\Gamma^i_1$ is planar. Indeed, $\Gamma^i_1$ might not be planar, in which case the candidate straight-line realization $\Gamma^i_1$ of $G^i_1$ is not required to satisfy any specific geometric constraint. 
		
		More formally, for $i=1,\dots,x-1$, we construct a straight-line realization $\Gamma^i_1$ of $G^i_1$, together with a plane embedding $\mathcal E(\Gamma^i_1)$, satisfying the following invariant: If an $e_i$-outer realization of $G^i_1$ exists, then $\Gamma^i_1$ is an
		$e_i$-optimal realization of $G^i_1$ and $\mathcal E(\Gamma^i_1)$ is the plane embedding of $G^i_1$ corresponding to $\Gamma^i_1$. We again remark that, if an $e_i$-outer realization of $G^i_1$ does not exist, then $\Gamma^i_1$ is non-planar or the edge $e_i$ is not incident to the outer face of $\Gamma^i_1$.
		
		In order to construct $\Gamma^i_1$, we maintain:
		\begin{itemize}
			\item the boundary $\mathcal B^i_1$ of the convex hull of $\Gamma^i_1$
			(that is, the sequence of vertices along $\mathcal B^i_1$);
			\item a Boolean value $\beta(e_i)$ such that
			$\beta(e_i)=\texttt{True}$ if and only if $e_i$ is an edge of
			$\mathcal B^i_1$; and
			\item a distinguished vertex $\xi_i$ of $\mathcal B^i_1$ (this
			information is only available if $\beta(e_i)=\texttt{False}$; the
			meaning of $\xi_i$ will be explained later).
		\end{itemize}
		Our algorithm only guarantees the correctness of the above information
		in the case in which $G^i_1$ admits an $e_i$-outer realization. For example, if $G^i_1$ does not admit an $e_i$-outer realization, $\mathcal B^i_1$ might not correctly represent the boundary of the convex hull of $\Gamma^i_1$.
		
		Before describing our algorithm, we introduce a problem and a known
		solution for it. Let $\mathcal Z$ be a planar straight-line drawing of
		a path $S=(z_1,z_2,\dots,z_m)$ and, for any $i\in \{1,\dots,m\}$, let
		$\mathcal Z_i$ be the restriction of $\mathcal Z$ to the subpath
		$Z_i=(z_1,z_2,\dots,z_i)$ of $Z$. Suppose that the locations of
		$z_1,z_2,\dots,z_m$ in $\mathcal Z$ are unveiled one at a time, in
		that order. For $i=1,\dots,m$, when the location of $z_i$ is unveiled,
		we want to compute the boundary of the convex hull of $\mathcal
		Z_i$. Melkman~\cite{m-occh-87} presented an algorithm to solve the
		problem (i.e., to compute the boundaries of the convex hulls of
		$\mathcal Z_1,\mathcal Z_2,\dots,\mathcal Z_m$) in total $O(m)$ time.
		
		This result allows us to compute the boundaries
		$\mathcal B^1_1,\mathcal B^2_1,\dots,\mathcal B^{x-1}_1$ of the convex
		hulls of $\Gamma^1_1,\Gamma^2_1,\dots,\Gamma^{x-1}_1$ in total $O(n)$
		time. Indeed, an easy inductive argument proves that $G^{x-1}_1$
		contains a path $Z$ spanning all its vertices and such that, for
		$i=1,\dots,x-1$, the vertices of $G^i_1$ induce a subpath $Z_i$ of
		$Z$. For $i=2,\dots,x-1$, our algorithm constructs $\Gamma^i_1$ from
		$\Gamma^{i-1}_1$ by assigning a position to the vertex $v_i$ of
		$\mathcal C_i$ not incident to $e_{i-1}$; once we have done that, the
		algorithm by Melkman~\cite{m-occh-87} can be employed in order to
		compute $\mathcal B^i_1$ from $\mathcal B^{i-1}_1$ (more precisely, we
		run Melkman's algorithm~\cite{m-occh-87} in parallel with our
		algorithm and feed it the locations of the vertices of $Z$ in the
		order our algorithm computes them). Note that, if $G^i_1$ admits an $e_i$-outer realization, then $\mathcal B^i_1$ is correctly computed, otherwise
		$\mathcal B^i_1$ might not correctly represent the boundary of the
		convex hull of $\Gamma^i_1$. We will only sketch what Melkman's algorithm~\cite{m-occh-87} does in order to compute $\mathcal B^1_1,\mathcal B^2_1,\dots,\mathcal B^{x-1}_1$.
		
		The base case takes $O(1)$ time, as $G^1_1$ has a unique (planar)
		straight-line realization $\Gamma^1_1=\mathcal C_1$, up to a rigid
		transformation; further, $\mathcal E(\Gamma^1_1)$ is easily computed
		from $\Gamma^1_1$, and we have $\mathcal B^1_1=\mathcal C_1$,
		$\beta(e_1)=\texttt{True}$, and $\xi_1=\texttt{Null}$. Now assume
		that, for some $i\in \{2,\dots,x-1\}$, we have a candidate straight-line
		realization $\Gamma^{i-1}_1$ of $G^{i-1}_1$, together with a plane
		embedding $\mathcal E(\Gamma^{i-1}_1)$, satisfying the invariant, the
		boundary $\mathcal B^{i-1}_1$ of the convex hull of $\Gamma^{i-1}_1$,
		the Boolean value $\beta(e_{i-1})$, and the vertex $\xi_{i-1}$. We are
		going to distinguish some cases; common to all the cases is the
		observation that, if $\Gamma^{i-1}_1$ is not planar, then it is not
		important how our algorithm computes $\Gamma^{i}_1$,
		$\mathcal E(\Gamma^{i}_1)$, $\mathcal B^{i}_1$, $\beta(e_i)$, and
		$\xi_i$, because, by the invariant, $G^{i-1}_1$ admits no
		$e_{i-1}$-outer realization and hence $G^{i}_1$ admits no
		$e_{i}$-outer realization.
		
		{\em Case 1:} $\beta(e_{i-1})=\texttt{True}$. We check whether
		$\Gamma^{i-1}_1$ and $\mathcal C_i$ can be combined in order to form a
		same-side $e_i$-outer realization of $G^i_1$ (assuming that
		$\Gamma^{i-1}_1$ is planar). Since we know $\mathcal B^{i-1}_1$, this
		check only requires $O(1)$ time. Indeed, once $\mathcal C_i$ is
		embedded in such a way that $v_{i-1}$ and $v_i$ lie on the same side
		of the line through $e_{i-1}$, it suffices to compare the slopes of
		the edges of $\mathcal C_i$ different from $e_{i-1}$ with the slopes
		of the segments of $\mathcal B^{i-1}_1$ incident to the end-vertices
		of $e_{i-1}$ and different from $e_{i-1}$;
		see~\cref{fig:outerpath-linear-contained}.
		\begin{itemize}
			\item {\em Case 1.1:} If the check is successful, we let
			$\Gamma^{i}_1$ be the resulting straight-line realization of
			$G^i_1$; further, the plane embedding $\mathcal E(\Gamma^{i}_1)$ can
			be computed in $O(1)$ time from $\mathcal E(\Gamma^{i-1}_1)$ by
			inserting the edges of $c_i$ incident to $v_i$ next to $e_{i-1}$ in the
			cyclic order of the edges incident to the end-vertices of $e_{i-1}$
			and by letting the outer face of $\mathcal E(\Gamma^{i}_1)$ be
			bounded by the cycle $c_i$. Finally, we let
			$\mathcal B^{i}_1=\mathcal C_i$, $\beta(e_i)=\texttt{True}$, and
			$\xi_i=\texttt{Null}$. Note that, if $\Gamma^{i-1}_1$ is planar,
			then $\Gamma^{i}_1$ is an $e_i$-optimal realization of $G^i_1$ (this
			is the realization which was earlier called~$\Gamma^{i,b}_1$).
			\item {\em Case 1.2:} If the check is unsuccessful, we let
			$\Gamma^{i}_1$ be the straight-line realization obtained from
			$\Gamma^{i-1}_1$ by embedding $\mathcal C_i$ in such a way that
			$v_{i-1}$ and $v_i$ lie on different sides of the line through
			$e_{i-1}$. Thus, $\Gamma^{i}_1$ can be constructed in $O(1)$ time
			from $\Gamma^{i-1}_1$. The plane embedding
			$\mathcal E(\Gamma^{i}_1)$ can also be computed in $O(1)$ time from
			$\mathcal E(\Gamma^{i-1}_1)$; indeed, the set of clockwise orders of the edges incident to each vertex in
			$\mathcal E(\Gamma^{i}_1)$ is computed as in Case~1.1, while the
			boundary of the outer face of $\mathcal E(\Gamma^{i}_1)$ is computed
			from the one of $\mathcal E(\Gamma^{i-1}_1)$ by letting the two
			edges of $c_i$ incident to $v_i$ replace the edge $e_{i-1}$. Note that, if
			$\Gamma^{i-1}_1$ is planar, then $\Gamma^{i}_1$ is an $e_i$-optimal
			realization of $G^i_1$; this is the realization which was earlier
			called~$\Gamma^{i,a}_1$.
			
			\begin{figure}[ht]
				\centering
				\subcaptionbox{\label{fig:outerpath-linear-contained}}
				{\includegraphics[scale=1]{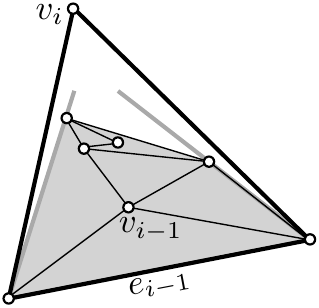}}%
				\hfil
				\subcaptionbox{\label{fig:outerpath-linear-hull}}
				{\includegraphics[scale=1]{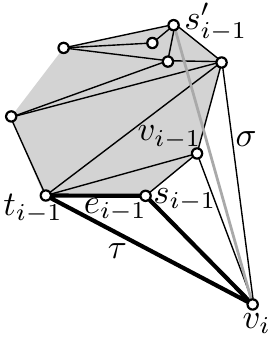}}%
				\hfil
				\subcaptionbox{\label{fig:outerpath-linear-way-out}}
				{\includegraphics[scale=1]{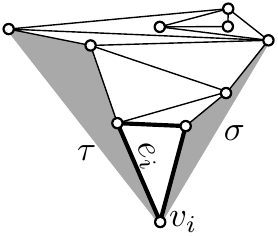}}%
				\hfil
				\subcaptionbox{\label{fig:outerpath-linear-out}}
				{\includegraphics[scale=1]{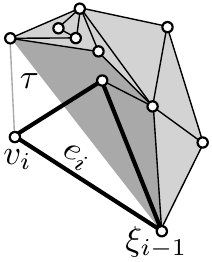}}%
				\hfil
				\subcaptionbox{\label{fig:outerpath-linear-in}}
				{\includegraphics[scale=1]{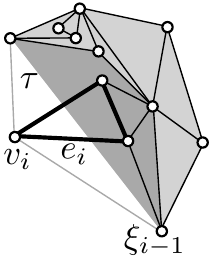}}%
				\caption{Embedding $\mathcal C_i$ (thick black) together with
					$\Gamma^{i-1}_1$. (a) In order to establish whether
					$\Gamma^{i-1}_1$ (whose convex hull is gray) fits inside
					$\mathcal C_i$, the slopes of the edges of $\mathcal C_i$
					incident to $v_i$ are compared with the slopes of the segments
					of $\mathcal B^{i-1}_1$ incident to the end-vertices of
					$e_{i-1}$ and different from $e_{i-1}$ (lines through these
					segments are dark gray). (b) If $\mathcal C_i$ is embedded so
					that $v_{i-1}$ and $v_i$ lie on different sides of the line
					through $e_{i-1}$, then $\mathcal B^{i}_1$ has to be computed;
					Melkman's algorithm~\cite{m-occh-87} examines the segments from
					$v_i$ to the vertices that follow $s_{i-1}$ in counter-clockwise
					order along $\mathcal B^{i-1}_1$, until it finds a segment (dark
					gray) from $v_i$ to a vertex $s'_{i-1}$ that keeps the
					previously encountered segment $\sigma$ to the right. (c) If
					$e_i$ lies in the interior of the convex hull of $\Gamma^i_1$,
					except possibly for its end-vertices, then the segments $\sigma$
					and $\tau$ of $\mathcal B^{i}_1$ incident to $\xi_i=v_i$
					represent the only ``way out'' of $\mathcal B^{i}_1$ for a
					planar straight-line realization of $G^{x-1}_1$; in particular,
					in any $e_{x-1}$-outer realization of $G^{x-1}_1$, the vertices
					and edges that are not in $G^i_1$ and that lie inside
					$\mathcal B^{i}_1$ can only lie in the dark gray regions
					incident to $\sigma$ and $\tau$. (d) Illustration for the case
					in which $v_i$ lies outside the convex hull of $\Gamma^{i-1}_1$
					and $\xi_{i-1}$ is the end-vertex of $e_i$ different from
					$v_i$. (e) Illustration for the case in which $v_i$ lies outside
					the convex hull of $\Gamma^{i-1}_1$ and $\xi_{i-1}$ is not the
					end-vertex of $e_i$ different from $v_i$.}
				\label{fig:outer-linear}
			\end{figure}
			
			In order to compute $\mathcal B^{i}_1$ from $\mathcal B^{i-1}_1$,
			Melkman's algorithm~\cite{m-occh-87} works as follows; refer
			to~\cref{fig:outerpath-linear-hull}. Let $s_{i-1}$ and $t_{i-1}$ be
			the end-vertices of $e_{i-1}$ and assume that $s_{i-1}$ immediately
			precedes $t_{i-1}$ in clockwise order along $\mathcal
			B^{i-1}_1$. Since $v_i$ is a vertex of $\mathcal B^{i}_1$, the
			algorithm needs to find the vertices of $\mathcal B^{i-1}_1$ that
			follow $v_i$ in clockwise and counter-clockwise order along
			$\mathcal B^{i}_1$. In order to find the vertex that follows $v_i$ in
			counter-clockwise order along $\mathcal B^{i}_1$ (the computation of
			the other vertex is done similarly), the algorithm computes the
			segments from $v_i$ to the vertices of $\mathcal B^{i-1}_1$, starting
			from $s_{i-1}$ and proceeding in counter-clockwise order along
			$\mathcal B^{i-1}_1$. This computation stops when a segment from $v_i$
			to a vertex $s'_{i-1}$ of $\mathcal B^{i-1}_1$ is considered such that
			the oriented line that passes first through $v_i$ and then through
			$s'_{i-1}$ keeps the previously encountered segment $\sigma$ to the
			right. When that happens, the end-vertex of $\sigma$ different from
			$v_i$ is the vertex that follows $v_i$ in counter-clockwise order
			along $\mathcal B^{i}_1$. This computation does not, in general, take
			$O(1)$ time, however it takes linear time in the number of vertices of
			$\mathcal B^{i-1}_1$ that are not in $\mathcal B^{i}_1$, hence it
			takes $O(n)$ time over all the convex hull computations.
			
			Let $\sigma$ and $\tau$ be the segments of $\mathcal B^{i}_1$ incident
			to $v_i$. In order to compute $\beta(e_i)$ and $\xi_i$, we check
			whether $e_i$ coincides with one of $\sigma$ and $\tau$. In the
			positive case, we let $\beta(e_i)=\texttt{True}$ and
			$\xi_i=\texttt{Null}$. In the negative case, we let
			$\beta(e_i)=\texttt{False}$ and $\xi_i=v_i$. We are now ready to
			explain the meaning of the information conveyed in $\xi_i$; refer
			to~\cref{fig:outerpath-linear-way-out}. If $e_i$ lies in the interior
			of the convex hull of $\Gamma^i_1$, except possibly for one of its
			end-vertices, then the segments $\sigma$ and $\tau$ of
			$\mathcal B^{i}_1$ incident to $\xi_i$ represent the only ``way out''
			of $\mathcal B^{i}_1$ for a planar straight-line realization of
			$G^{x-1}_1$. Formally, let $j$ be the smallest index such that $j>i$
			and such that $v_j$ lies outside $\mathcal B^{i}_1$ in a planar
			straight-line realization of $G^{j}_1$, if any such an index exists;
			then (at least) one of the two edges incident to $v_j$ in $G^{j}_1$
			crosses $\sigma$ or $\tau$. This very same property is exploited in
			Melkman's algorithm~\cite{m-occh-87}, and comes from the fact that, by
			the connectivity of $G^i_1$, any path that starts at $v_i$, that does
			not share any vertex with $G^i_1$ other than $v_i$, and that first
			intersects the interior of the convex hull of $\Gamma^i_1$ and then
			crosses $\mathcal B^{i}_1$ in a segment different from $\sigma$ and
			$\tau$ crosses $\Gamma^i_1$ as well.
		\end{itemize}  
		
		{\em Case 2:} $\beta(e_{i-1})=\texttt{False}$. In this case
		$\Gamma^{i-1}_1$ lies on both sides of the line through $e_{i-1}$,
		hence $\Gamma^{i-1}_1$ and $\mathcal C_i$ cannot be combined in order
		to form a same-side $e_i$-outer realization of $G^i_1$, even if
		$\Gamma^{i-1}_1$ is planar. Thus, as in Case~1.2, we let
		$\Gamma^{i}_1$ be the straight-line realization obtained from
		$\Gamma^{i-1}_1$ by embedding $\mathcal C_i$ in such a way that
		$v_{i-1}$ and $v_i$ lie on different sides of the line through
		$e_{i-1}$; this is again the realization which was earlier called
		$\Gamma^{i,a}_1$. The plane embedding $\mathcal E(\Gamma^{i}_1)$ is
		computed in $O(1)$ time from $\mathcal E(\Gamma^{i-1}_1)$ as in
		Case~1.2. Then $\Gamma^{i}_1$ satisfies the invariant; indeed, in the
		description of the algorithm, it was proved that, if $\Gamma^{i,a}_1$
		is planar, then it is $e_i$-optimal, otherwise $G^i_1$ admits no
		$e_i$-outer realization. Further, $\Gamma^{i}_1$ can be constructed in
		$O(1)$ time from $\Gamma^{i-1}_1$. Note that, in this case, it might
		happen that $\Gamma^{i}_1$ is not planar even if $\Gamma^{i-1}_1$ is
		planar; however, we do not check the planarity of $\Gamma^{i}_1$ at
		this time.
		
		In order to compute $\mathcal B^{i}_1$, $\beta(e_{i})$, and $\xi_{i}$,
		we exploit the information stored in $\xi_{i-1}$. As explained in
		Case~1.2, if $\Gamma^{i-1}_1$ is planar, the value $\xi_{i-1}$
		represents a vertex on the boundary $\mathcal B^{i-1}_1$ of the convex
		hull of $\Gamma^{i-1}_1$ such that $v_i$ lies outside
		$\mathcal B^{i-1}_1$ only if the (at least one) edge $e$ of $c_i$
		incident to $v_i$ and not incident to $\xi_{i-1}$ crosses $\sigma$ or
		$\tau$, where $\sigma$ and $\tau$ are the segments incident to
		$\xi_{i-1}$ in $\mathcal B^{i-1}_1$. Hence, we (as in Melkman's
		algorithm~\cite{m-occh-87}) check in $O(1)$ time whether $e$ crosses
		$\sigma$ or $\tau$.
		\begin{itemize}
			\item In the positive case, $v_i$ lies outside the convex hull of
			$\Gamma^{i-1}_1$; then $\mathcal B^{i}_1$ is constructed as in
			Case~1.2, where the segment between $\sigma$ and $\tau$ that crosses
			$e$ plays the role of $e_{i-1}$. If $\xi_{i-1}$ is the end-vertex of
			$e_i$ different from $v_i$, as in~\cref{fig:outerpath-linear-out},
			then we set $\beta(e_{i})=\texttt{True}$ and $\xi_i=\texttt{Null}$
			(in this case $G^i_1$ ``made it out'' of $\mathcal B^{i-1}_1$),
			otherwise, as in~\cref{fig:outerpath-linear-in}, we set
			$\beta(e_{i})=\texttt{False}$ and $\xi_i=v_i$.
			\item In the negative case, $v_i$ lies inside the convex hull of
			$\Gamma^{i-1}_1$ or $\Gamma^i_1$ is not planar (or both). Hence, we
			set $\mathcal B^{i}_1=\mathcal B^{i-1}_1$,
			$\beta(e_{i})=\texttt{False}$, and $\xi_{i}=\xi_{i-1}$.
		\end{itemize}
		
		This concludes the description of the inductive case of our
		algorithm. Each inductive step takes $O(1)$ time, except for the
		computation of the boundaries
		$\mathcal B^{1}_1,\mathcal B^{2}_1,\dots,\mathcal B^{x-1}_1$, which
		however takes total $O(n)$ time as explained above. Hence, the overall
		running time of the algorithm is in $O(n)$.
	\end{proofA}
	
	\begin{figure}[ht]
		\centering
		\includegraphics[scale=1]{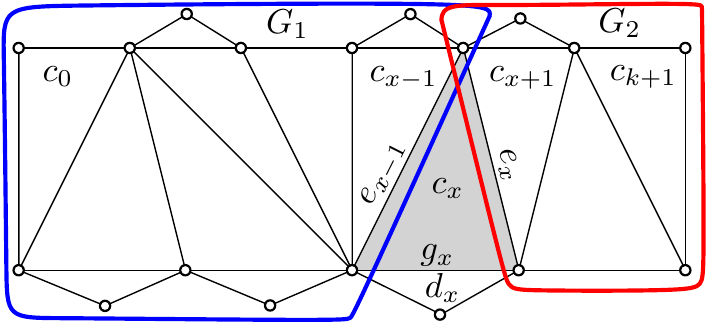}%
		\caption{The outerplane embedding of an outerpillar. The gray face is the one bounded by $c_x$. The graphs $G_1$ and $G_2$ are enclosed inside fat curves.}
		\label{fig:outerpillar-structure}
	\end{figure}
	
	\medskip
	\begin{proofB}
		The proof uses several ideas from the proof
		of~\cref{th:outerpaths}. We hence omit details of arguments that
		have been already detailed in that proof.
		
		Let $G=(V,E,\lambda)$ be an $n$-vertex weighted maximal outerpillar
		(as the one in~\cref{fig:outerpillar-structure}) and let
		$\mathcal O$ be its outerplane embedding, which can be found in
		$O(n)$ time~\cite{d-iroga-07,m-laarogmog-79,w-rolt-87}. Since $G$ is
		an outerpillar, its dual tree is a caterpillar $T$, whose spine we
		denote by $P:=(p_1,\dots,p_k)$; note that $T$ can be easily
		recovered from $\mathcal O$ in $O(n)$ time. We arbitrarily select
		two leaves of $T$, respectively adjacent to $p_1$ and $p_k$, and
		denote them by $p_0$ and $p_{k+1}$; we further denote by
		$\overline P$ the path $(p_0,p_1,\dots,p_k,p_{k+1})$. Each node
		$p_i$ of $\overline P$ corresponds to an internal face of
		$\mathcal O$; we denote by $c_i$ the $3$-cycle of $G$ bounding such
		a face. Further, each node $p_i$ in $\overline P$ is adjacent to at
		most one leaf not in $\overline P$ (the satisfaction of this
		property was the reason to augment $P$ to $\overline{P}$), which
		corresponds to an internal face of $\mathcal O$; we denote by $d_i$
		the $3$-cycle of $G$ bounding such a face. Each cycle $c_i$ has a unique (planar)
		straight-line realization (up to a rigid transformation), which we
		denote by $\mathcal C_i$, and each cycle $d_i$ has a unique (planar)
		straight-line realization (up to a rigid transformation), which we
		denote by $\mathcal D_i$. For $i=0,\dots,k$, let $e_i$ be the edge
		of $G$ that is dual to the edge $(p_i,p_{i+1})$ of $\overline P$;
		note that $e_i$ is an internal edge of $\mathcal O$. Further, for
		every index $i\in\{1,\dots,k\}$ such that $d_i$ exists, let $g_i$ be
		the edge shared by $c_i$ and $d_i$.
		
		Let $c^*$ be a $3$-cycle of $G$ whose edges have a sum of the
		lengths which is maximum among all the $3$-cycles of $G$. Let
		$x\in\{0,1,\dots,k+1\}$ be the index such that we have either
		$c^*=c_x$ or $c^*=d_x$. Let $G_1$ be the subgraph of $G$ composed of
		the $3$-cycles $c_i$ and of the $3$-cycles $d_i$ with
		$i=0,\dots,x-1$ (some of the $3$-cycles $d_i$ might not
		exist). Analogously, let $G_2$ be the subgraph of $G$ composed of
		the $3$-cycles $c_i$ and of the $3$-cycles $d_i$ with
		$i=x+1,\dots,k+1$ (again, some of the $3$-cycles $d_i$ might not
		exist). One of these graphs might be undefined if $x=0$ or $x=k+1$;
		however, in the following, we assume that $1\leq x\leq k$ as the
		arguments for the cases $x=0$ and $x=k+1$ are analogous and actually
		simpler. Note that $G_1$ and $G_2$ are maximal outerpillars. Also,
		we assume that the $3$-cycle $d_x$ exists, as the argument for the
		case in which it does not is simpler.
		
		Assume that a planar straight-line realization $\Gamma$ of $G$
		exists; for $i=1,2$, let $\Gamma_i$ be the restriction of $\Gamma$
		to $G_i$. If $c^*=c_x$, then $c_x$ does not lie inside any $3$-cycle
		of $G$ in $\Gamma$. If $c^*=d_x$, then $c_x$ might lie inside $d_x$
		in $\Gamma$, however it does not lie inside any other $3$-cycle of
		$G$ in $\Gamma$, as otherwise, by planarity, such a $3$-cycle would
		also contain $d_x$, which is not possible. Therefore, both if
		$c^*=c_x$ and if $c^*=d_x$, we have that $\Gamma_1$ is an
		$e_{x-1}$-outer realization of $G_1$ and that $\Gamma_2$ is an
		$e_x$-outer realization of $G_2$.
		
		We show an $O(n^3)$-time algorithm that either concludes that $G_1$
		admits no $e_{x-1}$-outer realization or constructs an
		$e_{x-1}$-optimal set $\mathcal R_1$ of $e_{x-1}$-outer realizations
		of $G_1$, where for each realization $\Gamma_1$ in $\mathcal R_1$
		the plane embedding $\mathcal E(\Gamma_1)$ is constructed as
		well. The same algorithm is then also used in order to either
		conclude that $G_2$ admits no $e_{x}$-outer realization or to
		construct an $e_{x}$-optimal set $\mathcal R_2$ of $e_{x}$-outer
		realizations of $G_2$, where for each realization $\Gamma_2$ in
		$\mathcal R_2$ the plane embedding $\mathcal E(\Gamma_2)$ is
		constructed as well.
		
		For $i=0,1,\dots,x-1$, let $G^i_1$ be the subgraph of $G_1$ composed
		of the $3$-cycles $c_0,c_1,\dots,c_i$ and of the $3$-cycles
		$d_0,d_1,\dots,d_i$ (some of these $3$-cycles might not exist); note
		that $G^{x-1}_1=G_1$. As in the proof of~\cref{th:outerpaths}, we
		use the observation that the restriction of any $e_{x-1}$-outer
		realization of $G_1$ to $G^i_1$ is an $e_i$-outer realization of
		$G^i_1$.
		
		Our algorithm works by induction on $i$ in order to either conclude
		that $G^i_1$ admits no $e_i$-outer realization (and thus $G_1$
		admits no $e_{x-1}$-outer realization), or to determine an
		$e_i$-optimal set $\mathcal R^i_1$ of $e_i$-outer realizations of
		$G^i_1$ with $|\mathcal R^i_1|\leq i+1$, where for each realization
		$\Gamma^i_1$ in $\mathcal R^i_1$ the plane 
		$\mathcal E(\Gamma^i_1)$ is constructed as well. This is trivial
		when $i=0$, namely $G^0_1$ is the $3$-cycle $c_0$ which has a unique
		planar straight-line realization $\Gamma^0_1:=\mathcal C_0$, up to a
		rigid transformation; then $\mathcal R^0_1:=\{\Gamma^0_1\}$ is
		indeed an $e_0$-optimal set of $e_0$-outer realizations of $G^0_1$
		with $|\mathcal R^0_1|\leq 1$ and $\mathcal E(\Gamma^0_1)$ is easily
		computed from $\Gamma^0_1$.
		
		Assume that we have an $e_{i-1}$-optimal set $\mathcal R^{i-1}_1$ of
		$e_{i-1}$-outer realizations of $G^{i-1}_1$ with
		$|\mathcal R^{i-1}_1|\leq i$, for some $i\in \{1,\dots,x-1\}$, and
		we have, for each realization $\Gamma^{i-1}_1$ in
		$\mathcal R^{i-1}_1$, the plane embedding
		$\mathcal E(\Gamma^{i-1}_1)$; indeed, if we
		already concluded that $G_1$ admits no $e_{x-1}$-outer realization,
		there is nothing else to do. We present an algorithm that either
		concludes that $G_1$ admits no $e_{x-1}$-outer realization or
		constructs the set $\mathcal R^i_1$ and, for each realization
		$\Gamma^{i}_1$ in $\mathcal R^i_1$, the plane embedding
		$\mathcal E(\Gamma^i_1)$. The algorithm consists of
		three steps.
		
		\begin{itemize}
			\item {\bf Step 1.} In the first step, we aim for the construction of an
			$e_{i}$-outer realization of $G^{i}_1$ in which both $G^{i-1}_1$
			and $d_i$ lie inside $c_i$ (except for the edges $e_{i-1}$ and
			$g_i$, which $c_i$ shares with $G^{i-1}_1$ and $d_i$,
			respectively), as the one
			in~\cref{fig:outerpillar-inside-inside}. This is done by checking
			whether any of the $O(n)$ realizations of $G^{i-1}_1$ in
			$\mathcal R^{i-1}_1$ fits inside $\mathcal C_i$ together with
			$\mathcal D_i$. More precisely, we act as follows.
			
			\begin{enumerate} [(A)]
				\item We test in $O(1)$ time whether $\mathcal C_{i-1}$
				lies inside $\mathcal C_i$ once they coincide on the
				representation of $e_{i-1}$ and once their vertices not incident
				to $e_{i-1}$ lie on the same side of the line through
				$e_{i-1}$. If the test is positive, we proceed with the first step of the
				algorithm, otherwise we conclude that no realization of $G^{i-1}_1$ fits
				inside $\mathcal C_i$, we terminate the first step of the
				algorithm and proceed to the second one.
				\item We test in $O(1)$ time whether $\mathcal D_i$ lies
				inside $\mathcal C_i$ once they coincide on the representation
				of $g_i$ and once their vertices not incident to $g_i$ lie on
				the same side of the line through $g_i$; if $d_i$ does not
				exist, this test is skipped. If the test is positive, we proceed
				with the first step of the
				algorithm, otherwise we conclude that $\mathcal D_i$ does
				not fit inside $\mathcal C_i$, we terminate the first step of
				the algorithm and proceed to the second one.
				\item For each realization $\Gamma^{i-1}_1$ of $G^{i-1}_1$ in $\mathcal R^{i-1}_1$, we construct a straight-line realization $\Gamma^i_1$ of $G^i_1$ in which $G^{i-1}_1$ is represented by $\Gamma^{i-1}_1$ and $\mathcal C_{i-1}$ and $\mathcal D_i$ are
				embedded inside $\mathcal C_i$. Such a realization can be
				obtained in $O(1)$ time by augmenting $\Gamma^{i-1}_1$ with a
				representation of the edges of $c_i$ and $d_i$ not in
				$G^{i-1}_1$. Further, the plane embedding
				$\mathcal E(\Gamma^i_1)$ can be constructed in
				$O(1)$ time by inserting the edges incident to $v_i$ next to
				$e_{i-1}$ in the cyclic order of the edges incident to the
				end-vertices of $e_{i-1}$, by then inserting the edges of $d_i$
				different from $g_i$ next to $g_i$ in the cyclic order of the
				edges incident to the end-vertices of $g_i$, and by letting the
				outer face of $\mathcal E(\Gamma^i_1)$ be bounded by the cycle
				$c_i$. We then test in $O(n)$
				time, by means
				of~\cref{thm:straight-line_realization_planarity}, whether $\Gamma^i_1$ is a planar straight-line realization of $G^i_1$ respecting $\mathcal E(\Gamma^i_1)$.     \end{enumerate}

			\begin{figure}[tb!]
				\centering
				\subcaptionbox{\label{fig:outerpillar-inside-inside}}
				{\includegraphics[scale=1]{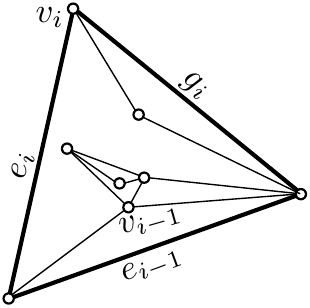}}%
				\hfil
				\subcaptionbox{\label{fig:outerpillar-inside-outside}}
				{\includegraphics[scale=1]{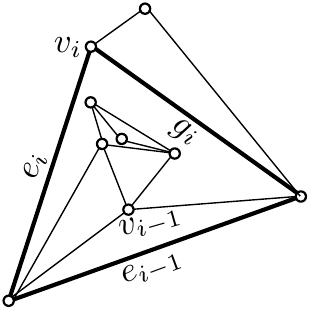}}%
				\hfil
				\subcaptionbox{\label{fig:outerpillar-outside}}
				{\includegraphics[scale=1]{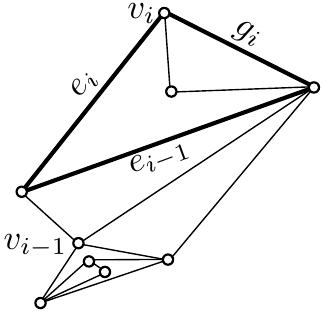}}%
				\caption{Three $e_{i}$-outer realizations of $G^{i}_1$. The
					cycle $c_i$ is represented by a thick line. (a) $G^{i-1}_1$
					and $d_i$ lie inside $c_i$; (b) $G^{i-1}_1$ lies inside $c_i$
					and $d_i$ lies outside $c_i$; (c) $G^{i-1}_1$ lies outside
					$c_i$.}
				\label{fig:outerpillar-types}
			\end{figure}
			
			The total running time of the first step of the algorithm is hence in $O(n^2)$. If
			(at least) one of these tests succeeds in constructing a
			realization $\Gamma^i_1$, then we complete the induction with
			$\mathcal R^i_1=\{\Gamma^i_1\}$. If none of the tests succeeds, we
			proceed to the second step of the algorithm.
			\item {\bf Step 2.} In the second step, we aim for the construction of an
			$e_{i}$-outer realization of $G^{i}_1$ in which $G^{i-1}_1$ lies
			inside $c_i$ (except for the edge $e_{i-1}$, which is shared by
			$G^{i-1}_1$ and $c_i$) and $d_i$ lies outside $c_i$ (except for
			the edge $g_i$, which is shared by $c_i$ and $d_i$), as
			in~\cref{fig:outerpillar-inside-outside}. If $d_i$ does not exist,
			then this step is actually ignored, given that the desired
			representation of $G^i_1$ would have been already constructed in
			the first step. As in the first step, we first check whether
			$\mathcal C_{i-1}$ lies inside $\mathcal C_i$ once they coincide
			on the representation of $e_{i-1}$ and once their vertices not
			incident to $e_{i-1}$ lie on the same side of the line through
			$e_{i-1}$. If the test is negative, we conclude that no
			realization of $G^{i-1}_1$ fits inside $\mathcal C_i$, we
			terminate the second step of the algorithm and proceed to the
			third one with $\mathcal R^i_1=\emptyset$. If the test is
			positive, we independently consider each of the $O(n)$
			realizations of $G^{i-1}_1$ in $\mathcal R^{i-1}_1$, say
			$\Gamma^{i-1}_1$. We augment $\Gamma^{i-1}_1$ to a realization
			$\Gamma^i_1$ of $G^{i}_1$ (and compute the plane embedding
			$\mathcal E(\Gamma^{i}_1)$ from the plane embedding
			$\mathcal E(\Gamma^{i-1}_1)$) in $O(1)$ time similarly as in the
			first step, except that $\mathcal D_i$ is now embedded outside
			$\mathcal C_i$ and hence the outer face of
			$\mathcal E(\Gamma^{i}_1)$ is now delimited by the edges of $c_i$
			and $d_i$ different from $g_i$. By means
			of~\cref{thm:straight-line_realization_planarity}, we test in $O(n)$ time whether
			$\Gamma^i_1$ is a planar straight-line realization of $G^i_1$ respecting $\mathcal E(\Gamma^{i}_1)$. Hence, the total running time of the
			second step is in $O(n^2)$. If (at least) one of these tests shows
			that a realization $\Gamma^i_1$ of $G^i_1$ is planar, we conclude
			this step by initializing $\mathcal R^i_1=\{\Gamma^i_1\}$. If none
			of the tests succeeds, then we initialize
			$\mathcal R^i_1=\emptyset$.  In both cases, we proceed to the
			third step of the algorithm.
			\item {\bf Step 3.} In the third step, we aim for the construction of (possibly
			multiple) $e_{i}$-outer realizations of $G^{i}_1$ in which
			$G^{i-1}_1$ lies outside $c_i$ (except for the edge $e_{i-1}$,
			which is shared by $G^{i-1}_1$ and $c_i$), as
			in~\cref{fig:outerpillar-outside}. This is done as
			follows. Preliminarly, we construct a realization of $c_i\cup d_i$
			(if $d_i$ exists). In order to do that, we check in $O(1)$ time
			whether $\mathcal D_i$ fits inside $\mathcal C_i$; in the positive
			case, we place it inside $\mathcal C_i$, otherwise we place it
			outside $\mathcal C_i$. We now independently consider each
			$e_{i-1}$-outer realization $\Gamma^{i-1}_1$ of $G^{i-1}_1$ in
			$\mathcal R^{i-1}_1$ and we put it together with the constructed
			representation of $c_i\cup d_i$ (if $d_i$ does not exist, we just
			put $\Gamma^{i-1}_1$ together with $\mathcal C_i$), so that the
			vertices of $c_{i-1}$ and $c_i$ not incident to $e_{i-1}$ lie on
			different sides of the line through $e_{i-1}$. This can be done in
			$O(1)$ time, as in the first and second steps, by augmenting
			$\Gamma^{i-1}_1$ with a representation of the edges of $c_i$ and
			$d_i$ not in $G^{i-1}_1$. The plane embedding of the constructed
			representation of $G^i_1$ can also be constructed in $O(1)$ time
			from the plane embedding $\mathcal E(\Gamma^{i-1}_1)$, similarly as in the first and second steps. We
			now check in $O(n)$ time whether the constructed representation of $G^i_1$ is  a planar straight-line realization of $G^i_1$ respecting the constructed plane embedding, by means
			of~\cref{thm:straight-line_realization_planarity}; in the positive case, we add it to $\mathcal R^i_1$. Hence, the total running time of
			the third step is in $O(n^2)$. After the third step, if
			$\mathcal R^i_1=\emptyset$, we conclude that $G_1$ admits no
			$e_{x-1}$-outer realization.
		\end{itemize}
		
		This concludes the description of the algorithm. We now prove its
		correctness.
		
		First, we prove that $|\mathcal R^i_1|\leq i+1$. This is obvious if
		the first step succeeded in the construction of a realization of
		$G^i_1$, as in that case we have $|\mathcal
		R^i_1|=1<i+1$. Otherwise, at most one realization of $G^i_1$ was
		added to $\mathcal R^i_1$ in the second step; further, during the
		third step, at most one realization of $G^i_1$ was added to
		$\mathcal R^i_1$ for each realization in $\mathcal R^{i-1}_1$. Since
		$|\mathcal R^{i-1}_1|\leq i$, by the inductive hypothesis, it
		follows that $|\mathcal R^i_1|\leq i+1$.
		
		Second, we have that each realization in $\mathcal R^i_1$ is planar,
		as ensured by a test performed on it before the insertion into
		$\mathcal R^i_1$. Further, each realization in $\mathcal R^i_1$ is
		$e_i$-outer. Namely, if the vertices of $c_{i-1}$ and $c_i$ that are
		not incident to $e_{i-1}$ lie on the same side of the line through
		$e_{i-1}$, we test whether $\mathcal C_{i-1}$ lies inside
		$\mathcal C_i$; further, if the vertices of $d_i$ and $c_i$ that are
		not incident to $g_i$ lie on the same side of the line through
		$g_i$, we test whether $\mathcal D_i$ lies inside $\mathcal C_i$. A
		realization is added to $\mathcal R^i_1$ only if both tests succeed,
		which ensures that $\mathcal C_i$ does not lie inside
		$\mathcal C_{i-1}$ or $\mathcal D_i$ in any realization in
		$\mathcal R^i_1$. A containment of $\mathcal C_i$ inside any other
		$3$-cycle of $G^i_1$ would contradict the planarity of the
		corresponding realization.
		
		Finally, we prove that $\mathcal R^i_1$ is $e_i$-optimal. This can
		be done very similarly as in the proof
		of~\cref{th:outerpaths}. Namely, consider any $e_i$-outer
		realization $\Psi^i_1$ of $G^i_1$ and let $\Psi^{i-1}_1$ be the
		restriction of $\Psi^i_1$ to $G^{i-1}_1$. Since $\mathcal R^{i-1}_1$
		is $e_{i-1}$-optimal, there exists an $e_{i-1}$-outer realization
		$\Lambda^{i-1}_1$ of $G^{i-1}_1$ in $\mathcal R^{i-1}_1$ such that
		$\Lambda^{i-1}_1$ is $e_{i-1}$-smaller than $\Psi^{i-1}_1$. Hence,
		we can apply \cref{le:optimal-replacement} with $H=G^{i-1}_1$,
		$K=c_i\cup d_i$ (or $K=c_i$ if $d_i$ does not exist),
		$\mathcal G=\Psi^i_1$, $e=e_{i-1}$, and
		$\mathcal H'= \Lambda^{i-1}_1$, in order to obtain an $e_i$-outer
		realization $\Lambda^i_1$ of $G^i_1$ whose restriction to
		$G^{i-1}_1$ can be obtained by a rigid transformation of
		$\Lambda^{i-1}_1$ and such that $\mathcal C_{i-1}$ lies inside
		$\mathcal C_{i}$ in $\Lambda^{i}_1$ if and only if
		$\mathcal C_{i-1}$ lies inside $\mathcal C_{i}$ in $\Psi^i_1$. It
		only remains to prove that $\Lambda^i_1$, or an $e_i$-outer
		realization that is $e_i$-smaller than $\Lambda^i_1$, actually
		belongs to $\mathcal R^i_1$.
		
		First, if $\mathcal R^i_1$ contains a single realization
		$\Gamma^i_1$ added during the first step, then the outer face of
		$\Gamma^i_1$ is bounded by $\mathcal C_i$; since the set
		$\mathcal I_{\Lambda^{i}_1}$ contains at least the closure of the
		interior of $\mathcal C_i$, we have that $\Gamma^i_1$ is
		$e_i$-smaller than $\Lambda^{i}_1$. We can hence assume that
		$\Lambda^{i-1}_1$ lies outside $\mathcal C_i$ or that $\mathcal D_i$
		lies outside $\mathcal C_i$ (or both). Indeed, if both
		$\Lambda^{i-1}_1$ and $\mathcal D_i$ lied inside $\mathcal C_i$,
		then the first step of the algorithm would have constructed
		$\Lambda^i_1$, and hence $\mathcal R^i_1$ would contain a single
		realization $\Gamma^i_1$ added during the first step (note that
		$\Gamma^i_1$ would not necessarily be the same as $\Lambda^i_1$,
		however the outer faces of both $\Gamma^i_1$ and $\Lambda^i_1$ would
		be delimited by $\mathcal C_i$).
		
		Second, consider the case in which $\Lambda^{i-1}_1$ lies inside
		$\mathcal C_i$ and $\mathcal D_i$ lies outside $\mathcal C_i$ in
		$\Lambda^i_1$. This implies that the second step of the algorithm
		constructs $\Lambda^i_1$, and hence $\mathcal R^i_1$ contains a
		realization $\Gamma^i_1$ added during the second step of the
		algorithm. The outer face of $\Gamma^i_1$ is bounded by the edges of
		$\mathcal C_i$ and $\mathcal D_i$ different from $g_i$, same as
		$\Lambda^i_1$, hence $\Gamma^i_1$ is $e_i$-smaller
		than~$\Lambda^{i}_1$.
		
		Third, consider the case in which $\Lambda^{i-1}_1$ lies outside
		$\mathcal C_i$. If $\mathcal D_i$ lies inside $\mathcal C_i$ in
		$\Lambda^{i}_1$, or if $\mathcal D_i$ lies outside $\mathcal C_i$ in
		$\Lambda^{i}_1$ and does not fit inside $\mathcal C_i$, then the
		third step of the algorithm constructs $\Lambda^{i}_1$ and adds it
		to $\mathcal R^i_1$. If $\mathcal D_i$ lies outside $\mathcal C_i$
		in $\Lambda^{i}_1$ and fits inside $\mathcal C_i$, then the third
		step of the algorithm constructs a realization $\Gamma^i_1$ which
		coincides with $\Lambda^{i}_1$ when restricted to
		$G^{i-1}_1\cup c_i$, but in which $\mathcal D_i$ lies inside
		$\mathcal C_i$. Then we have that $\Gamma^i_1$ is $e_i$-smaller than
		$\Lambda^{i}_1$. This concludes the proof of the $e_i$-optimality of
		$\mathcal R^i_1$ and of the correctness of the inductive step.
		
		The inductive step is hence performed in $O(n^2)$ time. Since there
		are $O(n)$ inductive steps, our algorithm takes $O(n^3)$ time in
		order to either conclude that $G_1$ admits no $e_{x-1}$-outer
		realization or to construct an $e_{x-1}$-optimal set $\mathcal R_1$
		with $|\mathcal R_1|\in O(n)$ of $e_{x-1}$-outer realizations of
		$G_1$ together with the plane embedding $\mathcal E(\Gamma_1)$ of
		each realization $\Gamma_1$ in $\mathcal R_1$; the algorithm also
		takes $O(n^3)$ time in order to either conclude that $G_2$ admits no
		$e_{x}$-outer realization or to construct an $e_{x}$-optimal set
		$\mathcal R_2$ with $|\mathcal R_2|\in O(n)$ of $e_{x}$-outer
		realizations of $G_2$ together with the plane embedding
		$\mathcal E(\Gamma_2)$ of each realization $\Gamma_2$ in
		$\mathcal R_2$.
		
		As argued above, if $G_1$ admits no $e_{x-1}$-outer realization or
		$G_2$ admits no $e_{x}$-outer realization, then $G$ admits no planar
		straight-line realization. Assume hence that $\mathcal R_1$ and
		$\mathcal R_2$ are both non-empty. We independently consider each
		pair $(\Gamma_1,\Gamma_2)$ such that $\Gamma_1 \in \mathcal R_1$ and
		$\Gamma_2 \in \mathcal R_2$; since $|\mathcal R_1|$ and
		$|\mathcal R_2|$ are both in $O(n)$, there are $O(n^2)$ such
		pairs. We construct eight straight-line realizations of $G$, each
		with the corresponding plane embedding, according to the following
		choices:
		\begin{itemize}
			\item the vertices of $c_{x-1}$ and $c_{x}$ not incident to
			$e_{x-1}$ might lie on the same side or on different sides of the
			line through $e_{x-1}$;
			\item the vertices of $c_{x+1}$ and $c_{x}$ not incident to $e_{x}$
			might lie on the same side or on different sides of the line
			through $e_{x}$; and
			\item the vertices of $d_x$ and $c_{x}$ not incident to $g_x$ might
			lie on the same side or on different sides of the line through
			$g_x$.
		\end{itemize}
		We test in   $O(n)$ time whether each of these straight-line realizations is planar and respects the corresponding plane embedding, by means
		of~\cref{thm:straight-line_realization_planarity}. Hence, all these
		tests can be performed in overall $O(n^3)$ time.  If (at least) one
		of these realizations is planar, then it is the desired planar
		straight-line realization of $G$. Otherwise, there exists no planar
		straight-line realization of $G$; this can be proved as in the proof
		of~\cref{th:outerpaths}. Namely, if a planar straight-line
		realization $\Gamma$ of $G$ exists, then the restriction of $\Gamma$
		to $G_1$ and the restriction of $\Gamma$ to $G_2$ can be replaced by
		straight-line realizations obtained by rigid transformations of
		suitable realizations in $\mathcal R_1$ and $\mathcal R_2$,
		respectively. By~\cref{le:optimal-replacement}, this results in a
		planar straight-line realization of $G$ which, by construction, is
		one of the straight-line realizations for which the planarity
		testing was negative. This contradiction concludes the proof of the
		theorem.
	\end{proofB}
	
	
	
	\section{2-Trees with Short Longest Path} \label{sec:longest_path}
	
	In this section, we show that the \FEPRshort problem is polynomial-time solvable for weighted $2$-trees whose longest path has bounded length, where the \emph{length} of a path is the number of edges in it. Our algorithm actually runs in polynomial time in a more general case: if the height of the SPQ-tree of the $2$-tree is bounded. We start this section by introducing the necessary definitions and notation.
	
	The notion of \emph{SPQ-tree} is a specialization of the well-known notion of \emph{SPQR-tree}~\cite{dt-opl-96,gm-lti-00} for \emph{series-parallel graphs}~\cite{dett-gd-19,vtl-rspd-82}. A series-parallel graph is any graph that can be obtained by means of the following recursive construction (see~\cref{fig:spq-tree}). 
	
	\begin{figure}[htb]
		\centering
		{\includegraphics[width=.7\textwidth]{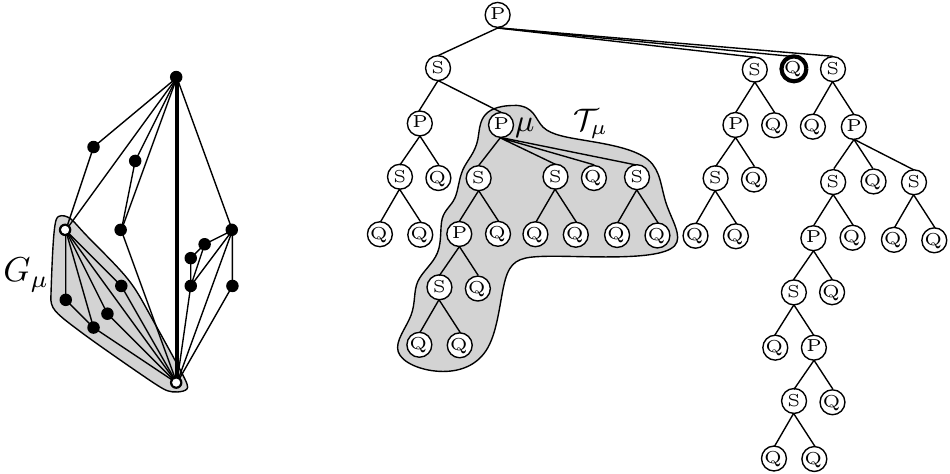}}%
		\caption{A $2$-tree $G$ and its SPQ-tree $\mathcal T$. The gray regions show $G_{\mu}$ and $\mathcal T_{\mu}$, for a node $\mu$ of $\mathcal T$. The poles of $\mu$ are represented by empty disks in $G$. The fat edge corresponds to the Q-node that is a child of the root of $\mathcal T$. The height of $\mathcal T$ is $8$, the one of $\mathcal T_{\mu}$ is $4$.}
		\label{fig:spq-tree}
	\end{figure}
	
	\begin{itemize}
		\item An edge $(u,v)$ is a series-parallel graph with two poles $u$ and $v$; its SPQ-tree $\mathcal T$ is rooted tree with a single \emph{Q-node}. We say that $(u,v)$ \emph{corresponds} to the Q-node (and vice versa).
		\item Let $G_1,\dots,G_k$ be series-parallel graphs, where $G_i$ has poles $u_i$ and $v_i$, for $i=1,\dots,k$; also, let $\mathcal T_1,\dots, \mathcal T_k$ be their SPQ-trees. The graph obtained by identifying the vertices $u_1,\dots,u_k$ into a single vertex $u$ and the vertices $v_1,\dots,v_k$ into a single vertex $v$ is a series-parallel graph $G$ with poles $u$ and $v$. The SPQ-tree $\mathcal T$ of $G$ is the tree whose root is a \emph{P-node} which is the parent of the roots of $\mathcal T_1,\dots, \mathcal T_k$. We say that $G$ \emph{corresponds} to $\mathcal T$.
		\item Let $G_1,\dots,G_k$ be series-parallel graphs, where $G_i$ has poles $u_i$ and $v_i$, for $i=1,\dots,k$; also, let $\mathcal T_1,\dots, \mathcal T_k$ be their SPQ-trees. The graph obtained by identifying the vertex $v_i$ with the vertex $u_{i+1}$, for $i=1,\dots,k-1$, is a series-parallel graph $G$ with poles $u_1$ and $v_k$. The SPQ-tree $\mathcal T$ of $G$ is the tree whose root is an \emph{S-node} which is the parent of the roots of $\mathcal T_1,\dots, \mathcal T_k$. We say that $G$ \emph{corresponds} to $\mathcal T$.
	\end{itemize}
	
	Every $2$-tree $G$ is a series-parallel graph (see, e.g.,~\cite[Chapter 5]{bbs-gcs-99}) such that:
	\begin{enumerate}[(C1)]
		\item The root of the SPQ-tree $\mathcal T$ of $G$ is a P-node.
		\item Each P-node has at least two children; one of them is a Q-node and the others are S-nodes. 
		\item Each S-node has exactly two children, each of which is either a Q-node or a P-node.
	\end{enumerate}

	Moreover, if $e$ is a prescribed edge of $G$, then the SPQ-tree of $G$ can be constructed in such a way that the Q-node corresponding to $e$ is a child of the root $\rho$ of the SPQ-tree (that is, $e$ is the edge between the poles of $\rho$). We will use this property later.
	
	Let $G$ be a $2$-tree, $\mathcal T$ be the SPQ-tree of $G$, and $\mu$ be a node of $\mathcal T$. We denote by $u_{\mu}$ and $v_{\mu}$ the poles of $\mu$, by $\mathcal T_{\mu}$ the subtree of $\mathcal T$ rooted at $\mu$, and by $G_\mu$ the subgraph of $G$ corresponding to $\mathcal T_{\mu}$. The \emph{height} of $\mathcal T$ is the maximum number of edges in any root-to-leaf path in $\mathcal T$. By Conditions (C1)--(C2), the height of $\mathcal T$ is even and positive. 
	
	
	%
	
	We can upper bound the height of an SPQ-tree of a $2$-tree $G$ in terms of the length of the longest path of $G$, as proved in the following lemma. 
	
	\begin{lemma} \label{le:height-diameter}
		Let $G$ be a $2$-tree, let $\ell$ be the length of the longest path of $G$, let $\mathcal T$ be the SPQ-tree of $G$, and let $h$ be the height of $\mathcal T$. Then we have $h\leq 2 \ell-2$.
	\end{lemma}
	
	\begin{proof}
		We prove the following claim. For every P-node $\mu$ in $\mathcal T$, there exists a path $P_{\mu}$ in $G_{\mu}$ that connects the poles $u$ and $v$ of $\mu$ and whose length is $\frac{h_{\mu}}{2}+1$, where $h_{\mu}$ is the height of $\mathcal T_{\mu}$. The claim implies the lemma by choosing $\mu$ to be the root of $\mathcal T$; this can be done since, by Condition~(C1), the root of $\mathcal T$ is a P-node.
		
		We prove the claim by induction on $h_{\mu}$. Recall that  $h_{\mu}$ is even and positive. Let $\nu$ be a child of $\mu$ such that the height of $\mathcal T_{\nu}$ is $h_{\mu}-1$; by Condition~(C2) and since $h_{\mu}\geq 2$, we have that $\nu$ is an S-node with poles $u$ and $v$. By Condition~(C3), we have that $\nu$ has two children $\nu_1$ and $\nu_2$, each of which is either a P-node or a Q-node.  Assume, w.l.o.g.\ up to a relabeling, that $u$ is a pole of $\nu_1$, that $v$ is a pole of $\nu_2$, and that $w$ is the pole shared by $\nu_1$ and $\nu_2$. 
		
		In the base case, we have $h_{\mu}=2$. Then $\nu_1$ and $\nu_2$ are both Q-nodes, as otherwise the height of $\mathcal T$ would be greater than $2$. It follows that $G$ contains the path $P_{\mu}=(u,w,v)$. 
		
		In the inductive case, we have $h_{\mu}>2$. Assume that the height of $\mathcal T_{\nu_1}$ is $h_{\mu}-2$, as the case in which the height of $\mathcal T_{\nu_2}$ is $h_{\mu}-2$ is analogous. Since $h_{\mu}>2$, we have that $\nu_1$ is a P-node. By induction, $G_{\nu_1}$ contains a path $\mathcal P_{\nu_1}$ with length $\frac{h_{\mu}-2}{2}+1$ between $u$ and $w$. Further, $G_{\nu_2}$ contains the edge $(w,v)$; this is obvious if $\nu_2$ is a Q-node and comes from Condition~(C2) otherwise. The path $\mathcal P_{\nu_1}\cup (w,v)$ is the desired path $\mathcal P_{\mu}$ with length $\frac{h_{\mu}}{2}+1$ between $u$ and $v$. This completes the induction and hence the proof of the lemma.
	\end{proof}
	
	\newcommand{\algoritmAny}{n^{O(2^{h^*})}}
	\newcommand{\algoritmPath}{n^{O(4^{\ell})}}
	
	The main result of this section is the following.
	
	\begin{theorem}\label{th:spq-tree}
		Let $G$ an $n$-vertex weighted $2$-tree, let $e^*$ be an edge of $G$ with maximum length, let $\mathcal T$ be the SPQ-tree of $G$ such that the Q-node corresponding to $e^*$ is a child of the root of $\mathcal T$, and let $h^*$ be the height of $\mathcal T$. There exists an $\algoritmAny$-time algorithm that tests whether $G$ admits a planar straight-line realization and, in the positive case, constructs such a realization.
	\end{theorem}
	
	Due to~\cref{le:height-diameter},~\cref{th:spq-tree} implies the following.
	
	\begin{corollary}
		Let $G$ an $n$-vertex weighted $2$-tree and let $\ell$ be length of the longest path of $G$. There exists an $\algoritmPath$-time algorithm that tests whether $G$ admits a planar straight-line realization and, in the positive case, constructs such a realization.
	\end{corollary}
	
	In the remainder of the section, we prove~\cref{th:spq-tree}. Let $H$ be a $2$-tree and let $(u,v)$ be an edge of $H$. A \emph{$uv$-external} realization of $H$ is a planar straight-line realization of $H$ such that both $u$ and $v$ are incident to the outer face; note that the edge $(u,v)$ is not necessarily incident to the outer face of a $uv$-external realization of $H$. Given two $uv$-external realizations $\Gamma$ and $\Gamma'$ of $H$, we say that $\Gamma$ is \emph{$uv$-equivalent} to $\Gamma'$ if $\Gamma$ can be rotated and translated (note that reflections are not allowed) 
	so that the representation of $(u,v)$ coincides with the one in $\Gamma'$ and so that the boundary of the outer face of $\Gamma$ coincides with the boundary of the outer face of $\Gamma'$. Given a set $\mathcal R$ of $uv$-external realizations of $H$, we say that $\mathcal R$ is \emph{$uv$-complete} if, for every $uv$-external realization $\Gamma'$ of $H$, there exists a $uv$-external realization $\Gamma$ of $H$ in $\mathcal R$ such that $\Gamma$ is $uv$-equivalent to $\Gamma'$. Finally, a $uv$-complete set $\mathcal R$ of $uv$-external realizations of $H$ is \emph{minimal} if no realization in $\mathcal R$ is $uv$-equivalent to a distinct realization in $\mathcal R$.
	
	
	As a first step of our algorithm, we compute the SPQ-tree $\mathcal T$ of $G$ such that the Q-node corresponding to $e^*$ is a child of the root $\rho$ of $\mathcal T$; this can be done in $O(n)$ time~\cite{dt-opl-96,gm-lti-00,vtl-rspd-82}. One key point of our algorithm is that, for every P-node $\mu$ of $\mathcal T$, we only need to look for $u_{\mu}v_{\mu}$-external realizations of $G_{\mu}$ (recall that $u_{\mu}$ and $v_{\mu}$ are the poles of $\mu$). This is formalized in the following.
	
	\begin{lemma} \label{obs:e-external}
		Let $\Gamma$ be any planar straight-line realization of $G$. For any P-node $\mu$ of $\mathcal T$, the restriction $\Gamma_{\mu}$ of $\Gamma$ to $G_{\mu}$ is a $u_{\mu}v_{\mu}$-external realization of $G_{\mu}$.
	\end{lemma}
	
	\begin{proof}
		Consider first the case in which $\mu=\rho$. Then $\Gamma_{\mu}$ coincides with $\Gamma$. If $\Gamma$ were not a $u_{\mu}v_{\mu}$-external realization of $G_{\mu}$, then $e^*$ would lie inside a $3$-cycle $c$ of $G$ in $\Gamma$. However, this would contradict the planarity of $\Gamma$. Namely, since $e^*$ is an edge with maximum length, the length of $e^*$ would be larger than the one of each side of the triangle representing $c$ in $\Gamma$.
		
		Consider now the case in which $\mu\neq \rho$. Suppose, for a contradiction, that (at least) one of $u_{\mu}$ and $v_{\mu}$, say $u_{\mu}$, is not incident to the outer face of $\Gamma_{\mu}$. Hence, $u_{\mu}$ lies inside a cycle $c$ of $G_{\mu}$ in $\Gamma_{\mu}$ and in $\Gamma$. Furthermore, at most one of $u_{\rho}$ and $v_{\rho}$, say $u_{\rho}$, is a vertex of $c$ (if both $u_{\rho}$ and $v_{\rho}$ would belong to $G_{\mu}$, we would have $\mu=\rho$). Since $u_{\rho}$ and $v_{\rho}$ are incident to the outer face of $\Gamma$, we have that $v_{\rho}$ lies outside $c$ in $\Gamma$. Since $G$ is biconnected, it contains a path that connects $u_{\mu}$ with $v_{\rho}$ and that does not contain vertices of $G_{\mu}$ other than $u_{\mu}$. By the Jordan curve theorem, such a path crosses $c$ in $\Gamma$, a contradiction.
	\end{proof}
	
	%
	%
	
	For each P-node $\mu$ of $\mathcal T$, we construct an order $\nu_1,\dots,\nu_k$ of the children of $\mu$ as follows. First, we let $\nu_1$ be the only Q-node that is a child of $\mu$; this node exists by Condition~(C2). We denote by $\lambda(\nu_1)$ the length of the edge of $G$ corresponding to such a Q-node. Every other child $\nu_i$ of $\mu$ is an S-node. Let $\sigma_i$ and $\tau_i$ be the two children of $\nu_i$; by Condition~(C3), each of them is either a Q-node or a P-node. We denote by $w_{\nu_i}$ the pole that $\sigma_i$ and $\tau_i$ share (that is, their pole different from $u_{\mu}$ and $v_{\mu}$); also, we can assume, w.l.o.g.\ up to a relabeling of $\sigma_i$ with $\tau_i$, that $u_{\mu}$ is a pole of $\sigma_i$ and $v_{\mu}$ is a pole of $\tau_i$. Then the edges $(u_{\mu},w_{\nu_i})$ and $(v_{\mu},w_{\nu_i})$ belong to $G_{\nu_i}$, as each of $\sigma_i$ and $\tau_i$ is either a Q-node (and then it obviously contains the edge between its poles) or a P-node (and then it contains the edge between its poles by Condition~(C2)). Let $\lambda(\nu_i)$ denote the sum of the lengths of the edges $(u_{\mu},w_{\nu_i})$ and $(v_{\mu},w_{\nu_i})$. Then the order of the children of $\mu$ is established so that $\lambda(\nu_1)\leq \lambda(\nu_2)\leq \dots \leq \lambda(\nu_k)$. Clearly, such orders can be found for all the P-nodes of $\mathcal T$ in overall $O(n\log n)$ time. For a P-node $\mu$ of $\mathcal T$ with children $\nu_1,\dots,\nu_k$ and for any $i\in \{1,\dots,k\}$, we denote by $G^i_{\mu}$ the subgraph of $G_{\mu}$ which is the union of the graphs $G_{\nu_1},\dots,G_{\nu_i}$. Observe that $G^1_{\mu}$ is the edge $(u_{\mu},v_{\mu})$, while $G^k_{\mu}$ coincides with $G_{\mu}$. 
	
	Our algorithm either concludes that $G$ admits no planar straight-line realization or constructs the following sets:
	
	\begin{itemize}
		\item For each Q-node or P-node $\mu$, the algorithm constructs a non-empty minimal $u_{\mu}v_{\mu}$-complete set $\mathcal R_{\mu}$ of $u_{\mu}v_{\mu}$-external realizations of $G_{\mu}$. For each realization $\Gamma_{\mu}$ in $\mathcal R_{\mu}$, the algorithm also constructs the plane embedding $\mathcal E(\Gamma_{\mu})$. 
		\item For each P-node $\mu$ with children $\nu_1,\dots,\nu_k$ and for each $i=1,\dots,k$, the algorithm constructs a non-empty minimal $u_{\mu}v_{\mu}$-complete set $\mathcal R^i_{\mu}$ of $u_{\mu}v_{\mu}$-external realizations of $G^i_{\mu}$. For each realization $\Gamma^i_{\mu}$ in $\mathcal R^i_{\mu}$, the algorithm also constructs the plane embedding $\mathcal E(\Gamma^i_{\mu})$. 
	\end{itemize}
	
	In order to construct such sets, the algorithm bottom-up visits all the Q-nodes and all the P-nodes of $\mathcal T$.

	
	When a Q-node $\mu$ is visited, the set $\mathcal R_{\mu}$ is defined as $\{\Gamma_{\mu}\}$, where $\Gamma_{\mu}$ is the unique (up to rotation and translation) straight-line realization of the edge of $G$ corresponding to $\mu$. Obviously, $\mathcal R_{\mu}$ is a minimal $u_{\mu}v_{\mu}$-complete set of $u_{\mu}v_{\mu}$-external realizations of $G_{\mu}$.
	
	Suppose now that a P-node $\mu$ of $\mathcal T$ with children $\nu_1,\dots,\nu_k$ is considered. For $i=2,\dots,k$, since we already visited the children $\sigma_i$ and $\tau_i$ of $\nu_i$, we have the sets  $\mathcal R_{\sigma_i}$ and $\mathcal R_{\tau_i}$ (indeed, if we concluded that $G$ admits no planar straight-line realization, there is nothing else to do). Our algorithm processes the children $\nu_1,\dots,\nu_k$ of $\mu$ in this order. For $i=1,\dots,k$, after visiting $\nu_i$, the algorithm either concludes that $G$ admits no planar straight-line realization or constructs the set $\mathcal R^i_{\mu}$. In the latter case, for each realization $\Gamma^i_{\mu}$ in $\mathcal R^i_{\mu}$, the algorithm also constructs the plane embedding $\mathcal E(\Gamma^i_{\mu})$. If the algorithm constructs the set $\mathcal R^k_{\mu}$, then we set $\mathcal R_{\mu}=\mathcal R^k_{\mu}$. 
	
	We now describe how the algorithm processes the children $\nu_1,\dots,\nu_k$ of a P-node $\mu$. The algorithm starts by defining $\mathcal R^1_{\mu}=\mathcal R_{\nu_1}$ (recall that $\nu_1$ is a Q-node). Then, for $i=2,\dots,k$, it constructs the set $\mathcal R^i_{\mu}$ by combining the realizations of $G^{i-1}_{\mu}$ in $\mathcal R^{i-1}_{\mu}$ with the realizations of $G_{\sigma_i}$ in $\mathcal R_{\sigma_i}$ and with the realizations of $G_{\tau_i}$ in $\mathcal R_{\tau_i}$. This is done as follows.

	We initialize $\mathcal R^{i}_{\mu}=\emptyset$. Then, for every triple of realizations $\Gamma^{i-1}_{\mu} \in \mathcal R^{i-1}_{\mu}$, $\Gamma_{\sigma_i} \in \mathcal R_{\sigma_i}$, and $\Gamma_{\tau_i} \in \mathcal R_{\tau_i}$, we construct two straight-line realizations $\Phi^i_{\mu}$ and $\Psi^i_{\mu}$ of $G^i_{\mu}$ (and two plane embeddings $\mathcal E(\Phi^i_{\mu})$ and $\mathcal E(\Psi^i_{\mu})$ of $G^i_{\mu}$) as follows. Recall that the poles of $\sigma_i$ are $u_{\mu}$ and $w_{\nu_i}$, while the ones of $\tau_i$ are $v_{\mu}$ and $w_{\nu_i}$. Refer to~\cref{fig:spq-construction,fig:spq-embeddings}.
	
	\begin{figure}[tb!]
		\centering
		\includegraphics[scale=1]{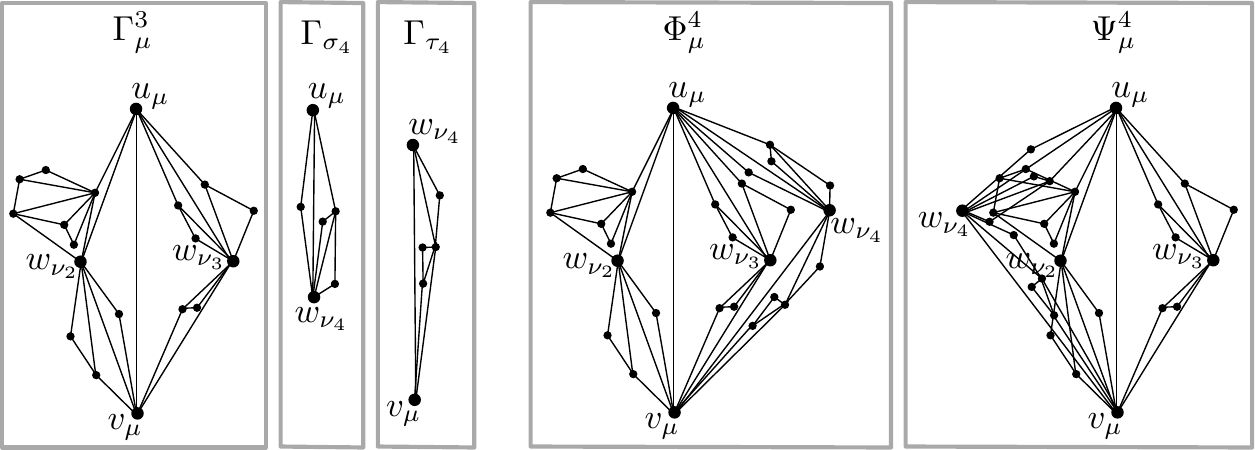}
		\caption{Construction of the realizations $\Phi^i_{\mu}$ and $\Psi^i_{\mu}$ of $G^i_{\mu}$ from the realizations $\Gamma^{i-1}_{\mu} \in \mathcal R^{i-1}_{\mu}$, $\Gamma_{\sigma_i} \in \mathcal R_{\sigma_i}$, and $\Gamma_{\tau_i} \in \mathcal R_{\tau_i}$. In this example, $i=4$; further, $\Phi^i_{\mu}$ is planar and $\Psi^i_{\mu}$ is non-planar.}
		\label{fig:spq-construction}
	\end{figure}
	
	\begin{figure}[tb!]
		\centering
		\includegraphics[scale=1]{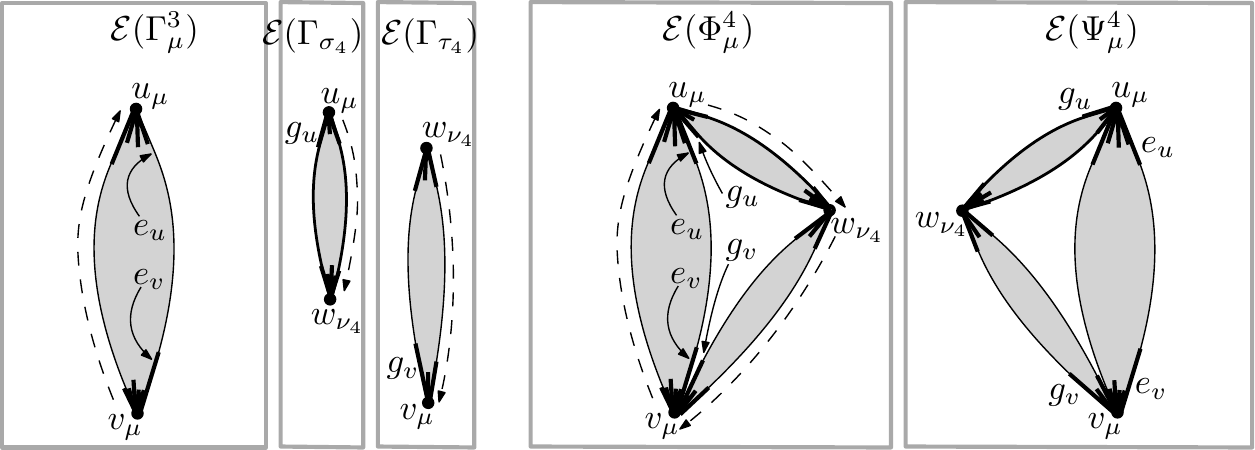}
		\caption{Construction of the plane embeddings $\mathcal E(\Phi^i_{\mu})$ and $\mathcal E(\Psi^i_{\mu})$ of $G^i_{\mu}$ from the plane embeddings $\mathcal E(\Gamma^{i-1}_{\mu})$, $\mathcal E(\Gamma_{\sigma_i})$, and $\mathcal E(\Gamma_{\tau_i})$. The illustration is coherent with the realizations of~\cref{fig:spq-construction}.}
		\label{fig:spq-embeddings}
	\end{figure}
	
	\begin{itemize}
		\item First, we initialize both $\Phi^i_{\mu}$ and $\Psi^i_{\mu}$ to $\Gamma^{i-1}_{\mu}$. 
		\item Second, we place the vertex $w_{\nu_i}$ in $\Phi^i_{\mu}$ and $\Psi^i_{\mu}$ with the two placements that realize the prescribed lengths of the edges $(u_{\mu},w_{\nu_i})$ and $(v_{\mu},w_{\nu_i})$; in particular, the placement of $w_{\nu_i}$ in $\Phi^i_{\mu}$ is such that the clockwise order of the vertices of the cycle $(u_{\mu},v_{\mu},w_{\nu_i})$ is $u_{\mu},w_{\nu_i},v_{\mu}$, while the placement of $w_{\nu_i}$ in $\Psi^i_{\mu}$ is such that the clockwise order of the vertices of the cycle $(u_{\mu},v_{\mu},w_{\nu_i})$ is $u_{\mu},v_{\mu},w_{\nu_i}$. 
		\item Third, we insert $\Gamma_{\sigma_i}$ in $\Phi^i_{\mu}$ so that the representation of the edge $(u_{\mu},w_{\nu_i})$ in $\Gamma_{\sigma_i}$ coincides with the one in $\Phi^i_{\mu}$. We also insert $\Gamma_{\tau_i}$ in $\Phi^i_{\mu}$ so that the representation of the edge $(v_{\mu},w_{\nu_i})$ in $\Gamma_{\tau_i}$ coincides with the one in $\Phi^i_{\mu}$. This concludes the construction of $\Phi^i_{\mu}$. The insertion of $\Gamma_{\sigma_i}$ and $\Gamma_{\tau_i}$ in $\Psi^i_{\mu}$ in order to conclude the construction of $\Psi^i_{\mu}$ is done in the same way.
		\item Finally, we construct a plane embedding $\mathcal E(\Phi^i_{\mu})$ as follows:
		
		\begin{itemize}
			\item We glue the lists of the edges incident to $u_{\mu}$ in $\mathcal E(\Gamma^{i-1}_{\mu})$ and $\mathcal E(\Gamma_{\sigma_i})$, so that, in the counter-clockwise order of the edges incident to $u_{\mu}$, the edge $e_{u}$ that enters $u_{\mu}$ when traversing in counter-clockwise order the boundary of the outer face of $\mathcal E(\Gamma^{i-1}_{\mu})$  comes immediately before the edge $g_{u}$ that enters $u_{\mu}$ when traversing in clockwise order the boundary of the outer face of $\mathcal E(\Gamma_{\sigma_i})$. 
			\item We glue the lists of the edges incident to $v_{\mu}$ in $\mathcal E(\Gamma^{i-1}_{\mu})$ and $\mathcal E(\Gamma_{\tau_i})$, so that, in the clockwise order of the edges incident to $v_{\mu}$, the edge $e_{v}$ that enters $v_{\mu}$ when traversing in clockwise order the boundary of the outer face of $\Gamma^{i-1}_{\mu}$ comes immediately before the edge $g_{v}$ that enters $v_{\mu}$ when traversing in counter-clockwise order the boundary of the outer face of $\Gamma_{\tau_i}$.
			\item We let the cycle bounding the outer face of $\mathcal E(\Phi^i_{\mu})$ be composed of the path that goes from $v_{\mu}$ to $u_{\mu}$ in clockwise direction along the outer face of $\mathcal E(\Gamma^{i-1}_{\mu})$, of the path that goes from $u_{\mu}$ to $w_{\nu_i}$ in clockwise direction along the outer face of $\mathcal E(\Gamma_{\sigma_i})$, and of the path that goes from $w_{\nu_i}$ to $v_{\mu}$ in clockwise direction along the outer face of $\mathcal E(\Gamma_{\tau_i})$.
		\end{itemize}
		We remark that $\mathcal E(\Phi^i_{\mu})$ is not necessarily the plane embedding of $\Phi^i_{\mu}$, in fact $\Phi^i_{\mu}$ might even be non-planar. 
		
		A plane embedding $\mathcal E(\Psi^i_{\mu})$ is constructed analogously. 
	\end{itemize}
	
	The straight-line realizations $\Phi^i_{\mu}$ and $\Psi^i_{\mu}$, together with the plane embeddings $\mathcal E(\Phi^i_{\mu})$ and $\mathcal E(\Psi^i_{\mu})$, are hence constructed in $O(n)$ time. 
	
	We remark that the plane embeddings $\mathcal E(\Phi^i_{\mu})$ and  $\mathcal E(\Psi^i_{\mu})$ both place $G_{\nu_i}$ ``externally'' with respect to $G^{i-1}_{\mu}$. This is a main ingredient for bounding the running time of our algorithm. It is also not a loss of generality, because of the ordering that has been established for the children of $\mu$, as proved in the following lemma.
	
	\begin{lemma} \label{le:external-components}
		Consider a $u_{\mu}v_{\mu}$-external realization $\Gamma^i_{\mu}$ of $G^i_{\mu}$. Then the restriction of $\Gamma^i_{\mu}$ to $G_{\nu_i}$ lies in the outer face of the restriction of $\Gamma^i_{\mu}$ to $G^{i-1}_{\mu}$. 
	\end{lemma} 
	\begin{proof}
		Suppose, for a contradiction, that the restriction of $\Gamma^i_{\mu}$ to $G_{\nu_i}$ lies inside an internal face $f$ of the restriction of $\Gamma^i_{\mu}$ to $G^{i-1}_{\mu}$. In particular, the path $(u_{\mu},w_{\nu_i},v_{\mu})$ lies inside $f$ in $\Gamma^i_{\mu}$. Since $\Gamma^i_{\mu}$ is a $u_{\mu}v_{\mu}$-external realization of $G^i_{\mu}$, both $u_{\mu}$ and $v_{\mu}$ are incident to the outer face of $\Gamma^i_{\mu}$, hence it is possible to draw a Jordan arc $\mathcal C_{\mu}$ connecting $u_{\mu}$ and $v_{\mu}$ and lying in the outer face of $\Gamma^i_{\mu}$, except at its end-points. Now $\mathcal C_{\mu}$ and the path $(u_{\mu},w_{\nu_i},v_{\mu})$ form a closed curve which contains the edge $(u_{\mu},v_{\mu})$ on one side and a graph $G_{\nu_j}$ on the other side (as otherwise $f$ would be the outer face of $\Gamma^i_{\mu}$), for some $2\leq j<i$. However, this implies that the path $(u_{\mu},w_{\nu_i},v_{\mu})$ lies inside the triangle representing the $3$-cycle $(u_{\mu},w_{\nu_j},v_{\mu})$ in $\Gamma^i_{\mu}$. This contradicts the planarity of $\Gamma^i_{\mu}$; indeed, we have $\ell(\nu_i)>\ell(\nu_j)>\ell(\nu_1)$, hence the length of the path $(u_{\mu},w_{\nu_i},v_{\mu})$ is larger than the lengths of the path $(u_{\mu},w_{\nu_j},v_{\mu})$ and of the edge $(u_{\mu},v_{\mu})$ in $\Gamma^i_{\mu}$. 
	\end{proof}
	
	%
	%
	
	We now resume the description of how the algorithm constructs $\mathcal R^i_{\mu}$. By means of~\cref{thm:straight-line_realization_planarity}, we test in $O(n)$ time whether $\Phi^i_{\mu}$ is a planar straight-line realization respecting $\mathcal E(\Phi^i_{\mu})$. By the construction of $\mathcal E(\Phi^i_{\mu})$, the test is positive if and only if $\Phi^i_{\mu}$ is a $u_{\mu}v_{\mu}$-external realization of $G^i_{\mu}$. If the test is negative, then we discard $\Phi^i_{\mu}$. 
	Otherwise, for every realization $\Lambda^i_{\mu}$ in $\mathcal R^{i}_{\mu}$, we test whether $\Phi^i_{\mu}$ is $u_{\mu}v_{\mu}$-equivalent to $\Lambda^i_{\mu}$ and, if it is, we discard $\Phi^i_{\mu}$. This can be done in $O(n)$ time by rotating and translating $\Phi^i_{\mu}$ so that the representation of the edge $(u_{\mu},v_{\mu})$ coincides with the one in $\Lambda^i_{\mu}$ and by then checking whether the vertices along on the boundary of the outer face of $\Phi^i_{\mu}$ have the same coordinates as the vertices along on the boundary of the outer face of $\Lambda^i_{\mu}$. If we did not discard $\Phi^i_{\mu}$, then we insert it into $\mathcal R^{i}_{\mu}$. 
	
	We process $\Psi^i_{\mu}$ analogously to $\Phi^i_{\mu}$.

	After we have processed every triple of drawings $\Gamma^{i-1}_{\mu} \in \mathcal R^{i-1}_{\mu}$, $\Gamma_{\sigma_i} \in \mathcal R_{\sigma_i}$, and $\Gamma_{\tau_i} \in \mathcal R_{\tau_i}$, if the set $\mathcal R^{i}_{\mu}$ is empty, then we report that $G$ admits no planar straight-line realization, otherwise $\mathcal R^{i}_{\mu}$ is a minimal $u_{\mu}v_{\mu}$-complete set of $u_{\mu}v_{\mu}$-external realizations of $G^i_{\mu}$. The correctness of these conclusions is proved in the following lemma.
	
	\begin{lemma} \label{le:length-correctness}
		If $\mathcal R^{i}_{\mu}$ is empty, then $G$ admits no planar straight-line realization. Otherwise, $\mathcal R^{i}_{\mu}$ is a minimal $u_{\mu}v_{\mu}$-complete set of $u_{\mu}v_{\mu}$-external realizations of $G^i_{\mu}$. 
	\end{lemma}
	
	\begin{proof}
		The main ingredient of the proof is the following claim: If $G^i_{\mu}$ admits a $u_{\mu}v_{\mu}$-external realization $\Lambda^i_{\mu}$, then $\mathcal R^{i}_{\mu}$ contains a $u_{\mu}v_{\mu}$-external realization that is $u_{\mu}v_{\mu}$-equivalent to $\Lambda^i_{\mu}$. 
		
		We prove the claim. For a contradiction, suppose that $G^i_{\mu}$ admits a $u_{\mu}v_{\mu}$-external realization $\Lambda^i_{\mu}$ and that $\mathcal R^{i}_{\mu}$ does not contain any $u_{\mu}v_{\mu}$-external realization that is $u_{\mu}v_{\mu}$-equivalent to $\Lambda^i_{\mu}$. We are going to replace parts of $\Lambda^i_{\mu}$ with realizations in $\mathcal R^{i-1}_{\mu}$, $\mathcal R_{\sigma_i}$, and $\mathcal R_{\tau_i}$, so to obtain a $u_{\mu}v_{\mu}$-external realization of $G^i_{\mu}$ that is constructed by the algorithm and that should hence have been inserted in $\mathcal R^{i}_{\mu}$ (unless a $u_{\mu}v_{\mu}$-equivalent $u_{\mu}v_{\mu}$-external realization of $G^i_{\mu}$ is already in $\mathcal R^{i}_{\mu}$), from which a contradiction follows. 
		
		Let $\Lambda^{i-1}_{\mu}$,  $\Lambda_{\sigma_i}$, and $\Lambda_{\tau_i}$ be the restrictions of $\Lambda^i_{\mu}$ to $G^{i-1}_{\mu}$,  $G_{\sigma_i}$, and $G_{\tau_i}$, respectively. By~\cref{le:external-components}, we have that $\Lambda_{\sigma_i}$ and $\Lambda_{\tau_i}$ lie in the outer face of $\Lambda^{i-1}_{\mu}$; furthermore, the fact that $\Lambda^i_{\mu}$ is a $u_{\mu}v_{\mu}$-external realization directly implies that $\Lambda_{\sigma_i}$ and $\Lambda^{i-1}_{\mu}$ lie in the outer face of $\Lambda_{\tau_i}$ and that $\Lambda_{\tau_i}$ and $\Lambda^{i-1}_{\mu}$ lie in the outer face of $\Lambda_{\sigma_i}$.
		
		Since $\mathcal R^{i-1}_{\mu}$ is a $u_{\mu}v_{\mu}$-complete set of $u_{\mu}v_{\mu}$-external realizations of $G^{i-1}_{\mu}$, it follows that there exists a $u_{\mu}v_{\mu}$-external realization $\Gamma^{i-1}_{\mu}\in \mathcal R^{i-1}_{\mu}$ of $G^{i-1}_{\mu}$ that is $u_{\mu}v_{\mu}$-equivalent to $\Lambda^{i-1}_{\mu}$. Analogously, there exists a $u_{\mu}w_{\nu_i}$-external realization $\Gamma_{\sigma_i}\in \mathcal R_{\sigma_i}$ of $G_{\sigma_i}$ that is $u_{\mu}w_{\nu_i}$-equivalent to $\Lambda_{\sigma_i}$ and there exists a $v_{\mu}w_{\nu_i}$-external realization $\Gamma_{\tau_i}\in \mathcal R_{\tau_i}$ of $G_{\tau_i}$ that is $v_{\mu}w_{\nu_i}$-equivalent to $\Lambda_{\tau_i}$. 
		
		We can now replace $\Lambda^{i-1}_{\mu}$ with $\Gamma^{i-1}_{\mu}$ in $\Lambda^i_{\mu}$, after rotating and translating $\Gamma^{i-1}_{\mu}$ so that the representation of the edge $(u_{\mu},v_{\mu})$ coincides with the one in $\Lambda^i_{\mu}$. Since the boundary of the outer face of $\Gamma^{i-1}_{\mu}$ coincides with the boundary of the outer face of $\Lambda^{i-1}_{\mu}$ and since $\Lambda_{\sigma_i}$ and $\Lambda_{\tau_i}$ lie in the outer face of $\Lambda^{i-1}_{\mu}$, it follows that after the replacement $\Lambda^i_{\mu}$ remains a $u_{\mu}v_{\mu}$-external realization of $G^i_{\mu}$ (in particular, it remains planar) that is $u_{\mu}v_{\mu}$-equivalent to $\Lambda^i_{\mu}$. We can analogously replace $\Lambda_{\sigma_i}$ with $\Gamma_{\sigma_i}$ and $\Lambda_{\tau_i}$ with $\Gamma_{\tau_i}$, obtaining a $u_{\mu}v_{\mu}$-external realization of $G_{\mu}$ that is constructed by the algorithm and that is $u_{\mu}v_{\mu}$-equivalent to $\Lambda^i_{\mu}$. This implies that $\mathcal R^{i}_{\mu}$ contains a $u_{\mu}v_{\mu}$-external realization that is $u_{\mu}v_{\mu}$-equivalent to $\Lambda^i_{\mu}$, a contradiction which proves the claim.
		
		The claim directly implies that, if $\mathcal R^{i}_{\mu}$ is empty, then $G^i_{\mu}$ admits no $u_{\mu}v_{\mu}$-external realization, hence $G_{\mu}$ admits no $u_{\mu}v_{\mu}$-external realization (given that $G^i_{\mu}$ is a subgraph of $G_{\mu}$), and finally that $G$ admits no planar straight-line realization, by~\cref{obs:e-external}. 
		
		Suppose now that $\mathcal R^{i}_{\mu}$ is non-empty. We prove that $\mathcal R^{i}_{\mu}$ is a minimal $u_{\mu}v_{\mu}$-complete set of $u_{\mu}v_{\mu}$-external realizations of $G^i_{\mu}$. First, that every realization $\Gamma^i_{\mu}$ in $\mathcal R^{i}_{\mu}$ is $u_{\mu}v_{\mu}$-external directly descends from the $O(n)$-time test that is performed by means of~\cref{thm:straight-line_realization_planarity} on  $\Gamma^i_{\mu}$; indeed, $\Gamma^i_{\mu}$ is inserted into $\mathcal R^{i}_{\mu}$ only if the test reports that $\Gamma^i_{\mu}$ is a planar straight-line realization whose plane embedding is the one constructed by the algorithm, in which $u_{\mu}$ and $v_{\mu}$ are incident to the outer face. Second, that $\mathcal R^{i}_{\mu}$ is $u_{\mu}v_{\mu}$-complete is a direct consequence of the claim above. Finally, the minimality of $\mathcal R^{i}_{\mu}$ comes from the fact that a realization is inserted into $\mathcal R^{i}_{\mu}$ only if it is not $u_{\mu}v_{\mu}$-equivalent to any realization already in $\mathcal R^{i}_{\mu}$.
	\end{proof}
	
	It remains to determine the running time of our algorithm. The main ingredient we are going to use is the following combinatorial lemma.
	
	\newcommand{\sizeSet}{n^{\frac{2^{h}-1}{3}}\cdot 2^{\frac{2^{h}-1}{3}}}
	\begin{lemma} \label{le:how-many-optimal}
		Let $\mu$ be a Q-node or a P-node of $\mathcal T$, and let $h$ be the height of ${\mathcal T}_{\mu}$. We have that $|\mathcal R_{\mu}|\leq \sizeSet$. Furthermore, if $\mu$ is a P-node, let $\nu_1,\dots,\nu_k$ be the children of $\mu$. For $i=1,\dots,k$, we have that $|\mathcal R^i_{\mu}|\leq \sizeSet$.
	\end{lemma}	
	
	\begin{proof}
		We prove the following claim. Let $f(k)$ be the function whose domain consists of the even natural numbers, that is recursively defined as follows: $f(k)=1$ if $k=0$, and $f(k)= 2n \cdot (f(k-2))^4$ if $k>0$. Then we claim that $|\mathcal R_{\mu}|\leq f(h)$ and $|\mathcal R^i_{\mu}|\leq f(h)$. The claim implies the lemma, as an easy inductive proof shows that the closed form of $f(k)$ is $f(k)=n^{\frac{2^{k}-1}{3}}\cdot 2^{\frac{2^{k}-1}{3}}$.
		
		The claim is proved by induction on $h$. If $h=0$, then $\mu$ is a Q-node, hence by construction we have that $|\mathcal R_{\mu}|=1=f(0)$. 
		
		Suppose now that $h>0$, hence $\mu$ is a P-node with children $\nu_1,\dots,\nu_k$. First, we have $\mathcal R^1_{\mu}=1<f(h)$. We now prove that $|\mathcal R^i_{\mu}|\leq \sizeSet$, for any $i\in \{2,\dots,k\}$; this suffices to complete the induction, since by construction $\mathcal R_{\mu}=\mathcal R^k_{\mu}$. Consider any realization $\Gamma^i_{\mu}$ in $\mathcal R^i_{\mu}$. Since $\Gamma^i_{\mu}$ is a $u_{\mu}v_{\mu}$-external realization of $G^i_{\mu}$ and since $i\geq 2$, it follows that there are exactly two graphs among $G_{\nu_1},\dots,G_{\nu_{i}}$ that contain edges incident to the outer face of $\Gamma^i_{\mu}$. By~\cref{le:external-components}, one of these graphs is $G_{\nu_i}$, while the other graph, say $G_{\nu_j}$, might be any graph among $G_{\nu_1},\dots,G_{\nu_{i-1}}$. When traversing the boundary of the outer face of $\Gamma^i_{\mu}$ in clockwise direction, the edge $e_{\nu_i}$ of $G_{\nu_i}$ incident to $u_{\mu}$ might appear right before or right after $u_{\mu}$. Altogether, there are less than $2n$ possibilities for the choice of the graph $G_{\nu_j}$ and for whether $e_{\nu_i}$ appears right before or right after $u_{\mu}$ in clockwise order along the boundary of the outer face of $\Gamma^i_{\mu}$. For each of these $2n$ possibilities, the interior of $\Gamma^i_{\mu}$ is determined by the choice of four realizations $\Gamma_{\sigma_i} \in \mathcal R_{\sigma_i}$, $\Gamma_{\tau_i} \in \mathcal R_{\tau_i}$, $\Gamma_{\sigma_j} \in \mathcal R_{\sigma_j}$, and $\Gamma_{\tau_j} \in \mathcal R_{\tau_j}$. Indeed, the restriction of $\Gamma^i_{\mu}$ to $G_{\sigma_i}$, $G_{\tau_i}$, $G_{\sigma_j}$, and $G_{\tau_j}$ is a realization in $\mathcal R_{\sigma_i}$, $\mathcal R_{\tau_i}$, $\mathcal R_{\sigma_j}$, and $\mathcal R_{\tau_j}$, respectively. Since the heights of $\mathcal T_{\sigma_i}$, $\mathcal T_{\tau_i}$, $\mathcal T_{\sigma_j}$, and $\mathcal T_{\tau_j}$ are each at most $h-2$, by induction each of $\mathcal R_{\sigma_i}$, $\mathcal R_{\tau_i}$, $\mathcal R_{\sigma_j}$, and $\mathcal R_{\tau_j}$ has size smaller than or equal to $f(h-2)$. Hence, the number of possible choices for the realizations $\Gamma_{\sigma_i} \in \mathcal R_{\sigma_i}$, $\Gamma_{\tau_i} \in \mathcal R_{\tau_i}$, $\Gamma_{\sigma_j} \in \mathcal R_{\sigma_j}$, and $\Gamma_{\tau_j} \in \mathcal R_{\tau_j}$ is smaller than or equal to $(f(h-2))^4$. Finally, the minimality of $\mathcal R^i_{\mu}$ implies that no two realizations in $\mathcal R^i_{\mu}$ coincide when restricted to $G_{\sigma_i}$, $G_{\tau_i}$, $G_{\sigma_j}$, and $G_{\tau_j}$. This proves that the number of realizations in $\mathcal R^i_{\mu}$ is smaller than or equal to $2n \cdot (f(h-2))^4$; this concludes the induction and hence the proof of the lemma. 
	\end{proof}
	
	We are now ready to determine the running time of our algorithm. Whenever a P-node $\mu$ with $k$ children $\nu_1,\dots,\nu_k$ is visited in the bottom-up traversal of $\mathcal T$, for $i=1,\dots,k$, the set $\mathcal R^i_{\mu}$ is constructed by considering every triple of drawings $\Gamma^{i-1}_{\mu} \in \mathcal R^{i-1}_{\mu}$, $\Gamma_{\sigma_i} \in \mathcal R_{\sigma_i}$, and $\Gamma_{\tau_i} \in \mathcal R_{\tau_i}$, where $\sigma_i$ and $\tau_i$ are the children of $\nu_i$. By~\cref{le:how-many-optimal}, the number of such triples is at most $\big(n^{\frac{2^{h^*}-1}{3}}\cdot 2^{\frac{2^{h^*}-1}{3}}\big)\big(n^{\frac{2^{(h^*-2)}-1}{3}}\cdot 2^{\frac{2^{(h^*-2)}-1}{3}}\big)^2=n^{(2^{(h^*-1)}-1)}\cdot 2^{(2^{(h^*-1)}-1)}$; indeed, the height of $\mathcal T$ is $h^*$, hence the height of $\mathcal T_\mu$ is smaller than or equal to $h^*$, and the heights of $\mathcal T_{\sigma_i}$ and $\mathcal T_{\tau_i}$ are smaller than or equal to $h^*-2$. For each of such triples, two straight-line realizations $\Phi^i_{\mu}$ and $\Psi^i_{\mu}$ are constructed in $O(n)$ time, together with the plane embeddings $\mathcal E(\Phi^i_{\mu})$ and $\mathcal E(\Psi^i_{\mu})$. By means of~\cref{thm:straight-line_realization_planarity}, it is tested in $O(n)$ time whether $\Phi^i_{\mu}$ and $\Psi^i_{\mu}$ are planar straight-line realizations with plane embeddings $\mathcal E(\Phi^i_{\mu})$ and $\mathcal E(\Psi^i_{\mu})$, respectively. For each realization $\Lambda^i_{\mu}$ of $G^i_{\mu}$ in $\mathcal R^i_{\mu}$, it is tested in $O(n)$ time whether $\Phi^i_{\mu}$ is $u_{\mu}v_{\mu}$-equivalent to $\Lambda^i_{\mu}$. By~\cref{le:how-many-optimal}, the set $\mathcal R^i_{\mu}$ contains at most $\big(n^{\frac{2^{h^*}-1}{3}}\cdot 2^{\frac{2^{h^*}-1}{3}}\big)$ realizations of $G^i_{\mu}$, hence these $O(n)$-time tests are performed $\big(n^{(2^{(h^*-1)}-1)}\cdot 2^{(2^{(h^*-1)}-1)}\big)\big(n^{\frac{2^{h^*}-1}{3}}\cdot 2^{\frac{2^{h^*}-1}{3}}\big)=\big(n^{\frac{5\cdot 2^{h^*-1}-4}{3}}\cdot 2^{\frac{5\cdot 2^{h^*-1}-4}{3}}\big)$ times, with a $\big(n^{\frac{5\cdot 2^{h^*-1}-1}{3}}\cdot 2^{\frac{5\cdot 2^{h^*-1}-4}{3}}\big)$ running time. Since the number of P-nodes in $\mathcal T$ is in $O(n)$, the running time of our algorithm is in $\big(n^{\frac{5\cdot 2^{h^*-1}+2}{3}}\cdot 2^{\frac{5\cdot 2^{h^*-1}-4}{3}}\big)\in n^{O(2^{h^*})}$. This concludes the proof of~\cref{th:spq-tree}. 
	
	
	\section{Conclusions and Open Problems}
	\label{sec:conclusions}
	
	We studied the problem of testing if a $2$-tree has a planar straight-line drawing with prescribed edge lengths, both when the number of distinct edge lengths is bounded and when it is unbounded.
	
	In the bounded setting, we proved that the problem is linear-time solvable if the number of prescribed distinct lengths is at most two, while it is \NP-hard if this number is four or more. Hence, an intriguing question is to determine the complexity of the problem for weighted graphs with three prescribed distinct lengths.
	
	In the unbounded setting, we showed that the problem is linear-time solvable if the embedding of the input graph is prescribed. Moreover, we studied weighted maximal outerplanar graphs, and we proved that the problem is linear-time solvable if their dual tree is a path, and cubic-time solvable if their dual tree is a caterpillar. 
	We leave open the problem about general maximal outerplanar graphs.
	
	Finally, for the general problem on weighted $2$-trees, we described a slice-wise polynomial-time algorithm, parameterized by the length of the longest path. A natural open problem is to determine whether the problem admits an FPT algorithm with respect to the same parameter.

	\bibliographystyle{splncs04}
	\bibliography{bibliography}

\end{document}